\RequirePackage{fix-cm}
\documentclass[smallcondensed]{svjour3}     
\smartqed  
\usepackage{graphicx}
\usepackage[authoryear,round,sort&compress]{natbib}
\setcitestyle{aysep={}}
\usepackage{amsmath,amssymb,amsfonts} 
\usepackage{csquotes}
\usepackage{xcolor}
\usepackage[labelsep=space]{caption}
\usepackage[labelsep=space]{subcaption}
\DeclareCaptionFont{boldcap}{\bfseries}
\usepackage{numprint}
\usepackage{booktabs}
\usepackage{siunitx}
\usepackage[ruled]{algorithm2e}
\usepackage{silence}
\usepackage{paralist}
\usepackage{doi}

\usepackage{hyperref}
\graphicspath{{figures/}}

\newcommand{\graph}{\mathcal{G}}
\newcommand{\kavg}{\overline{k}} 
\newcommand{\ent}[1]{H(#1)}
\newcommand{\idxset}[1]{\dom{K}^\dom{#1}}
\newcommand{\E}[1]{\mathrm{E}\{#1\}}
\newcommand{\ami}[2]{I^{adj}(#1, #2)}
\newcommand{\mutinf}[1]{I(#1)}
\newcommand{\kld}[2]{D(#1\Vert #2)}
\newcommand{\kldr}[2]{\overline{D}(#1\Vert #2)}

\newcommand{\objective}[1]{\dom{J}(#1)}

\renewcommand{\binom}[2]{\left ( \genfrac{}{}{0pt}{}{#1}{#2} \right )}
\renewcommand{\Pr}[1]{\mathbb{P}(#1)}

\newcommand{\dom}[1]{\mathcal{#1}}

\newcommand{\indicator}[2]{\mathbb{I}_{#1}(#2)}

\newcommand{\clutrue}{\dom{Y}^{\mathrm{true}}}
\newcommand{\clupred}{\dom{Y}^{\mathrm{pred}}}
\renewcommand{\clupred}{\dom{Y}} 

\DeclareMathOperator*{\argmax}{arg\,max}
\DeclareMathOperator*{\argmin}{arg\,min}

\journalname{Data Mining and Knowledge Discovery}
\begin{document}
	
	\title{Synwalk - Community Detection via Random Walk Modelling
	}
	
	
	\author{Christian Toth  \and
		Denis Helic     \and
		Bernhard C. Geiger
	}
	
	
	\institute{C. Toth \at
		Signal Processing and Speech Communication Laboratory,
		Graz University of Technology,
		Inffeldgasse 16c, 8010 Graz, Austria \\
		Tel.: +43 316 873 - 4365\\
		\email{christian.toth@tugraz.at} \\
		ORCID: 0000-0002-0552-2874
		\and
		D. Helic \at
		Institute of Interactive Systems and Data Science,
		Graz University of Technology,
		Inffeldgasse 16c/I, 8010 Graz, Austria \\
		Tel.: +43 316 873 - 30610\\
		\email{dhelic@tugraz.at} \\
		ORCID: 0000-0003-0725-7450
		\and
		B. C. Geiger \at
		Know-Center GmbH,
		Inffeldgasse 13/6, 8010 Graz, Austria \\
		Tel.: +43 316 873 - 30881\\
		\email{geiger@ieee.org} \\
		ORCID: 0000-0003-3257-743X
	}
	
	\date{}
	
	\global\let\makeheadbox\relax
	\maketitle
	\vspace{-3em}
	\begin{abstract}
  Complex systems, abstractly represented as networks, are ubiquitous in everyday life. Analyzing and understanding these systems requires, among others, tools for community detection. As no single best community detection algorithm can exist, robustness across a wide variety of problem settings is desirable. 
  In this work, we present Synwalk, a random walk-based community detection method. Synwalk builds upon a solid theoretical basis and detects communities by synthesizing the random walk induced by the given network from a class of candidate random walks.
  We thoroughly validate the effectiveness of our approach on synthetic and empirical networks, respectively, and compare Synwalk's performance with the performance of Infomap and Walktrap. Our results indicate that Synwalk performs robustly on networks with varying mixing parameters and degree distributions. We outperform Infomap on networks with high mixing parameter, and Infomap and Walktrap on networks with many small communities and low average degree. Our work has a potential to inspire further development of community detection via synthesis of random walks and we provide concrete ideas for future research.
  
\keywords{Community detection \and Clustering \and Random walk modelling}
\end{abstract}
	
	\section{Introduction}\label{sec:introduction}

Large-scale systems of various kinds including social, informational or biological systems are pervasive in human life. Prominent examples of such systems are online social networks, the Internet, power grids or neural networks in the brain. Naturally, there is a strong interest in analyzing and understanding these systems, e.g., to estimate the effects of interventions on parts of the system, to alter or preserve their functionality, or to predict their future evolution. Commonly, we abstract such systems as networks, where nodes represent the entities in the system, and links between nodes represent the (form of) interaction between the entities.

We identify three basic necessities for a feasible empirical research of large networks:
\begin{inparaenum}[(i)]
\item identification of functional groups, frequently called \emph{communities}~\citep{Girvan2002,Radicchi2004a},
\item models for the interaction within and between these functional groups,
\item visualization of large networks and their dynamics in human graspable form.
\end{inparaenum}
Our research field established community detection~\citep{Girvan2002,Fortunato2010,Fortunato2016} as a tool satisfying these three basic requirements, as the identification of functional groups in networks is often facilitated by implicitly or explicitly assuming specific interaction models and further allows to visualize the network at a more granular level. Consequently, researchers proposed numerous community detection algorithms in recent years~\citep{Girvan2002,Clauset2004,Rosvall2008,Rosvall2009,Pons2005,Blondel2008,Raghavan2007,Reichardt2006}.

This plethora of algorithms leaves us frequently wondering which is the \enquote{best} community detection algorithm for a given practical application? Typically, researchers compare algorithms based on their ability to identify ground truth communities in artificially generated benchmark networks~\citep{Orman2009,Yang2016} or ground truth extracted from node-metadata in empirical networks. However, \citet{Peel2017} recently showed that such an evaluation is more delicate: they provide a No Free Lunch theorem, stating that there can be no single best algorithm for all possible detection scenarios, and furthermore, community evaluation on empirical networks based on node meta-data is contestable in a general case. Hence, as different algorithms (potentially) uncover different structural aspects when applied to a given network, the choice of method depends, among other criteria, on the type of community we are looking for.

One prominent class of community detection methods characterizes communities based on random walks on a network, with Infomap~\citep{Rosvall2008,Rosvall2009} and Walktrap~\citep{Pons2005} being two popular representatives (see Section~\ref{sec:related_work} for a compact description). Along the lines of the No Free Lunch theorem, both methods have their strengths and their weaknesses. Whereas Infomap accurately uncovers communities that are strongly connected internally (as characterized by the mixing parameter; see Section~\ref{ssec:clusterings}), it fails to do so for loosely connected communities (cf. ~\citep{Yang2016}, Section~\ref{sec:lfr_experiments}). On the other hand, while Walktrap delivers reasonable results over a broader spectrum of the strength of the community structures, we find that its performance strongly depends on the degree distribution of the network. In addition, Walktrap requires a selection of a hyper-parameter that is commonly chosen empirically. As we typically do not know about the community structure of a network a priori, it is unclear for practitioners how to select among these two random walk methods.

This raises an interesting question: can we combine the strengths of both methods to arrive at a community detection method that is more robust across a wider range of problem settings? In this paper, we tackle this question by presenting Synwalk -- a community detection method where we model community properties by designing a synthetic random walk model.
Specifically, Synwalk assumes a class of random walks with independent and identically distributed (i.i.d.) movements within and between candidate communities. It then simultaneously optimizes the distribution parameters of these i.i.d.\ movements (closed-form solution) and the candidate community structure (combinatorial optimization) such that the thus synthesized random walk resembles the random walk induced by the network under consideration.
Due to the structure of the i.i.d.\ movements and the aim to synthesize an existing random walk, Synwalk thus shares ideas from both Infomap and stochastic block modelling (Section~\ref{sec:rationale}). We discuss the properties of the resulting Synwalk objective in Section~\ref{sec:derivation} and thoroughly investigate and compare the behavior of Synwalk to Infomap and Walktrap on generated benchmark graphs in Section~\ref{sec:lfr_experiments}. Furthermore, we illustrate the applicability of our method on empirical undirected networks with non-overlapping communities (Section~\ref{sec:realworld}). 

Our work presents a novel
instance of community detection via random walk modelling, which adapts the concept of (stochastic) block modelling to random walk-based community detection. At the same time, Synwalk combines the strengths of the popular random walk-based community detection algorithms Infomap and Walktrap, achieving more robust results across a range of generated and empirical networks without the need for hyper-parameter optimization. We believe that our method and results can initiate future theoretical and practical work to fully unlock the potential of synthetic random walk models for community detection by, e.g., 
\begin{inparaenum}[(i)]
\item designing objective functions that enable robust detection of communities on specific classes of networks, or 
\item designing random walk models tailored for detecting specific types of communities.
\end{inparaenum} 

	\section{Related Work}\label{sec:related_work}

Different approaches to community detection have been inspired by different definitions of communities (see \citep{Fortunato2016} for an excellent survey). Accordingly, \citet{Rosvall2019} argue that different approaches to community detection can be categorized into four big groups, i.e., cut-based community detection, clustering, stochastic block modeling, and community detection based on network flows or random walks. We will now briefly summarize the concepts and approaches relevant for this work.

In a classical view, communities are densely connected subnetworks of a network that are well separated, which resonates with cut- or clustering-based community detection. This view takes the internal and external node degrees w.r.t.\ an assumed community structure into account. Popular metrics for measuring the existence and strength of a community structure are the mixing parameter~\citep{Lancichinetti2008, Lancichinetti2009b} and the modularity~\citep{Newman2004, Newman2006}. The mixing parameter $\mu$ of a node is the ratio between the number of links to nodes outside of its community and the total number of its links. The related quantity modularity compares the density of links within communities to links between communities and is used as an objective function for community detection~\citep{Clauset2004,Blondel2008}.

While we will use the mixing parameter and modularity in setting up and evaluating our experiments in Sections~\ref{sec:lfr_experiments} and~\ref{sec:realworld}, the method we propose falls into the category of random walk-based community detection methods. Random walks provide a simple proxy for diffusion processes describing the dynamics of a network. Here, the notion of a community is related to the average time a random walker spends within a certain subgroup of nodes of a network. Two prominent examples for random walk-based community detection methods are Infomap~\citep{Rosvall2008, Rosvall2009} and Walktrap~\citep{Pons2005}. 

Assuming a certain clustering (i.e., a candidate community structure), Infomap encodes the movements of a random walker on a network with a two-level codebook scheme. Each cluster has its own codebook with codewords for each member node, plus a dedicated exit codeword. Additionally, there is a global index codebook with codewords for each cluster. Now, for every move of the random walker, Infomap records the codeword of the next node from the codebook of its containing cluster. Moreover, whenever the random walker changes clusters, Infomap records the exit codeword of the old cluster's codebook and the codeword of the new cluster from the index codebook before it records the new node. By minimizing the average description length of realizations of such a random walk, Infomap obtains a clustering that compactly describes the network dynamics and hence should fit the true community structure well. Notably, it is not necessary to actually simulate random walks, as the movements of the random walker are characterized by the network topology, allowing to compute the average description length via the map equation~\citep{Rosvall2008, Rosvall2009}.

Walktrap formulates a random walk-based distance measure between clusters. Given a fixed number of steps, a random walker starting at a certain node will visit a neighboring node with a given probability. These probabilities hold information about how well two nodes are connected. Now, assuming two nodes are within the same community (i.e., well-connected), their probabilities to reach any other node within the network for a given number of steps should be similar. This observation yields a distance measure based on a weighted mean squared difference of such probabilities. Walktrap greedily merges nodes/clusters based on the described distance to arrive at a suitable clustering~\citep{Pons2005}.

Finally, the authors of~\citep{Hurley2015,Hurley2016} proposed a method for community detection that combines a random walk-based approach with (classic) block modelling. More specifically, the authors model aim to find a candidate clustering of the network under investigation such that the random walk induced on this clustering is similar to an arbitrarily chosen target random walk, where similarity is measured by the Kullback-Leibler divergence.

While the method proposed in~\citep{Hurley2015,Hurley2016} is based on random walks on clusters, Infomap, Walktrap, and Synwalk are based on the random walks induced by the network under investigation. Further, while Infomap and Walktrap predict clusterings by analyzing these random walks, our method predicts clusterings by synthesizing the network-induced random walk from a restricted class of candidate random walks. Searching for a proper random walk within this class makes our method robust across different network types. Additionally, being able to design the candidate class opens up possibilities for exploring alternative designs in future research.
	
\section{Preliminaries} \label{sec:preliminaries}
\subsection{Networks and Clusterings}\label{ssec:clusterings}
Let $\graph:=(\dom{X},E,W)$ be a weighted network with nodes $\dom{X}=\{1,\dots,N\}$, links $E\subseteq\dom{X}^2$ and weight matrix $W$. The weight matrix is given by $W:=[w_{\alpha\to\beta}]_{\alpha,\beta\in\dom{X}}$ where $w_{\alpha\to\beta}\ge 0$ denotes the weight of the link $(\alpha,\beta)\in E$ starting at node $\alpha$ and pointing at node $\beta$. (We use Greek letters to indicate nodes.) For an undirected network we set $(\alpha,\beta) = (\beta,\alpha)$ and require that either $(\alpha,\beta)\in E$ or $(\beta,\alpha)\in E$ to avoid the double counting of edges. A set $C\subseteq\dom{X}$ is a clique if it is a complete subnetwork of $\graph$, i.e., for any two distinct nodes $\alpha,\beta\in C$ there exists a connecting link $(\alpha,\beta) \in E$. 

For an undirected network, the degree $k_\alpha$ of node $\alpha$ is the number of links connected to it. We denote the average degree of the network as $\kavg$. The network density $\rho$ is defined as 
\begin{align}
    \rho = 2 \cdot \frac{|E|}{|\dom{X}| (|\dom{X}| - 1)} = \frac{\kavg}{|\dom{X}| - 1}.
\end{align}

Consider a \emph{clustering} $\dom{Y}$ of $\dom{X}$ into a set of $K$ nonempty elements $\dom{Y}_i$, i.e., $\dom{Y}:=\{\dom{Y}_i | i\in\idxset{Y}\}$ where $\idxset{Y} = \{1, \dots, K\}$ denotes the index set of $\dom{Y}$. We index the elements of such a clustering by Roman letters and refer to them as \emph{clusters} or \emph{communities}. If the clusters are disjoint we call them \emph{non-overlapping} and the clustering a \emph{partition}. A partition induces a mapping function $m{:}\ \dom{X}\to\idxset{Y}$, mapping each node of $\graph$ to the index of its containing cluster, i.e., $m(\alpha) = i$ iff $\alpha\in\dom{Y}_i$. For the remainder of this paper we assume all clusterings to be partitions.

For an undirected, unweighted network and a candidate clustering $\dom{Y}$, the mixing parameter of node $\alpha$ is defined as~\citep{Lancichinetti2008, Lancichinetti2009b}
\begin{align}
\mu(\alpha) = \frac{k^{ext}_\alpha}{k_\alpha}
\end{align}
where $k^{ext}_\alpha$ is the number of links between $\alpha$ and nodes outside of its community $\dom{Y}_{m(\alpha)}$. A cluster $\dom{Y}_i$ is a strong community~\citep{Radicchi2004a} if for all of its nodes $\mu(\alpha) < 0.5$, but communities can be defined in a weak sense also for larger values~\citep{Lancichinetti2009b}. Similarly, for an undirected, unweighted network and a candidate clustering $\dom{Y}$, we can define the modularity of the clustering as~\citep{Newman2004, Newman2006}
\begin{equation}\label{eq:modularity}
    Q=\sum_{i\in\mathcal{K}^{\dom{Y}}} \frac{|E_i|}{|E|} -\left ( \frac{1}{2|E|} \sum_{\alpha\in\dom{Y}_i} k_\alpha \right)^2
\end{equation}
where $|E_i|$ denotes then number of internal edges in cluster $\dom{Y}_i$. \citet{Brandes2008} showed that the modularity ranges from $-1/2$ to $1$, with small values indicating weak community structures of the candidate clustering $\dom{Y}$.

\subsection{Random Walks}
We consider random walks $\{X_t\}_{t\in\mathbb{N}}$ on the network $\graph$, i.e., $\{X_t\}$ is a first-order Markov chain on $\dom{X}$. We assume that its stationary transition probability matrix $P:=[p_{\alpha\to\beta}]_{\alpha,\beta\in\dom{X}}$ is derived from the network's weight matrix $W$ and that the random walk is initialized according to an invariant state distribution $p := [p_\alpha]_{\alpha\in\dom{X}}$ that satisfies
\begin{equation}\label{eq:one-step}
 p_\beta = \sum_{\alpha} p_\alpha p_{\alpha \to \beta}.
\end{equation}
We assume that the network is strongly connected, thus $p$ is unique and positive.

Setting $Y_t:=m(X_t)$ defines a stationary process $\{Y_t\}$ on the clusters. Specifically, the marginal and joint probabilities describing $\{Y_t\}$ are obtained as
\begin{subequations}\label{eq:cluster-probabilities}
 \begin{align}
  p_i &:= \Pr{Y_t=i}=\sum_{\alpha\in i }p_\alpha, \quad i \in \idxset{Y}
 \end{align}
and
  \begin{align}
  p_{i,j}&:=\Pr{Y_{t+1}=j,Y_{t}=i} = \sum_{\alpha\in i}\sum_{\beta\in j}p_\alpha p_{\alpha\to\beta}, \quad i,j \in \idxset{Y}.
 \end{align}
\end{subequations}
We further abbreviate $p_{\neg i} := 1 - p_i = \Pr{Y_t \neq i}$ for the marginal complement, $p_{i,\neg j} := p_i - p_{i,j} = \Pr{Y_{t+1}\neq j,Y_{t}=i}$ for the joint complement, and $ p_{i\to j} := \Pr{Y_{t+1}=j|Y_{t}=i}$, respectively $p_{i\not\to j} := 1- p_{i\to j} = \Pr{Y_{t+1}\neq j|Y_{t}=i}$ for the conditional and its complement.

\subsection{Information Theory} \label{ssec:infotheory}
We make use of the following quantities from information theory that are well-described by~\citet[Chapter~2]{Cover2006}. Let $Y,Z$ denote random variables (RV), then we call $\ent{Z}$ the entropy of $Z$, $H(Y|Z)$ the conditional entropy of $Y$ given $Z$ and $\mutinf{Y;Z}$ the mutual information between $Y$ and $Z$. Furthermore, let $p$ and $q$ denote discrete probability distributions over the same alphabet. Then we call $\kld{p}{q}$ the Kullback-Leibler divergence between $p$ and $q$.
If $p,q$ are Bernoulli distributions i.e., $p=[p_1,1-p_1]$ and $q=[q_1,1-q_1]$, then we abbreviate $\kld{p_1}{q_1} :=\kld{p}{q}$. Furthermore, let $P:=[p_{\alpha\to\beta}]_{\alpha,\beta\in\dom{Z}}$ and $Q:=[q_{\alpha\to\beta}]_{\alpha,\beta\in\dom{Z}}$ be transition probability matrices of equal size. The Kullback-Leibler divergence rate $\kldr{\cdot}{\cdot}$ between two stationary Markov chains governed by $P$ and $Q$ is
\begin{align}
 \kldr{P}{Q} := \sum_{\alpha\in\dom{Z}}\sum_{\beta\in\dom{Z}} p_\alpha p_{\alpha\to\beta}\log \frac{p_{\alpha\to\beta}}{q_{\alpha\to\beta}}
\end{align}
given that the Markov chains are irreducible~\citep[Th.~1]{Rached2004}.
	\section{Community Detection via Random Walk Modelling}
We now introduce the Synwalk objective, derive some of its properties, and discuss its relations to Infomap, stochastic block modelling, and model reduction techniques for random walks. For the sake of readability we defer proofs to Appendix~\ref{appx:proofs}.

\subsection{Derivation and Properties of the Synwalk Objective}\label{sec:derivation}

Assume a network $\graph = (\dom{X}, E, W)$ with an inherent community structure $\clutrue$. Consider further a random walker moving on $\graph$ governed by the transition probability matrix $P$, which is derived from the weight matrix $W$. We refer to this random walker as the \emph{network-induced} random walker, as its movements depend on the topology of $\graph$ (i.e., implicitly its community structure). In the next step we design a \emph{synthetic} random walker, governed by some transition probability matrix $Q^\dom{Y}$ which, in contrast to $P$, explicitly depends on some candidate partition $\dom{Y}$. In essence, our approach then aims to find a partition $\dom{Y}$ such that the synthetic random walker behaves (stochastically) as similarly to the network-induced walker as possible. Intuitively, the resulting partition will resemble the intrinsic partition $\clutrue$ very closely. We formalize this concept in the following.

The transition probability matrix that governs the synthetic random walker has a particular structure that depends on a candidate partition. Specifically, suppose that at a given time step the synthetic random walker is at node $\alpha$ in cluster $\dom{Y}_i\in\dom{Y}$. We decide whether to leave or to stay in the current cluster $\dom{Y}_i$ in the next time step based on a cluster-specific Bernoulli distribution $[s_i, 1-s_i]$. In case of a cluster change, we choose a new cluster $\dom{Y}_j \neq \dom{Y}_i$ according to a distribution over clusters $[u_i]_{i\in\idxset{Y}}$. Finally, we choose the next node $\beta$ lying in the new cluster $\dom{Y}_j$ by a cluster-specific distribution over nodes $[r_\beta^j]_{\beta\in \dom{Y}_j}$ (note that $\dom{Y}_j=\dom{Y}_i$ if we stay in the current cluster). This particular structure yields the transition probability matrix $Q^\dom{Y}=[q_{\alpha\to\beta}]_{\alpha,\beta\in\dom{X}}$ where
\begin{align}\label{eq:qmat}
q_{\alpha\to\beta} = 
\begin{cases}
r_\beta^{m(\beta)}  \cdot (1-s_{m(\alpha)}), & m(\alpha)=m(\beta),\\
r_\beta^{m(\beta)}  \cdot s_{m(\alpha)}     \cdot\frac{u_{m(\beta)}}{1 - u_m{(\alpha)}}, &\text{otherwise.}
\end{cases}
\end{align}
Note that when switching clusters we have to normalize the distribution over clusters by $1 - u_{m(\alpha)} = \sum_{k\neq m(\alpha)} u_k$ since we exclude the current cluster $m(\alpha)$ as a choice.

Given the network-induced and the synthetic random walker, the aim is now to optimize the candidate partition $Q^\dom{Y}$ and its parameters to maximize the similarity between the resulting random walks. We quantify this similarity via the Kullback-Leibler divergence rate $\kldr{P}{Q^\dom{Y}}$, i.e., the lower $\kldr{P}{Q^\dom{Y}}$, the more similar are $P$ and $Q^\dom{Y}$, and the more likely it is that the synthetic random walker produces realizations of random walks that are also \emph{typical} for the network-induced random walker~\citep{Kesidis1993}. Hence, the optimal partition $\dom{Y}^*$ satisfies
\begin{align}\label{eq:original}
    \dom{Y}^* \in \argmin_{\dom{Y}} \left [ \min_{\{[r_\alpha^i]_{\alpha\in \dom{Y}_i},\ s_i,\ u_i\}_{i\in\idxset{Y}}}  \kldr{P}{Q^\dom{Y}} \right ].
\end{align}

The cluster-specific distributions over nodes and the cluster-specific Bernoulli distributions minimizing~\eqref{eq:original} can be shown to be
\begin{align}
    r_\alpha^{i,*} &= \frac{p_\alpha}{p_i} = \Pr{X_t=\alpha,Y_t=i} \\
    s_i^* &= p_{i\not\to i} = \Pr{Y_{t+1}\neq i|Y_{t}=i}.
\end{align}
Regarding the distribution over clusters $[u_i]_{i\in\idxset{Y}}$ there exists no closed-form solution to the best of our knowledge. Nevertheless, by choosing
\begin{align}
u_i &= p_i = \Pr{Y_t=i}
\end{align}
as a sub-optimal solution we can relax the original optimization problem in~\eqref{eq:original} to arrive at (see Proposition~\ref{prop:mainProp} in Appendix~\ref{appx:mainProp})
\begin{align}\label{eq:synwalk}
    \dom{Y}^* \in \argmax_{\dom{Y}} \sum_{i\in\idxset{Y}} p_i\kld{p_{i\to i}}{p_i}.
\end{align}
We hence define the \emph{Synwalk objective} as follows.
\begin{definition}\label{def:synwalk}
The Synwalk objective for a given partition $\dom{Y}$ is
\begin{align}\label{eq:si_objective}
\objective{\dom{Y}} :&= \sum_{i\in\idxset{Y}} p_{i,i}\log \frac{p_{i\to i}}{p_i} + p_{i,\neg i}\log \frac{p_{i\not\to i}}{p_{\neg i}} = \sum_{i\in\idxset{Y}} p_i\kld{p_{i\to i}}{p_i} 
\end{align}
\end{definition}

Optimizing the Synwalk objective is a combinatorial and non-convex problem. As we show in Proposition~\ref{prop:bounds} in Appendix~\ref{appx:bounds}, $\objective{\dom{Y}}$ is bounded via
\begin{align}\label{eq:bounds}
    0 \leq \objective{\dom{Y}} \leq \mutinf{Y_t;Y_{t-1}} \leq \mutinf{X_t;X_{t-1}}.
\end{align}

The following observation supports the rationale behind Synwalk's aptness as a community detection method. Consider an unweighted network of disconnected cliques. As we show in Appendix~\ref{app:proof_optimal}, the Synwalk objective achieves its global maximum for a community structure identical to the clique structure of this network. Although isolated cliques are an unrealistic scenario for community detection, they carry the intuition of the concept of a community, i.e., strong internal and weak external connections. Synwalk's optimal behaviour in this idealized edge case theoretically grounds our strong experimental results in Section~\ref{sec:lfr_experiments}. 

 For additional insights based on theoretical considerations and synthetic toy data we refer the reader to~\citep[Sections~3.2 \& 5.1]{Toth2020}.

\subsection{Relation to Infomap and (Stochastic) Block Modeling}\label{sec:rationale}

The design of our random walk model was inspired by Infomap's coding scheme. Recall Infomap's two-level codebook structure described in Section~\ref{sec:related_work}, i.e., the cluster codebooks with node and exit codewords, and the global index codebook. The distributions assembling the dynamics of our synthetic random walker in~\eqref{eq:qmat} correspond to these codebooks:
\begin{inparaenum}[(i)]
\item the cluster-specific distributions over nodes $[r_\alpha^i]_{\alpha\in \dom{Y}_i}$ correspond to the cluster codebooks,
\item the cluster-specific Bernoulli distributions $\{s_i\}_{i\in\idxset{Y}}$ determining a cluster change correspond to the exit codewords, and
\item the distribution over clusters $[u_i]_{i\in\idxset{Y}}$ corresponds to the index codebook.
\end{inparaenum} Thus, while Infomap takes an \emph{analytic} approach to community detection by applying the minimum description length principle with a specific codebook structure, Synwalk takes a \emph{synthetic} approach by trying to mimic the network-induced random walk with our synthetic random walk model. 

The definition of $Q^\dom{Y}$ in~\eqref{eq:qmat} and of the optimization problem~\eqref{eq:original} are reminiscent of stochastic block modeling under Kullback-Leibler divergence. The main difference is that in stochastic block modeling, one tries to infer model parameters -- e.g., community structure, inter- and intra-community edge probabilities -- such that the likelihood of a given graph is maximized. In other words, block modeling infers the parameters of a random graph model, i.e., a generative model from which graphs can be drawn, such that the likelihood of the graph under consideration is maximized. An essential point for stochastic block models is that these models have limited degrees of freedom, and that a good fit between the model and the graph is achieved by selecting an appropriate candidate clustering for the former. In contrast, Synwalk first transforms the graph under consideration to a random walk model, characterized by the transition probability matrix $P$. Then, the aim of Synwalk is to infer the parameters -- i.e., the community structure and parameters of $Q^\dom{Y}$ -- of another random walk model such that the resulting random walk is "close" to the original one in a well-defined sense. Furthermore, it is essential that $Q^\dom{Y}$ has less degrees of freedom than $P$; while, for $N$ nodes and $K$ candidate clusters, $P$ has $N(N-1)$ degrees of freedom, the degrees of freedom of $Q^\dom{Y}$ are limited to $K + (K-1) + (N-K)=N+K-1$. Thus, Synwalk can adequately be interpreted as an approach to "random walk modeling".

Finally,~\citet{Hurley2015,Hurley2016} proposed a method that combines random walks on networks with (generalized) block modelling. While they consider the network-induced random walk clusters rather than on nodes, they also use the Kullback-Leibler divergence to measure the similarity with a target random walk and, thus, the fitness of the candidate clustering. For a specific target random walk it can be shown that their approach becomes equivalent to the goal of maximizing $\mutinf{Y_t;Y_{t-1}}$~\citep[Sec.~III.A]{Hurley2015}. The same cost function is also obtained by further relaxing our Synwalk objective~\eqref{eq:si_objective} (cf.~Appendix~\ref{appx:bounds}).

	\section{Experiments on the LFR Benchmark}\label{sec:lfr_experiments}


To ensure a fair comparison, we used the optimization heuristic of Infomap to maximize the Synwalk objective (see~\citep[Appendix A]{Rosvall2009, Toth2020} for algorithmic details) with the same set of default hyper-parameters for Infomap and Synwalk. 

The modified framework used in the course of this work can be found at \url{https://github.com/synwalk/synwalk}. For Walktrap we assume a default value of $T=4$ where $T$ is the hyper-parameter describing the random walk length used to compute the node and cluster distances. We use the same setup in our experiments with empirical networks in Section~\ref{sec:realworld}.

To validate and compare the results of community detection methods it is common practice to evaluate their performance on benchmark networks~\citep{Yang2016, Fortunato2016, Newman2004, Lancichinetti2009c, Orman2009} where the ground truth community structure is known. The prevalent benchmark in more recent studies~\citep{Yang2016, Orman2009} is the LFR benchmark~\citep{Lancichinetti2008, Lancichinetti2009b} and hence, we adopt it in our experiments. We generate the LFR benchmark networks with parameters as given in Table~\ref{tab:lfr}. 

\begin{table}[b]
\centering
\caption{Parameter setup for generating the LFR benchmark networks.} \label{tab:lfr}
\resizebox{\textwidth}{!}{ 
\begin{tabular}{@{}c|l|c|c@{}}
\toprule
\textbf{Parameter} 		& {\textbf{Description}} 	& {\textbf{Parameter Set A}}    & {\textbf{Parameter Set B}}	\\ \midrule
$N_c^{max}$ & Maximum community size            & $0.2 \cdot N$             & $0.1 \cdot N$         \\ \midrule
$N_c^{min}$ & Minimum community size            & $0.25 \cdot N_c^{max}$    & $10$    \\ \midrule
$k^{max}$   & Maximum node degree               & $0.95 \cdot N_c^{max}$    & $0.95 \cdot N_c^{max}$	\\ \midrule
$\kavg$     & Average node degree               & $\{15, 25, 50\}$          & $20$	    \\ \midrule
$\beta$ 	& Community size distribution exponent 	& $1.0$                 & $1.0$	                \\ \midrule
$\gamma$	& Degree distribution exponent 			& $2.0$                 & $2.0$                 \\ \bottomrule
\end{tabular}
}
\end{table}

We employ the adjusted mutual information (AMI,~\citep{Vinh2010}) as a performance measure when comparing the partitions found by different community detection algorithms with the ground truth community structure.
AMI values close to $1$ indicate high similarity between the found partition and the ground truth, whereas a values around $0$ reflect low similarity. Let $\clutrue$ denote the ground truth clustering and $\clupred$ any predicted clustering, then the AMI is defined as
\begin{align}\label{eq:ami}
\ami{\clutrue}{\clupred} = \frac{\mutinf{\clutrue; \clupred} - \E{\mutinf{\clutrue; \clupred}}}{\frac{1}{2}[ H(\clutrue) + H(\clupred)] - \E{\mutinf{\clutrue; \clupred}}},
\end{align}
where $\E{\cdot}$ denotes the expectation operator with respect to a chosen permutation model. We normalize the AMI by the arithmetic mean as in~\eqref{eq:ami}

The source code required for executing the experiments of this section can be found at \url{https://github.com/synwalk/synwalk-analysis}.


\subsection{AMI as a Function of the Mixing Parameter}\label{sssec:ami_vs_mu}
In this experiment we use parameter set A (see Table~\ref{tab:lfr}) to generate the LFR benchmark networks. We fix the network size and average degree of the generated LFR networks while varying their mixing parameter $\mu$ between $0.2$ and $0.8$. Experiments with varying network sizes are shown in Appendix~\ref{appx:additional_lfr}.

The results for the AMI as a function of the mixing parameter are shown in Fig.~\ref{fig:lfr_ami_vs_mu}. As can be seen, Infomap correctly identifies the communities for sufficiently small values of $\mu$ and transitions to vanishing AMI around $\mu \approx 0.5$. This behavior reflects the definition of communities in a strong and weak sense as proposed by~\citet{Radicchi2004a} and was also observed by~\citet{Yang2016}. We explain this behaviour by looking at Infomap's coding scheme. If $\mu>0.5$, then the random walker will have a higher probability of exiting a community than staying within it. Hence, for the ground truth community structure, the coding overhead due to sending exit codewords will dominate, and clusterings resulting in more efficient encodings can be found, e.g., by putting all nodes into a single common cluster. Indeed, we observed exactly this behavior for Infomap in our experiments.

Unlike Infomap, Synwalk does not penalize frequent transitions between communities, although our random walk model resembles Infomap's coding scheme (cp. Section~\ref{sec:rationale}). Thus, the performance transitions of Synwalk, similarly to Walktrap, occur at increasing values of $\mu$ for increasing network densities (Fig.~\ref{fig:lfr_ami_vs_mu}, columns from top to bottom and rows from right to left). Intriguingly, for roughly the same network density the transition phases shift to higher values of the mixing parameter as the average degree increases (cp. Appendix~\ref{appx:additional_lfr}). Hence, neither the mixing parameter nor the network density sufficiently characterizes the AMI performance of Synwalk and Walktrap. We analyze this phenomenon further in Section~\ref{sssec:node_statistics}. In contrast,  our experiments show that even as we vary the average degree, Infomap's performance mainly depends on the mixing parameter of the networks.

Overall, Synwalk outperforms Infomap in terms of AMI on sufficiently dense networks or networks with mixing parameters $\mu \gtrapprox 0.5$. We perform approximately on par with Walktrap, where we see slightly better performance on networks with lower density (Fig.~\ref{fig:lfr_ami_vs_mu}, top row) and a slight disadvantage on networks with higher density (Fig.~\ref{fig:lfr_ami_vs_mu}, bottom row).

\begin{figure*}[t]
    \centering
    \begin{subfigure}{0.32\textwidth}
        \centering
        \includegraphics[width=1.0\textwidth]{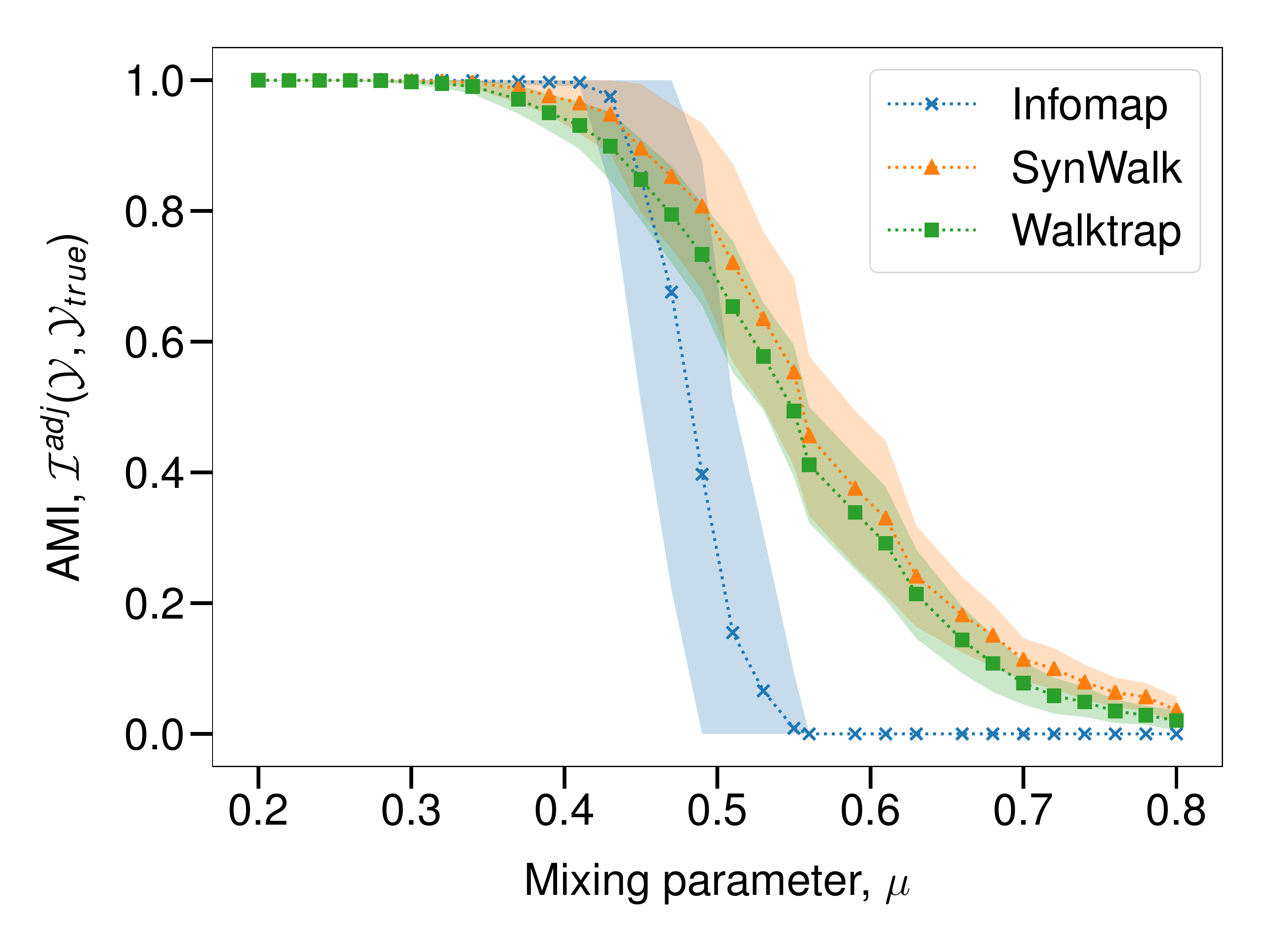}
        \caption{$\kavg = 15$, $N = 300$, $\rho = 0.050$.}

    \end{subfigure}
    \begin{subfigure}{0.32\textwidth}
        \centering
        \includegraphics[width=1.0\textwidth]{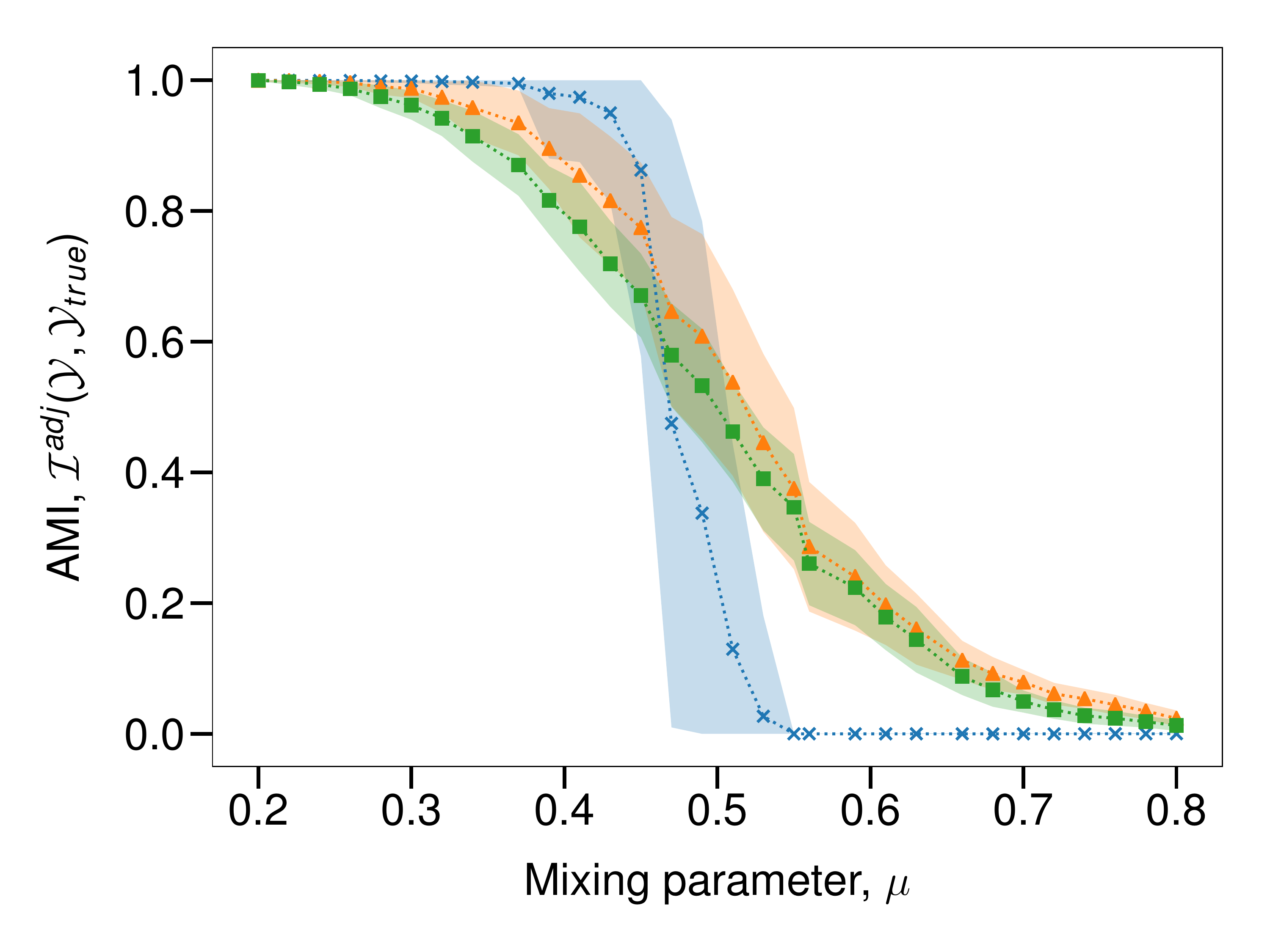}
        \caption{$\kavg = 15$, $N = 600$, $\rho = 0.025$.}
    \end{subfigure}
    \begin{subfigure}{0.32\textwidth}
        \centering
        \includegraphics[width=1.0\textwidth]{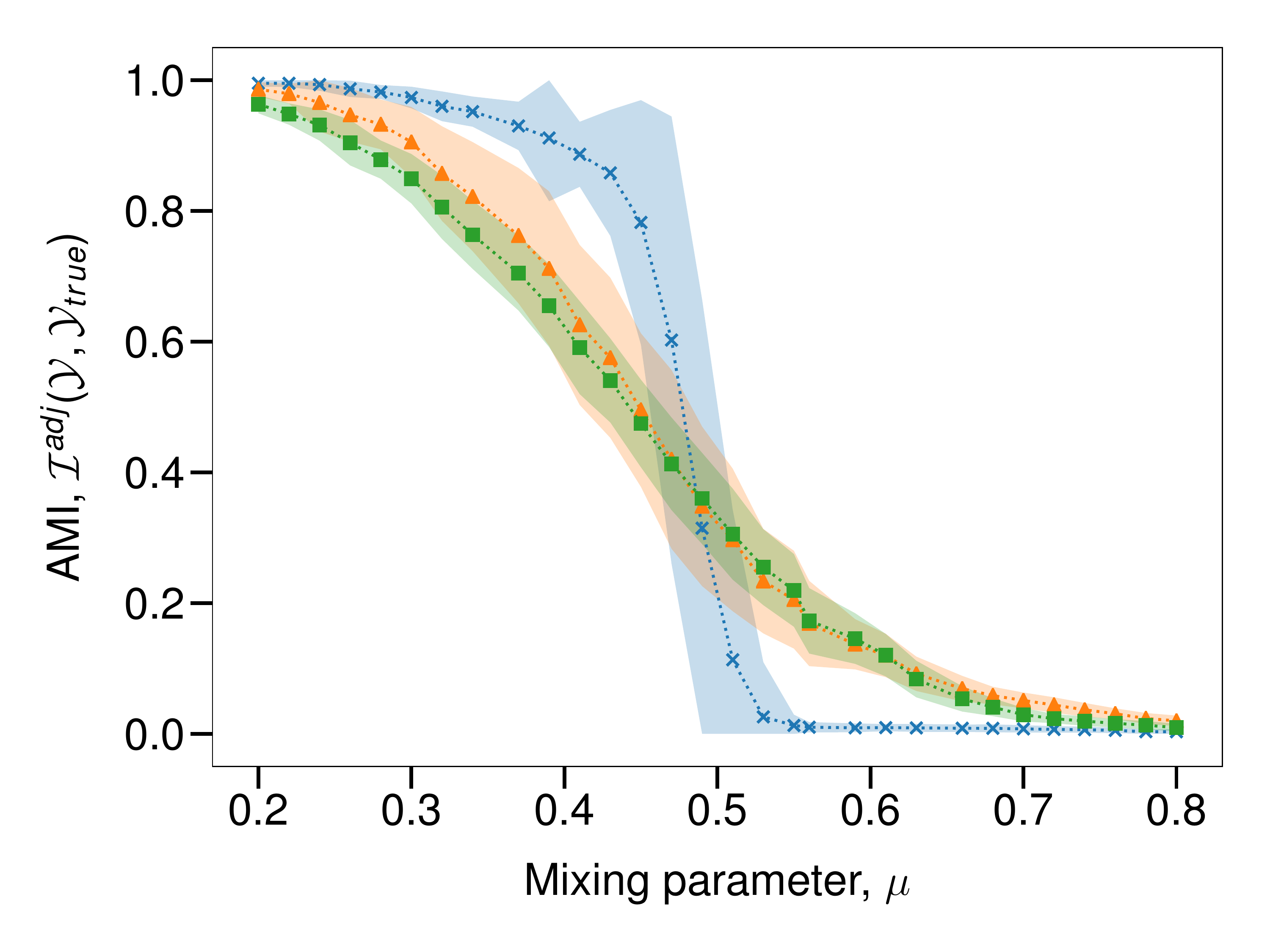}
        \caption{$\kavg = 15$, $N = 1200$, $\rho = 0.013$.}
    \end{subfigure}
    \\
    \begin{subfigure}{0.32\textwidth}
        \centering
        \includegraphics[width=1.0\textwidth]{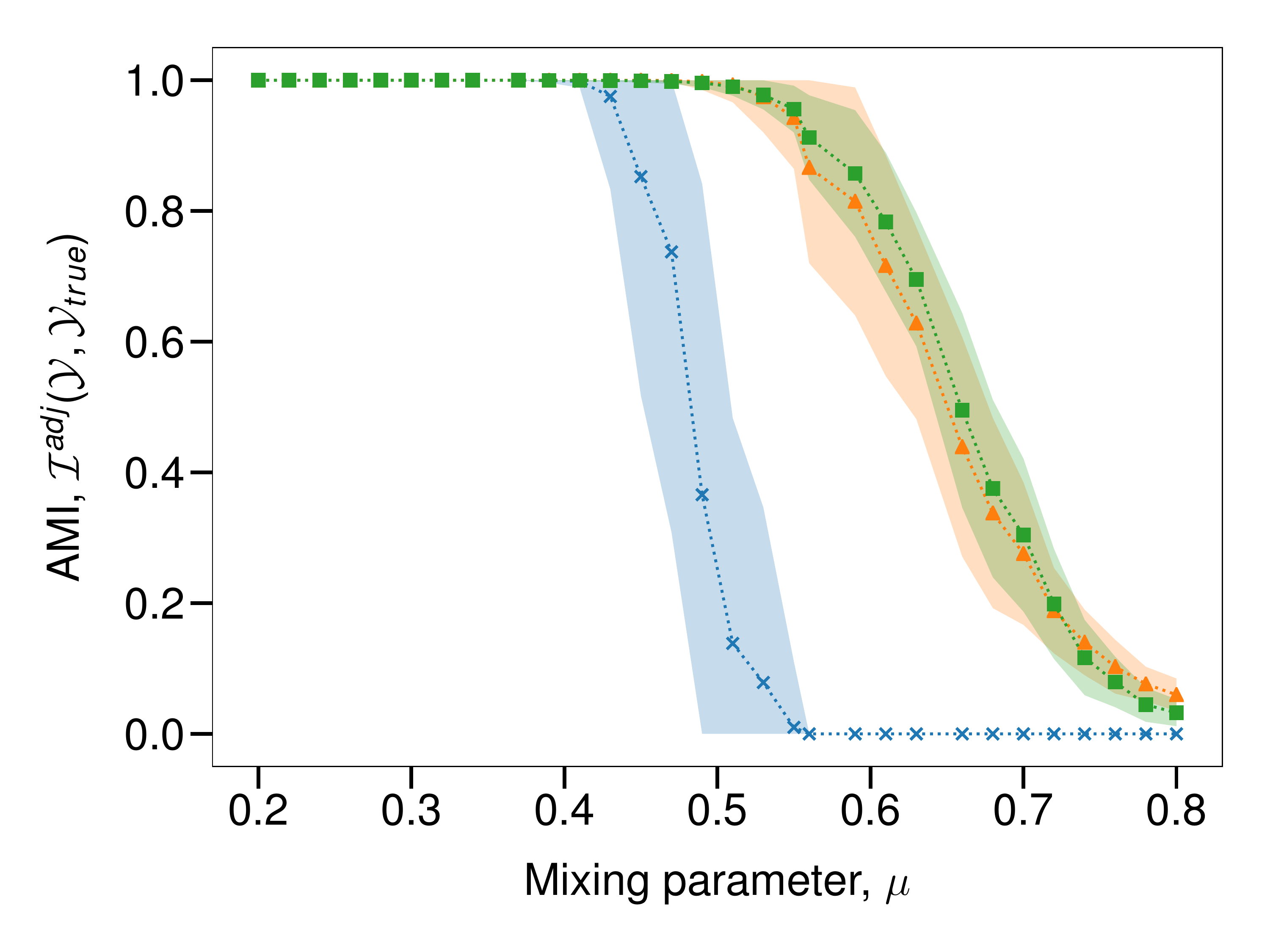}
        \caption{$\kavg = 25$, $N = 300$, $\rho = 0.084$.}
    \end{subfigure}
    \begin{subfigure}{0.32\textwidth}
        \centering
        \includegraphics[width=1.0\textwidth]{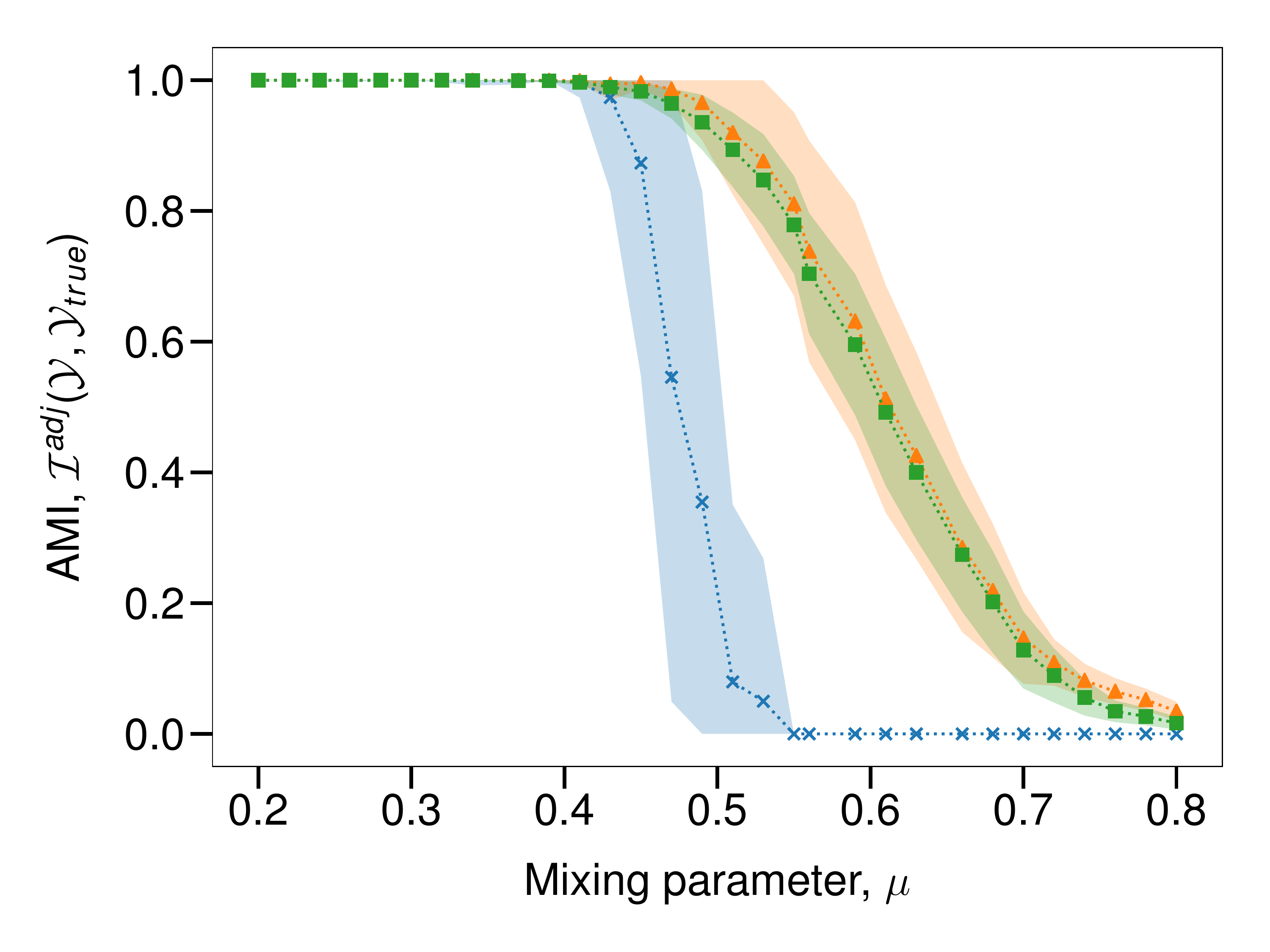}
        \caption{$\kavg = 25$, $N = 600$, $\rho = 0.042$.}
    \end{subfigure}
    \begin{subfigure}{0.32\textwidth}
        \centering
        \includegraphics[width=1.0\textwidth]{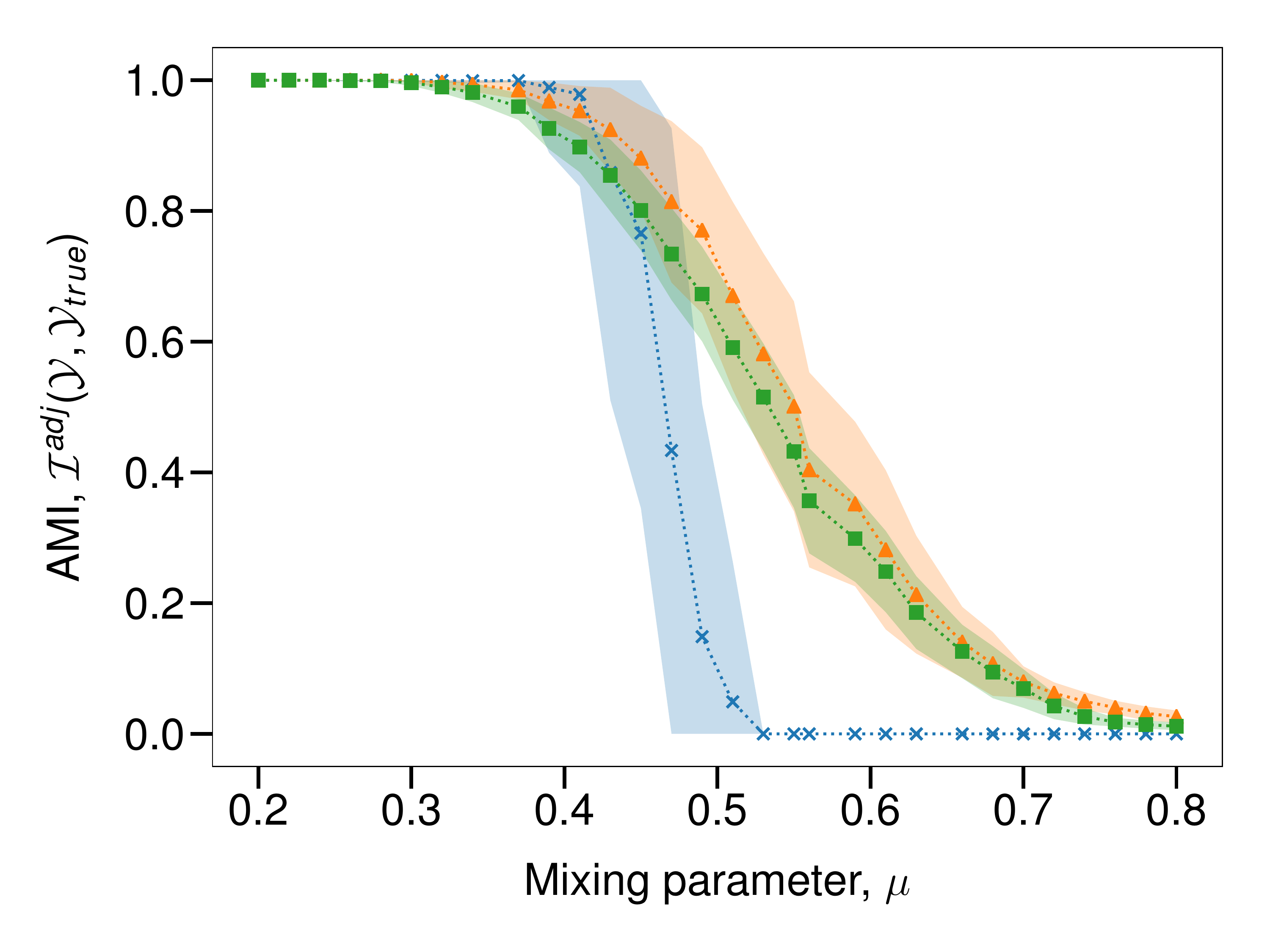}
        \caption{$\kavg = 25$, $N = 1200$, $\rho = 0.021$.}
    \end{subfigure}
    \\
    \begin{subfigure}{0.32\textwidth}
        \centering
        \includegraphics[width=1.0\textwidth]{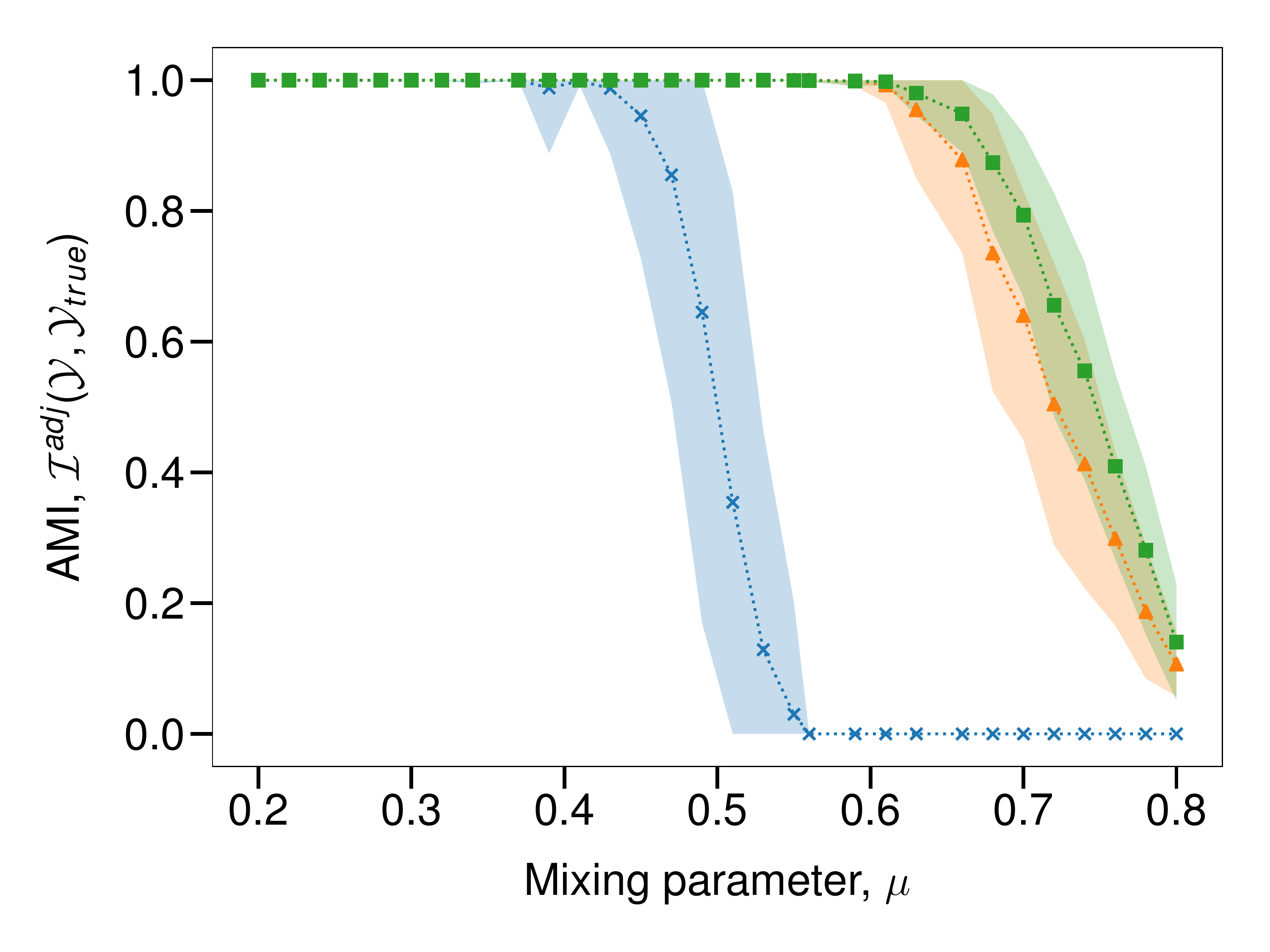}
        \caption{$\kavg = 50$, $N = 300$, $\rho = 0.167$.}
    \end{subfigure}
    \begin{subfigure}{0.32\textwidth}
        \centering
        \includegraphics[width=1.0\textwidth]{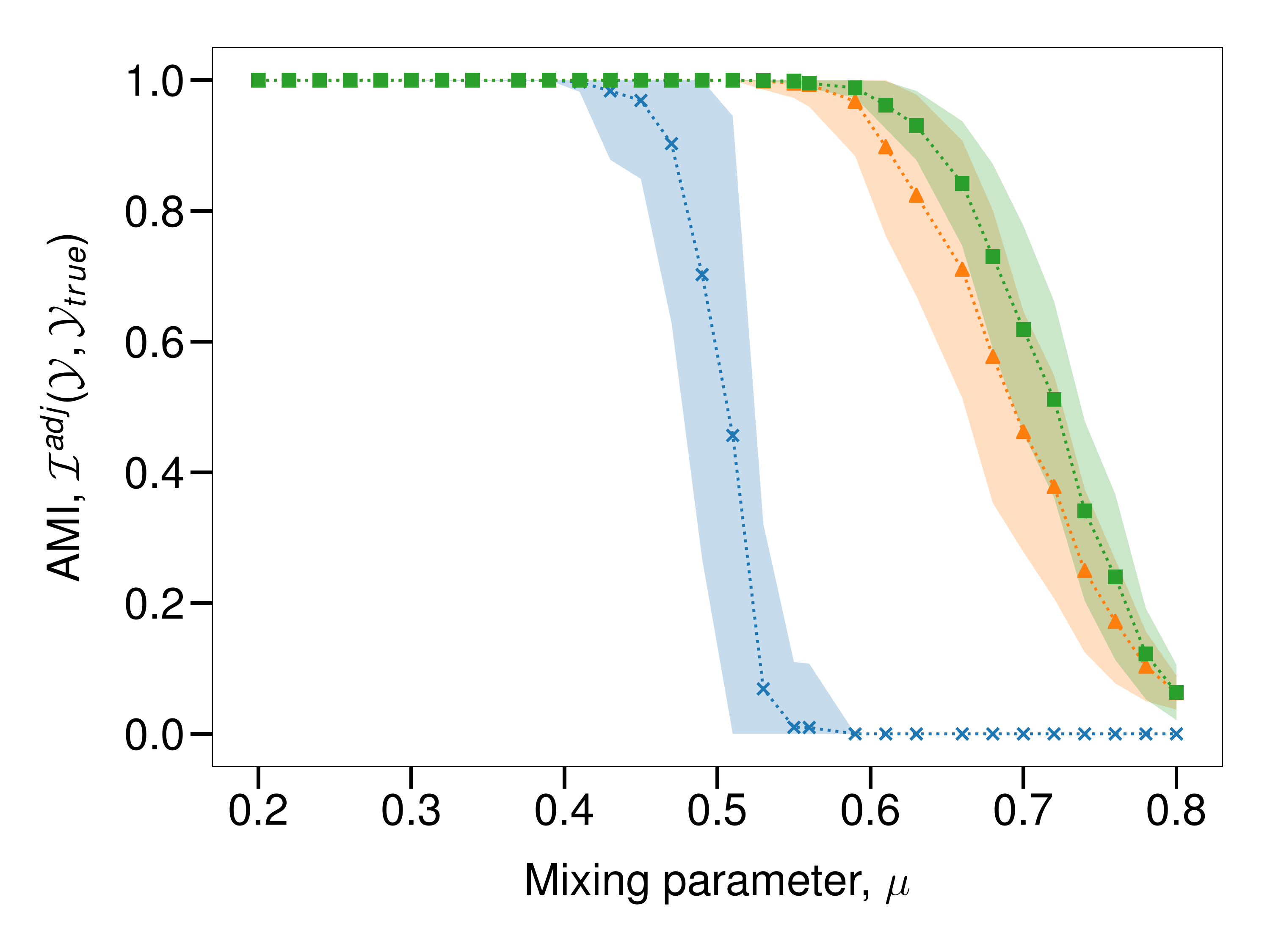}
        \caption{$\kavg = 50$, $N = 600$, $\rho = 0.083$.}
    \end{subfigure}
    \begin{subfigure}{0.32\textwidth}
        \centering
        \includegraphics[width=1.0\textwidth]{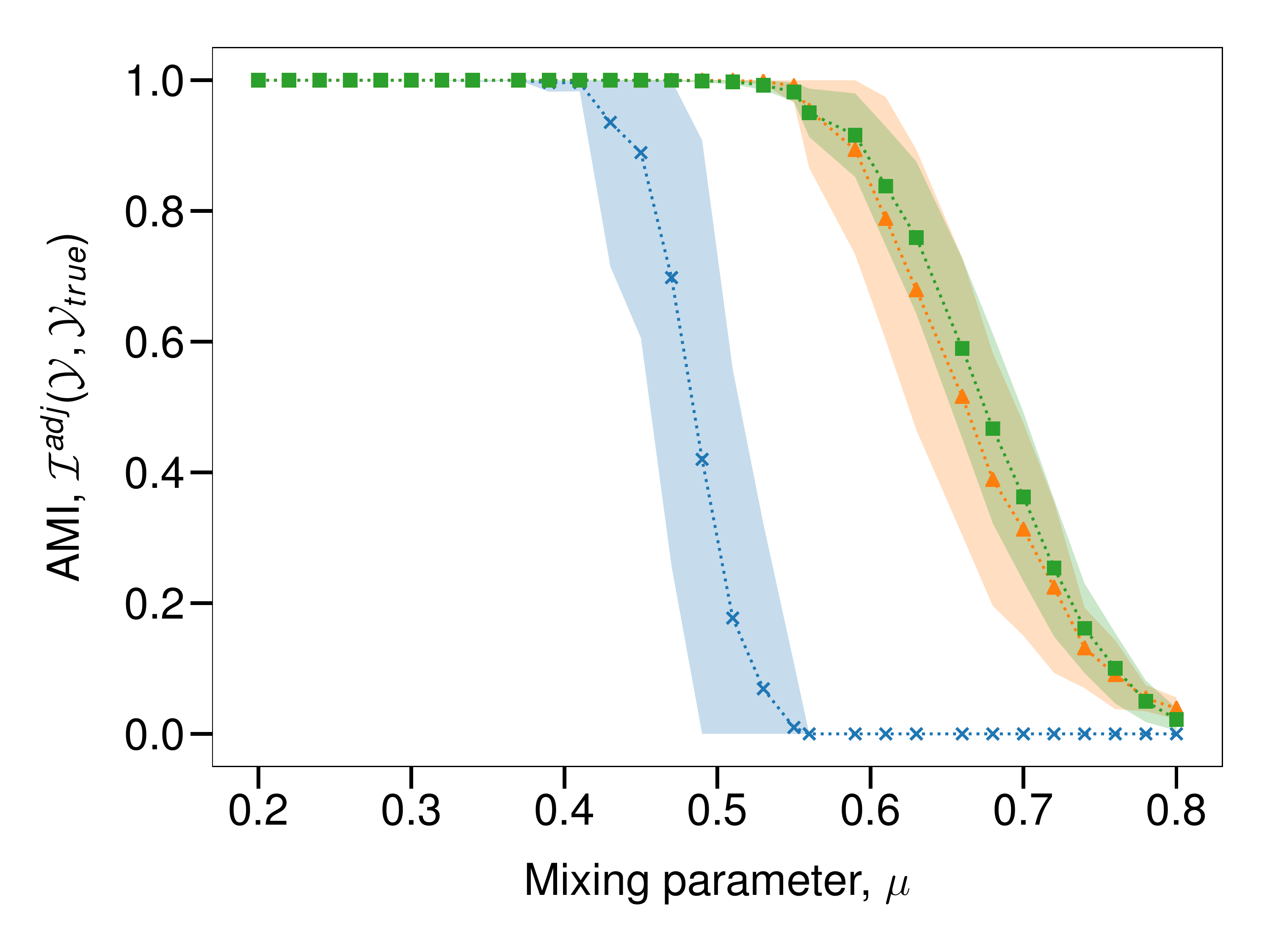}
        \caption{$\kavg = 50$, $N = 1200$, $\rho = 0.042$.}
    \end{subfigure}
    \caption{Comparison of Infomap, Synwalk and Walktrap on LFR benchmark networks with given average degree and network size. The lines and shaded areas show the mean and standard deviation of AMI as a function of the mixing parameter, obtained from $100$ different network realizations. Synwalk outperforms Infomap for sufficiently high mixing parameter and network density. Performance of Synwalk and Walktrap increases with higher average degrees while holding the network density fixed.}\label{fig:lfr_ami_vs_mu}
\end{figure*}

\subsection{Classification Analysis using Node Statistics}\label{sssec:node_statistics}
As we have seen in Section~\ref{sssec:ami_vs_mu}, the AMI performance of Synwalk and Walktrap on LFR networks transitions smoothly for varying values of the mixing parameter. To get deeper insights into the behavioral differences between Synwalk and Walktrap\footnote{Due to the absence of a smooth transition phase for Infomap we do not analyze its behavior in this section.} we analyze the different qualities of their predictions in these transition phases. 

For this purpose we analyze networks with varying network sizes and average degrees that are generated with parameter set A (see Table~\ref{tab:lfr}) while trying to keep the network density and the AMI (by appropriately setting the mixing parameter) constant (cp.\ main diagonal in Fig.~\ref{fig:lfr_ami_vs_mu}). We align any predicted partitions to their respective ground truth partitions using a greedy matching algorithm as described in Appendix~\ref{appx:cluster_matching}. We then consider the nodes in the intersection of the ground truth communities with their aligned counterparts as correctly classified nodes, whereas the residual set of nodes form the group of misclassified nodes.

Given this distinction, we can compare the degree distributions of correctly classified and misclassified nodes. In addition to the node degree $k_\alpha$, we consider the normalized local degree (NLD) $\hat{k}_\alpha$, which we define as the ratio between the node degree and the maximum number of possible links in its containing cluster:
\begin{align}
    \hat{k}_\alpha = \frac{k_\alpha}{\binom{|\dom{Y}_{m(\alpha)}|}{2}} \quad \text{for} \quad |\dom{Y}_{m(\alpha)}| \geq 2.
\end{align}
Note that in general the NLD of a node will be different when computed w.r.t.\ its ground truth community or its predicted community.

\begin{figure*}[t]
    \centering
    \begin{subfigure}{0.32\textwidth}
        \centering
        \includegraphics[width=1.0\textwidth]{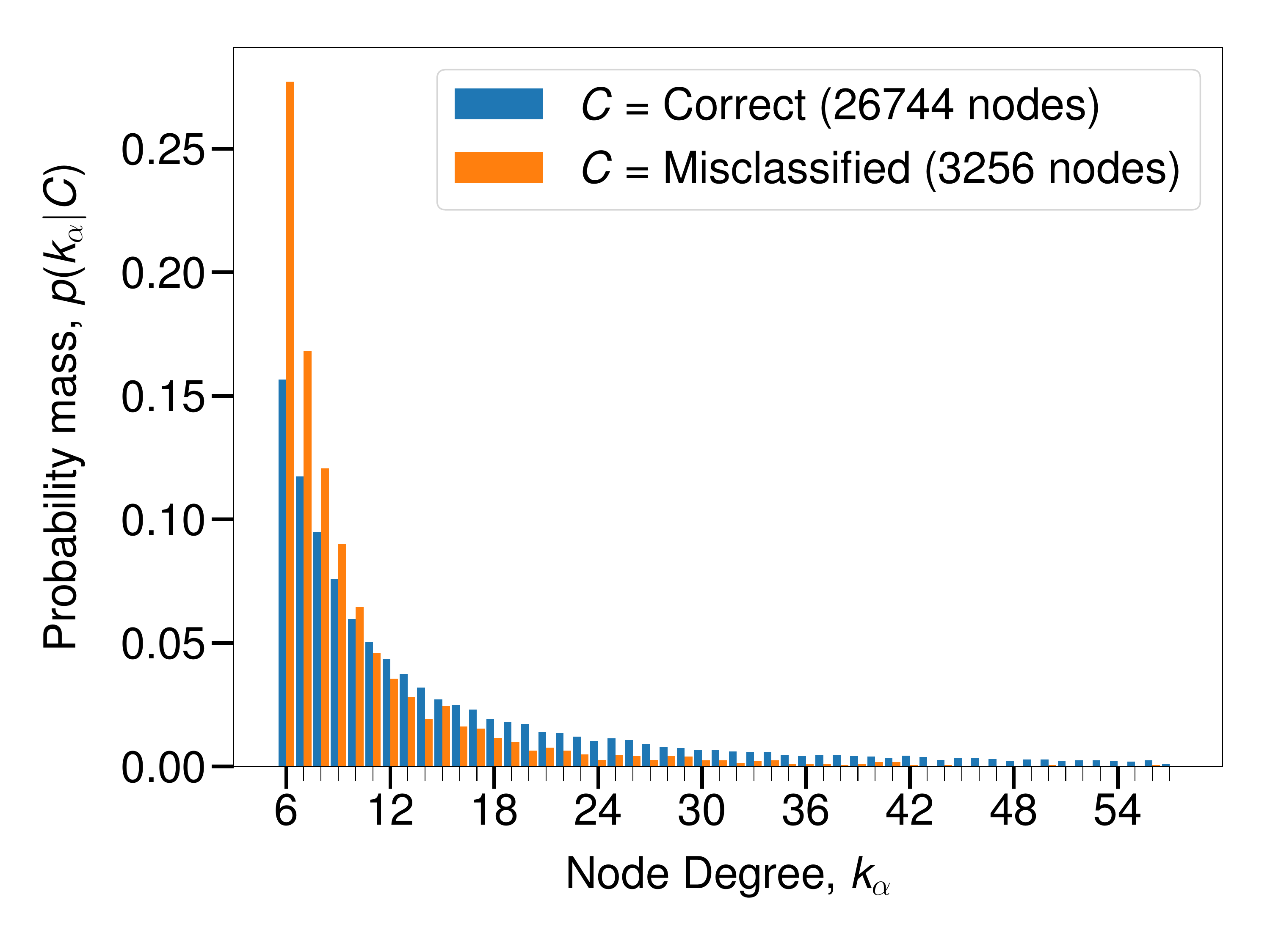}
        \caption{$\kavg = 15$, $n = 300$, $\mu = 0.45$.}
    \end{subfigure}
   \begin{subfigure}{0.32\textwidth}
        \centering
        \includegraphics[width=1.0\textwidth]{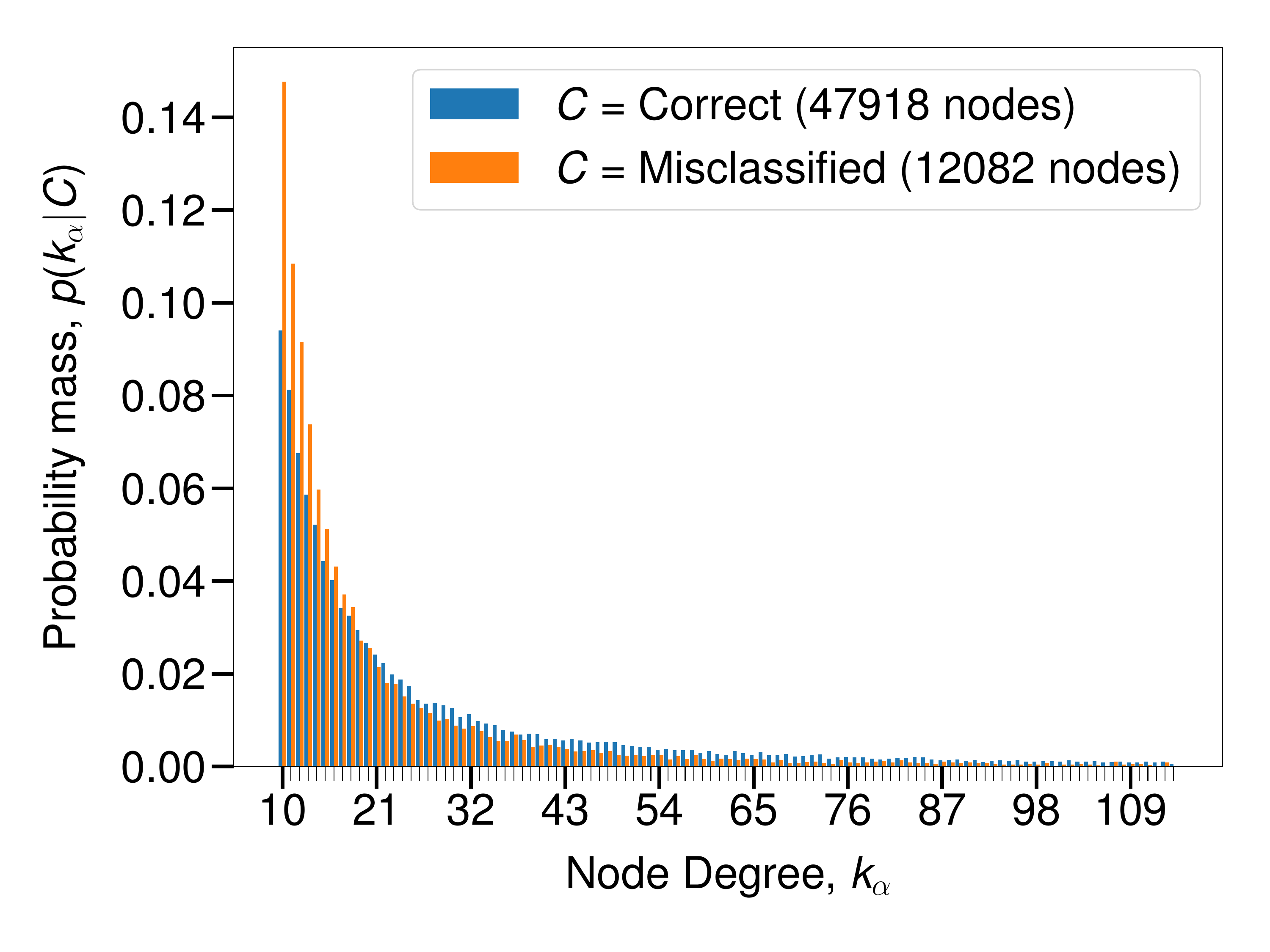}
        \caption{$\kavg = 25$, $n = 600$, $\mu = 0.55$.}
    \end{subfigure}
    \begin{subfigure}{0.32\textwidth}
        \centering
        \includegraphics[width=1.0\textwidth]{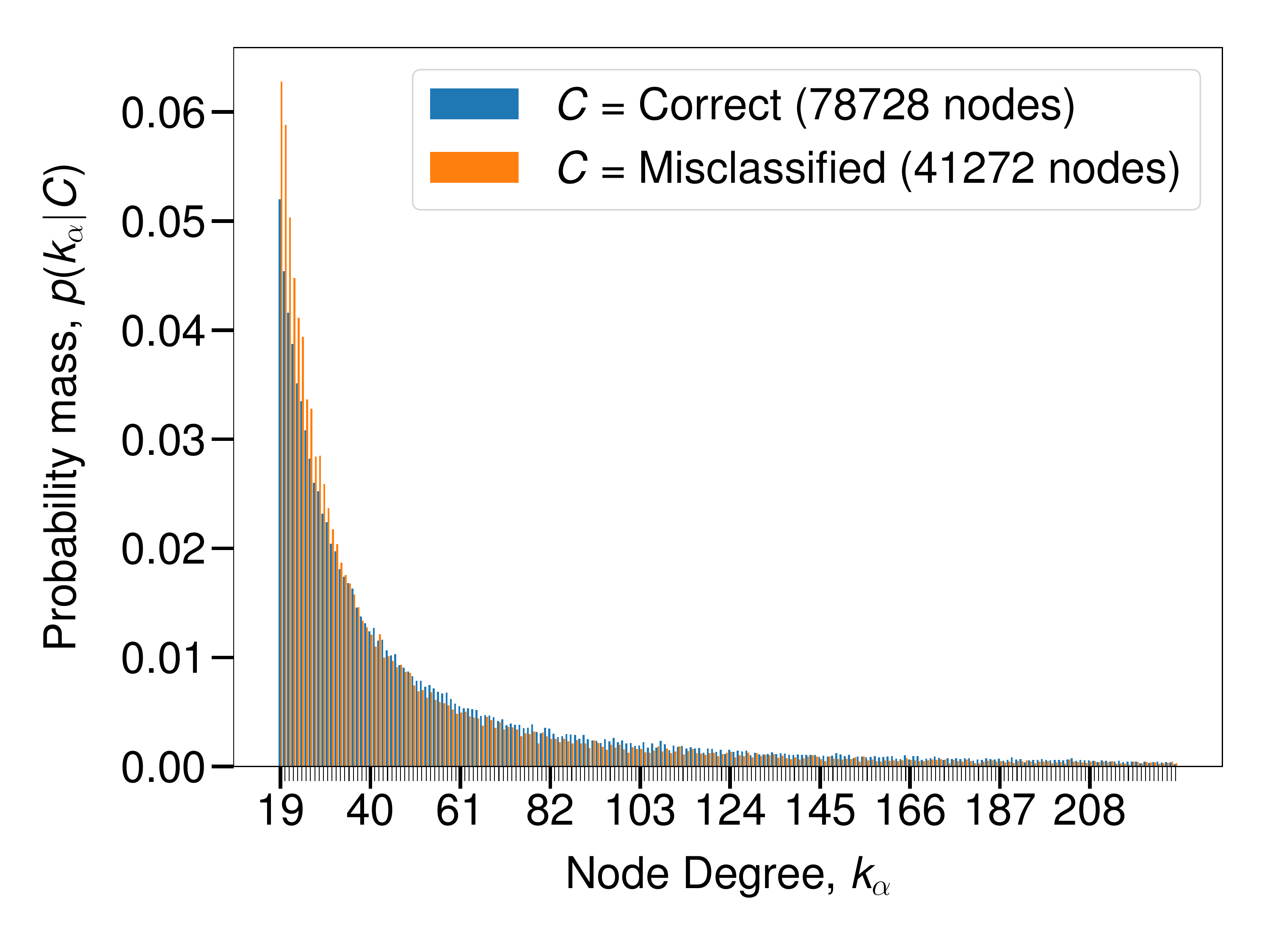}
        \caption{$\kavg = 50$, $n = 1200$, $\mu = 0.63$.}
    \end{subfigure}
    \\
    \begin{subfigure}{0.32\textwidth}
        \centering
        \includegraphics[width=1.0\textwidth]{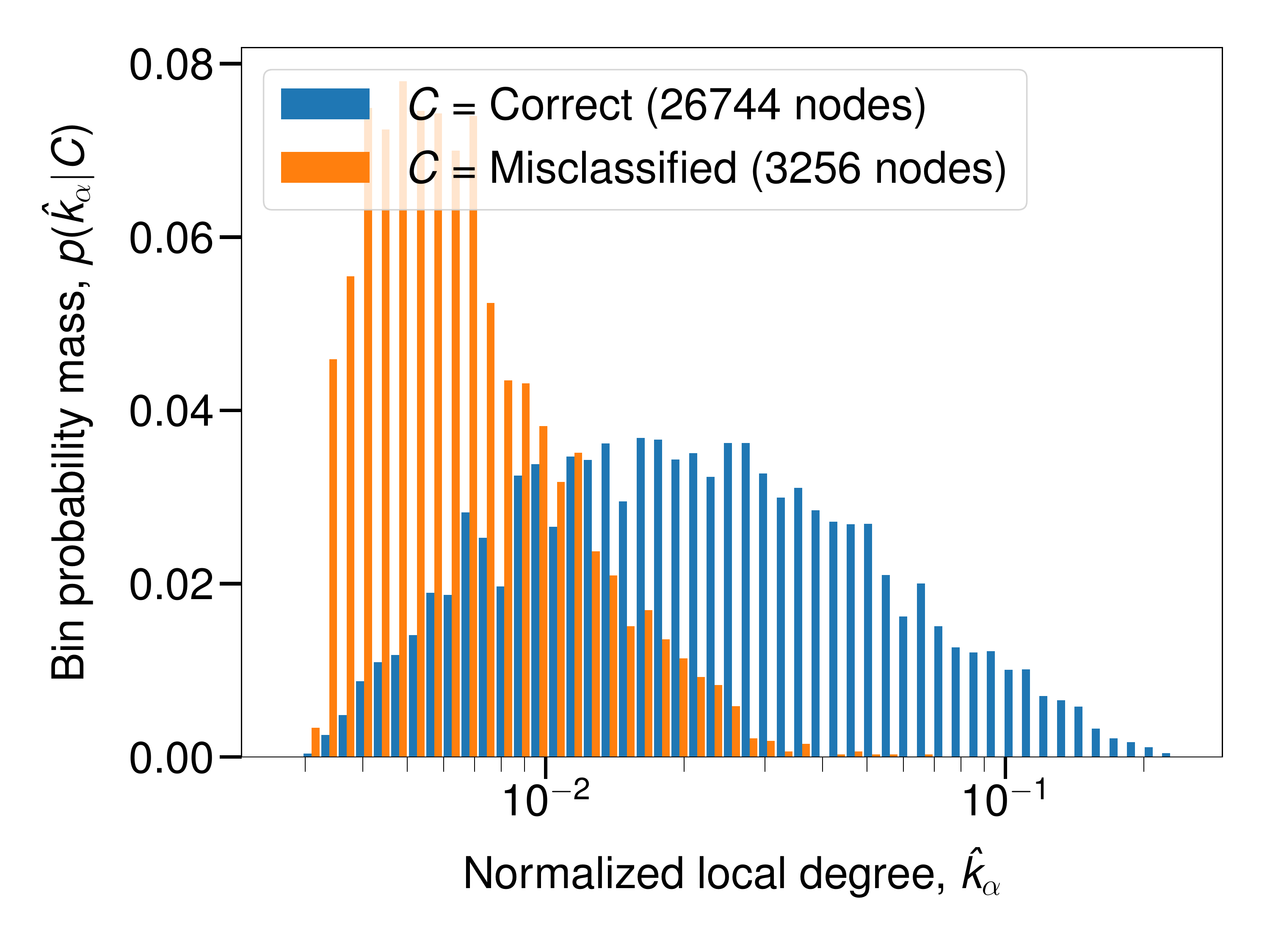}
        \caption{$\kavg = 15$, $n = 300$, $\mu = 0.45$.}
    \end{subfigure}
    \begin{subfigure}{0.32\textwidth}
        \centering
        \includegraphics[width=1.0\textwidth]{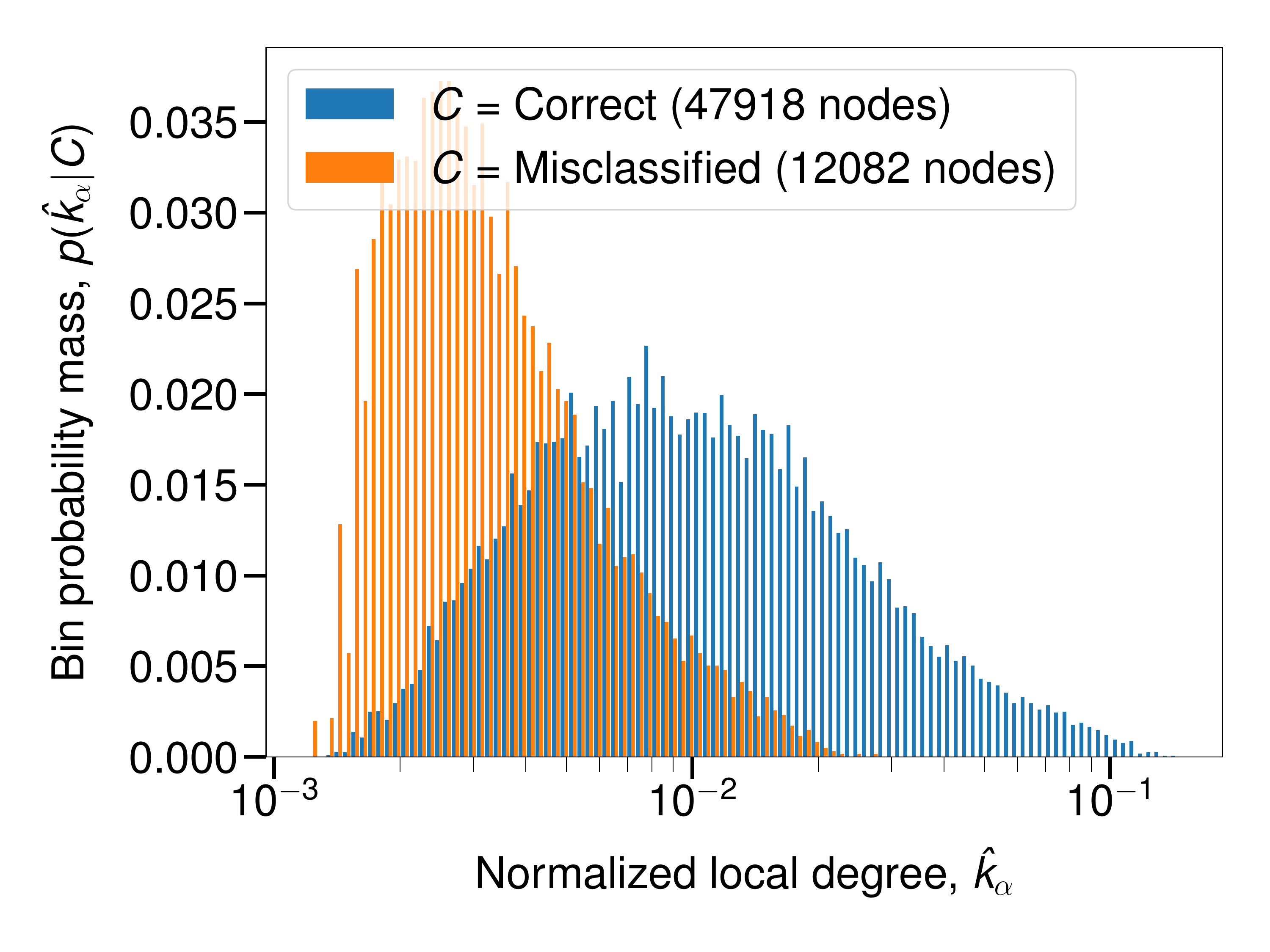}
        \caption{$\kavg = 25$, $n = 600$, $\mu = 0.55$.}
    \end{subfigure}
    \begin{subfigure}{0.32\textwidth}
        \centering
        \includegraphics[width=1.0\textwidth]{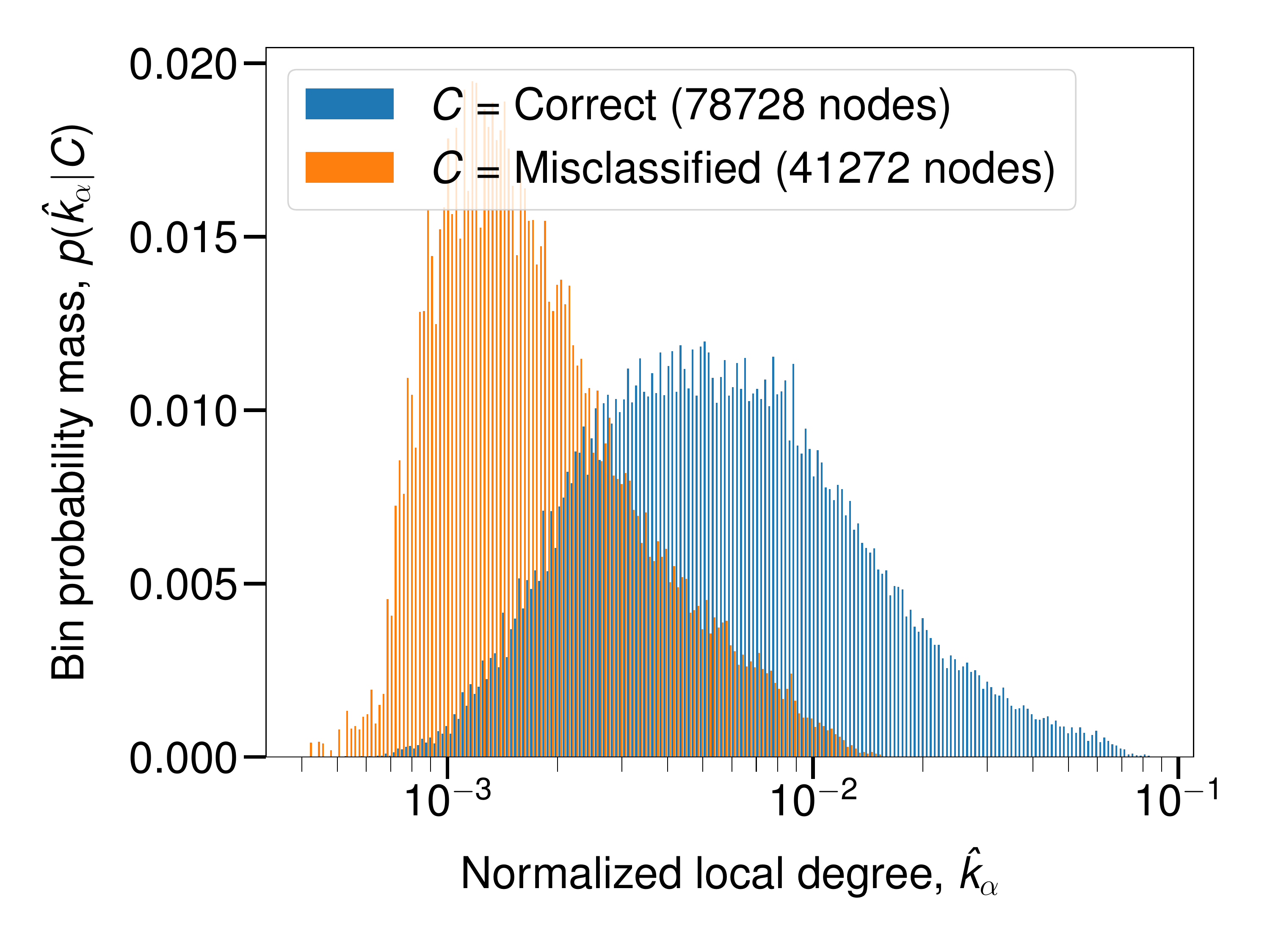}
        \caption{$\kavg = 50$, $n = 1200$, $\mu = 0.63$.}
    \end{subfigure}
    \\
    \begin{subfigure}{0.32\textwidth}
        \centering
        \includegraphics[width=1.0\textwidth]{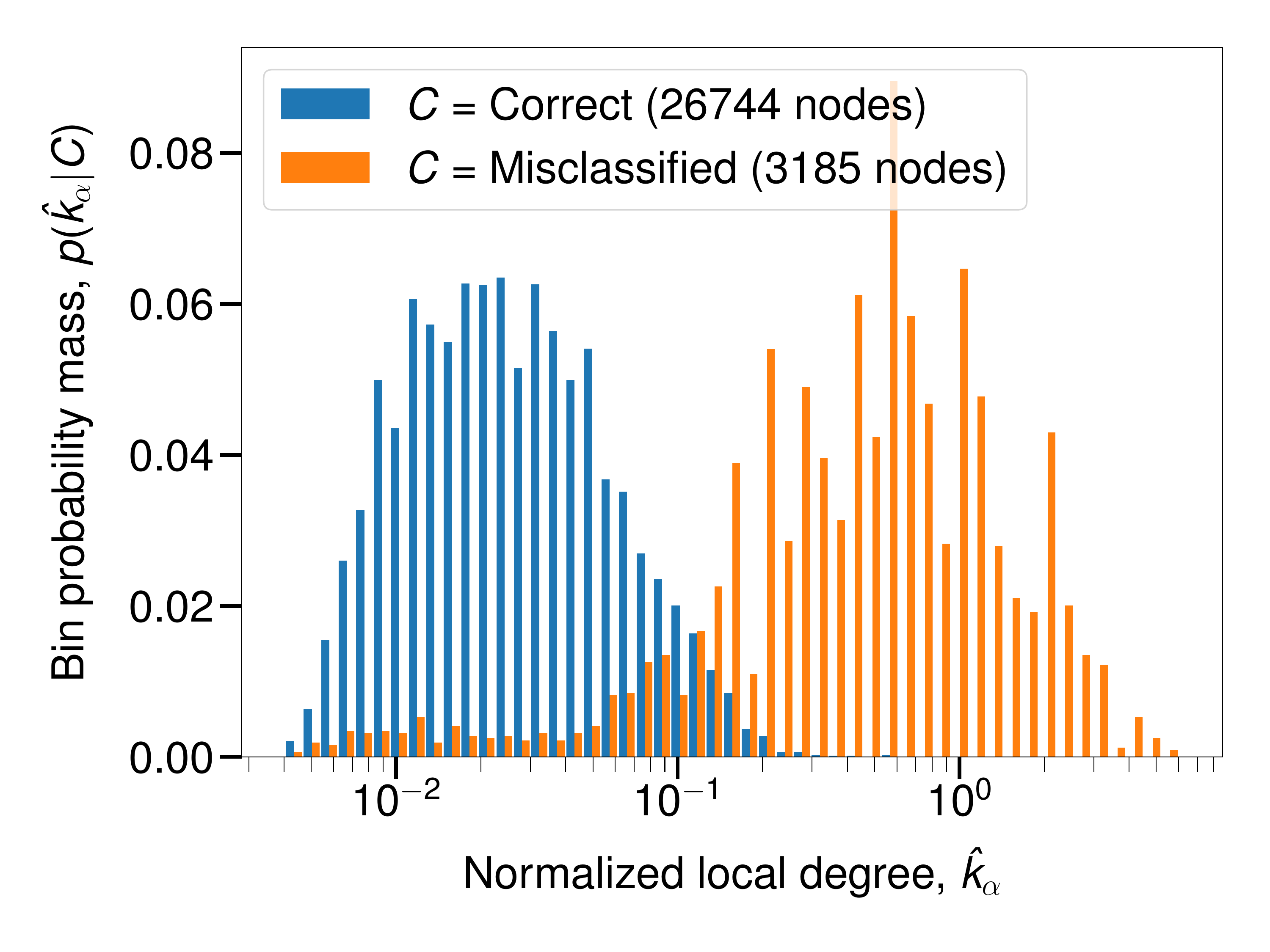}
        \caption{$\kavg = 15$, $n = 300$, $\mu = 0.45$.}
    \end{subfigure}
    \begin{subfigure}{0.32\textwidth}
        \centering
        \includegraphics[width=1.0\textwidth]{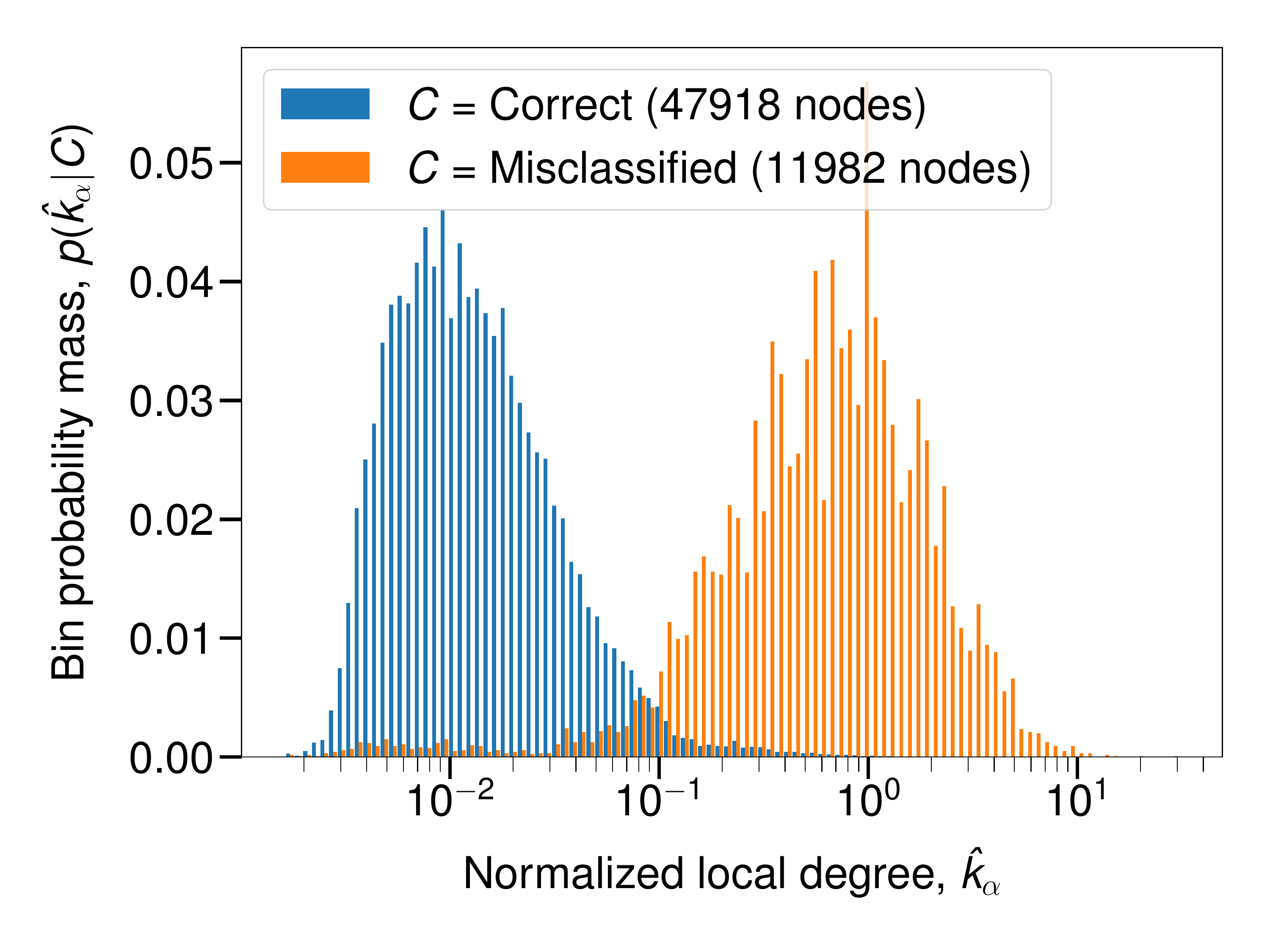}
        \caption{$\kavg = 25$, $n = 600$, $\mu = 0.55$.}
    \end{subfigure}
    \begin{subfigure}{0.32\textwidth}
        \centering
        \includegraphics[width=1.0\textwidth]{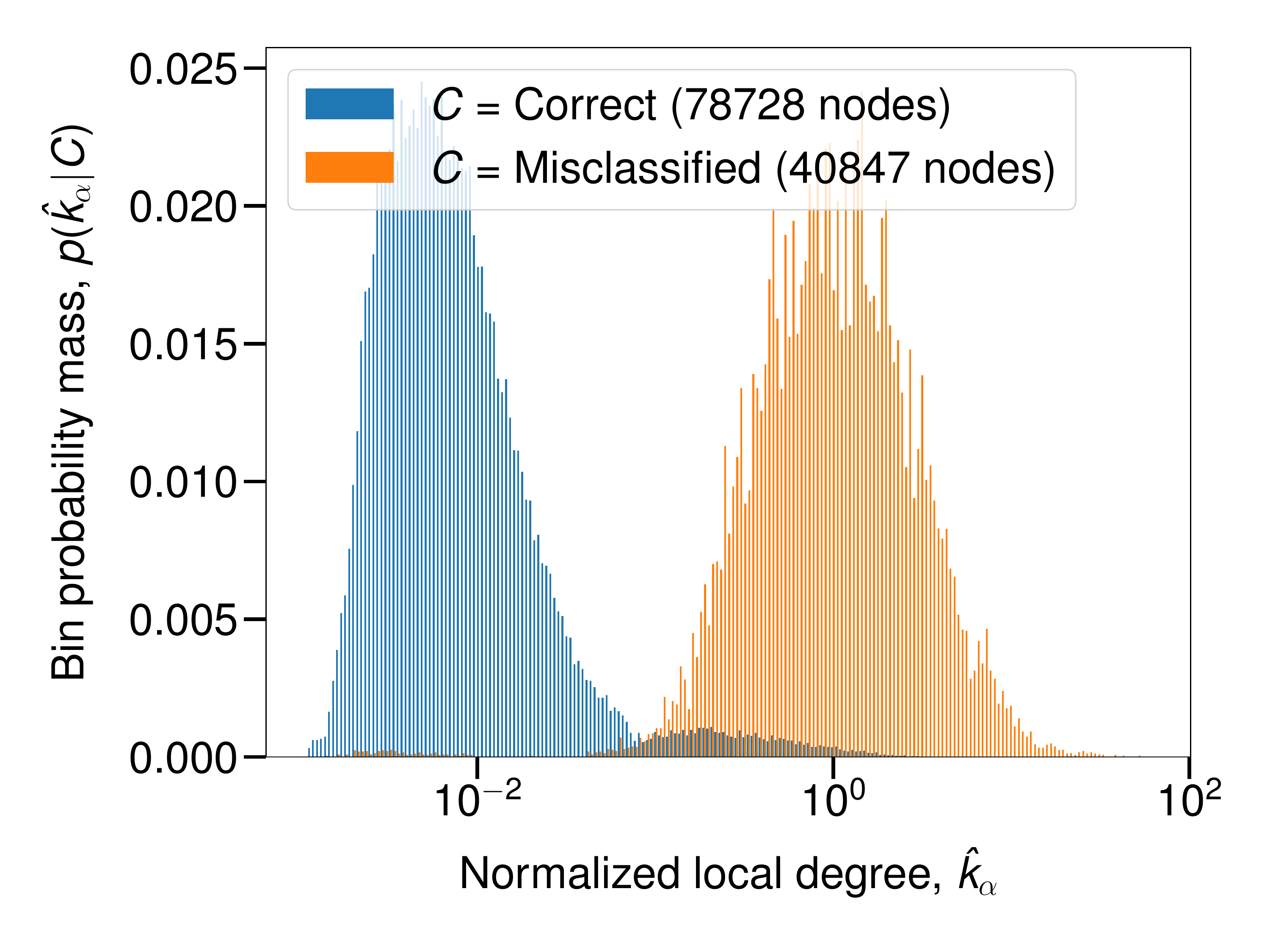}
        \caption{$\kavg = 50$, $n = 1200$, $\mu = 0.63$.}
    \end{subfigure}
    \caption{Degree distributions for correctly classified and misclassified nodes in Synwalk results, obtained from $100$ different LFR networks with common average degree, network size and mixing parameter. The top row shows the distributions of the node degrees, the middle row shows the distribution of the normalized local degrees w.r.t.\ the ground truth communities and the bottom row shows the distributions of the normalized local degrees w.r.t.\ the predicted communities. Synwalk tends to misclassify nodes with low normalized local degree (w.r.t.\ the ground truth communities), whereas the influence of the absolute node degree is negligible. The statistics of the normalized local degrees w.r.t.\ predicted communities indicate that misclassified nodes are assigned to additional, small communities.}\label{fig:nodestat_synwalk}
\end{figure*}

\begin{figure*}[t]
    \centering
    \begin{subfigure}{0.32\textwidth}
        \centering
        \includegraphics[width=1.0\textwidth]{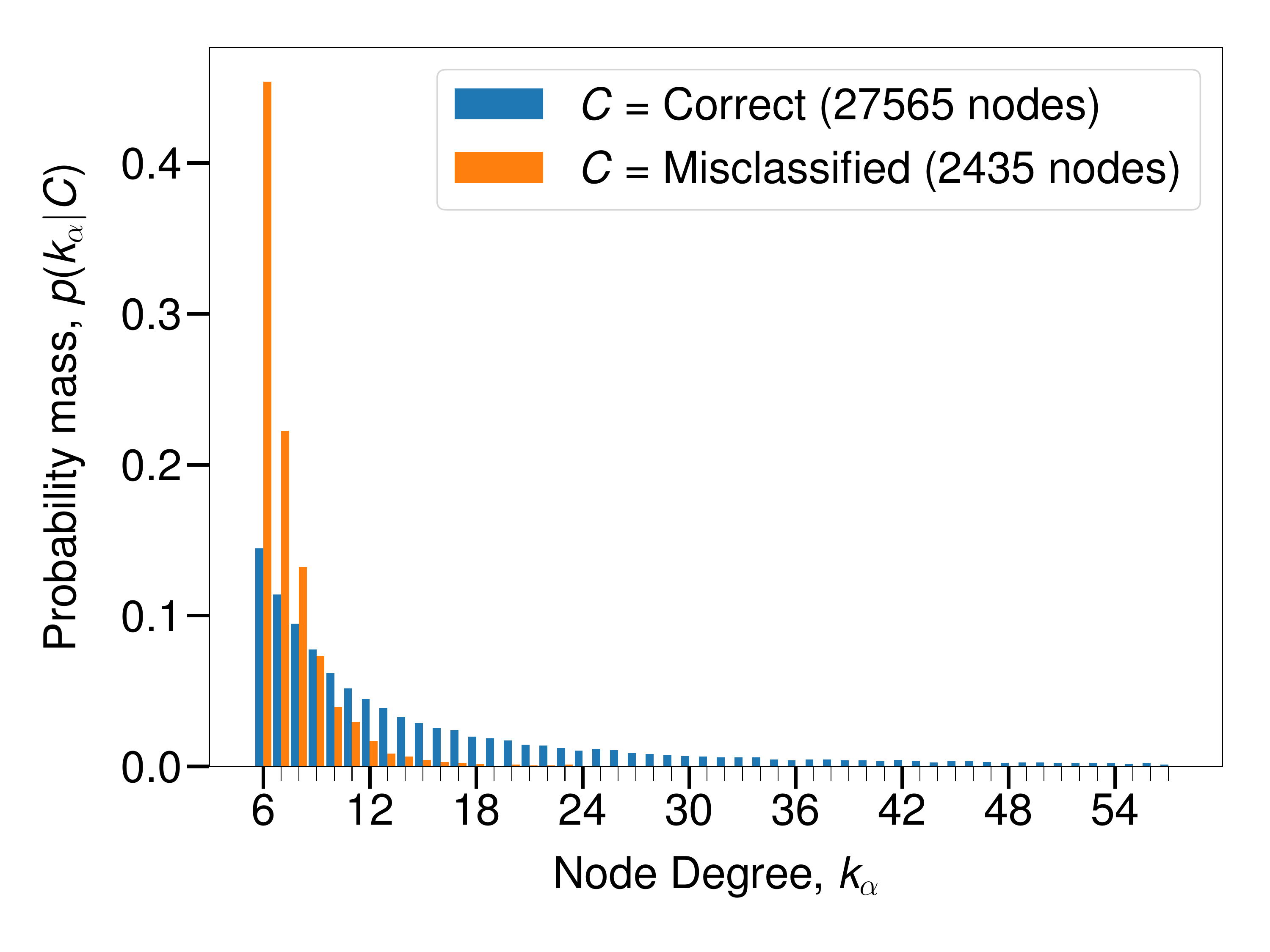}
        \caption{$\kavg = 15$, $n = 300$, $\mu = 0.45$.}
    \end{subfigure}
   \begin{subfigure}{0.32\textwidth}
        \centering
        \includegraphics[width=1.0\textwidth]{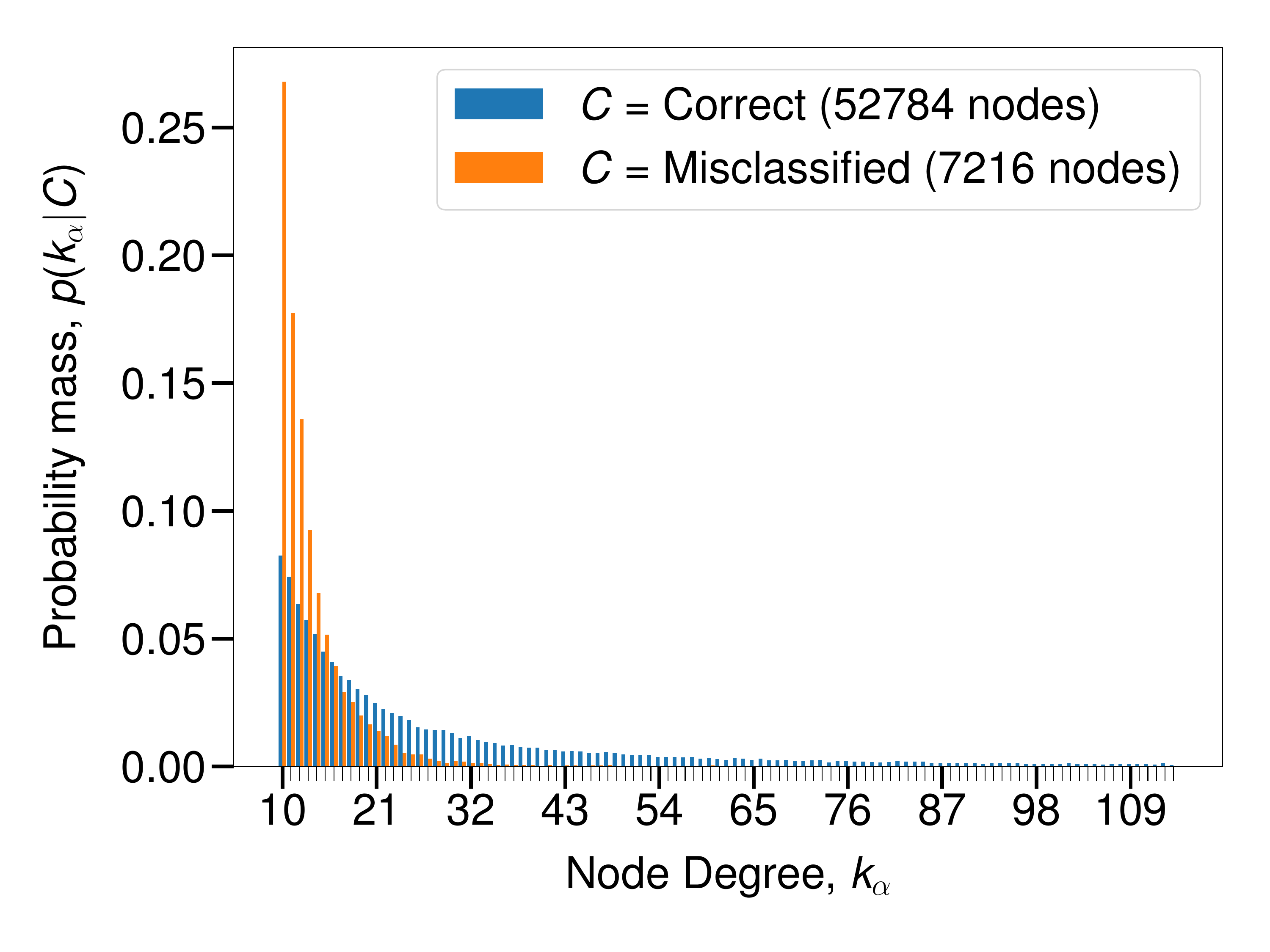}
        \caption{$\kavg = 25$, $n = 600$, $\mu = 0.55$.}
    \end{subfigure}
    \begin{subfigure}{0.32\textwidth}
        \centering
        \includegraphics[width=1.0\textwidth]{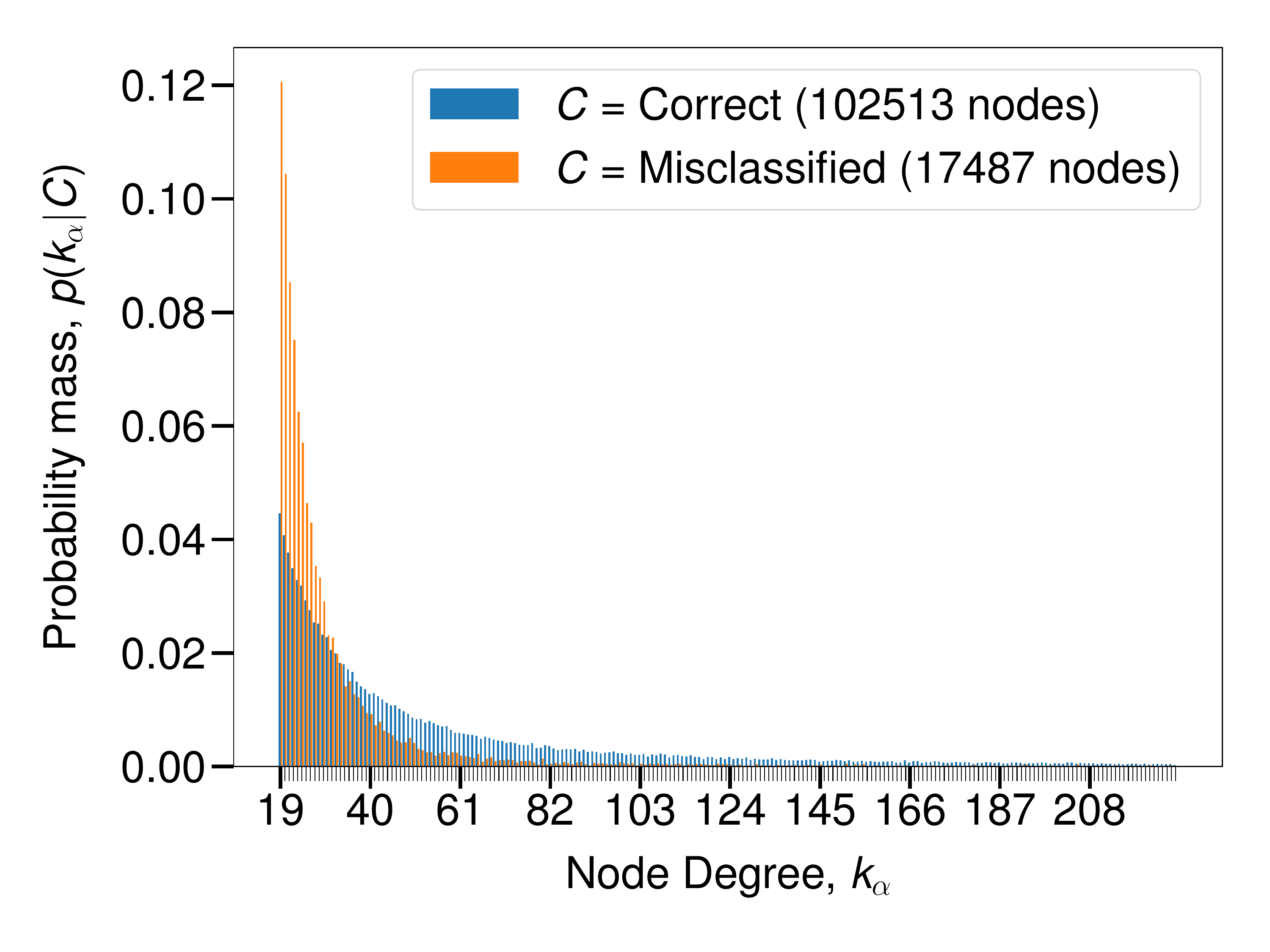}
        \caption{$\kavg = 50$, $n = 1200$, $\mu = 0.63$.}
    \end{subfigure}
    \\
    \begin{subfigure}{0.32\textwidth}
        \centering
        \includegraphics[width=1.0\textwidth]{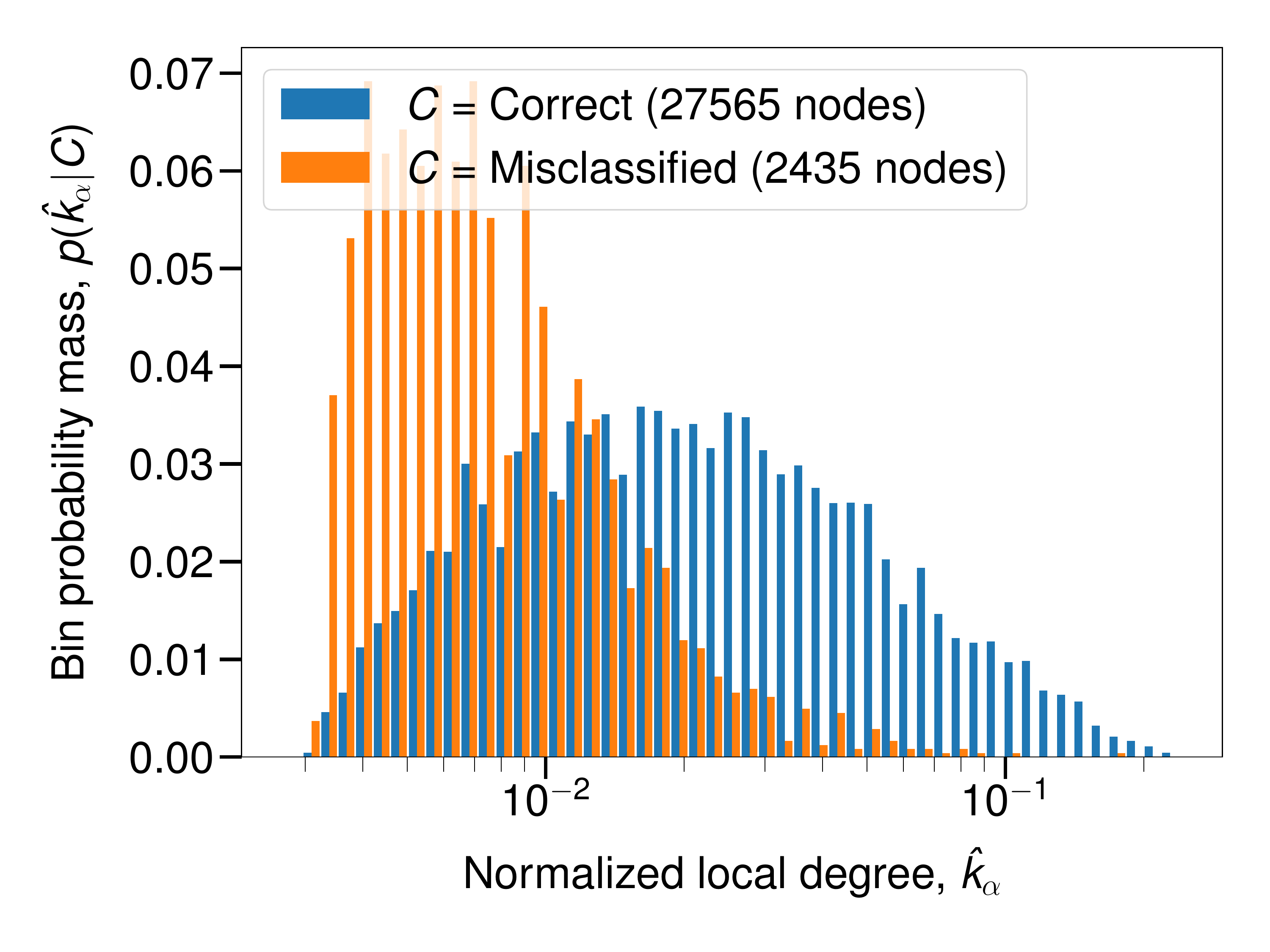}
        \caption{$\kavg = 15$, $n = 300$, $\mu = 0.45$.}
    \end{subfigure}
    \begin{subfigure}{0.32\textwidth}
        \centering
        \includegraphics[width=1.0\textwidth]{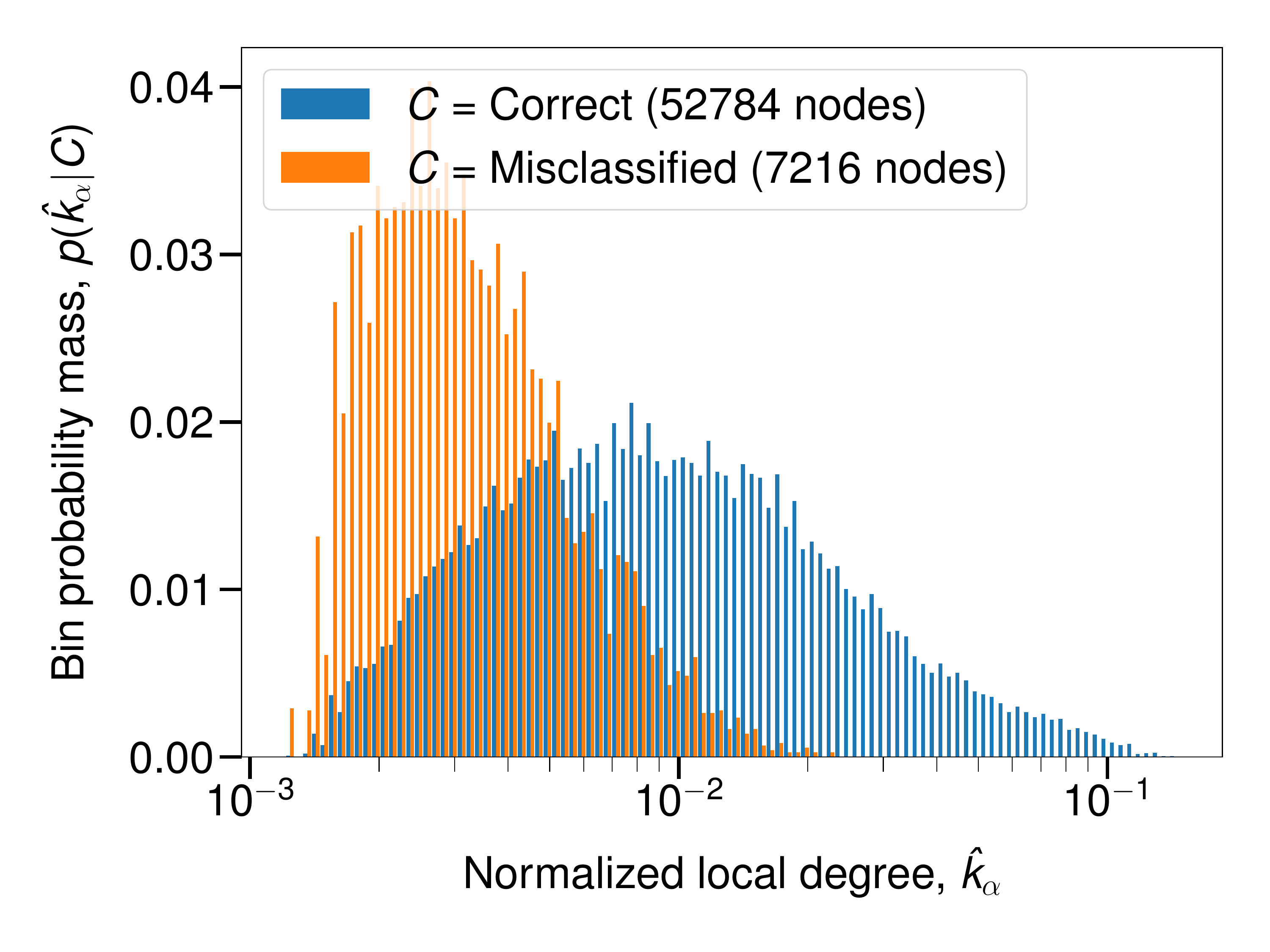}
        \caption{$\kavg = 25$, $n = 600$, $\mu = 0.55$.}
    \end{subfigure}
    \begin{subfigure}{0.32\textwidth}
        \centering
        \includegraphics[width=1.0\textwidth]{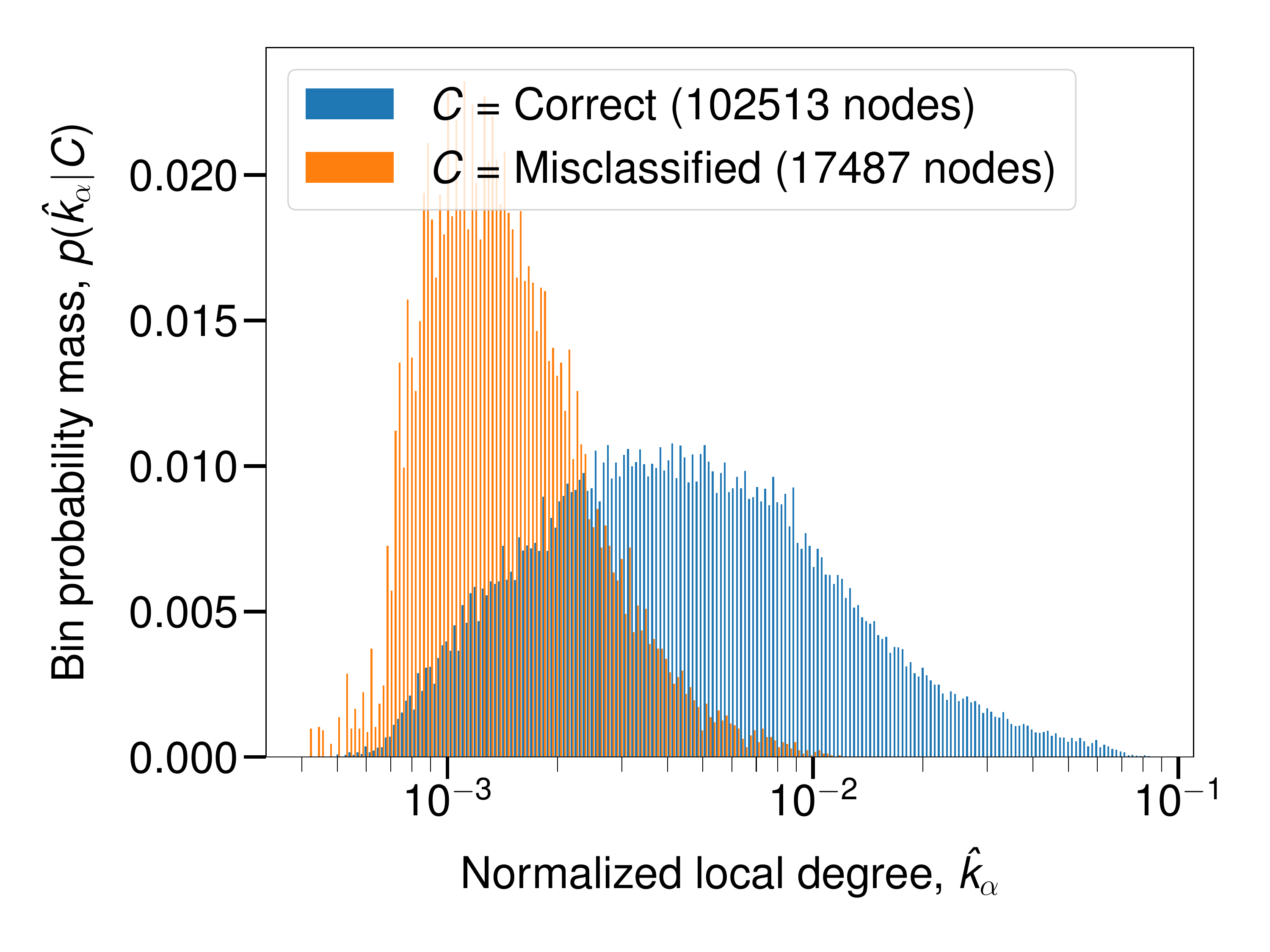}
        \caption{$\kavg = 50$, $n = 1200$, $\mu = 0.63$.}
    \end{subfigure}
    \\
    \begin{subfigure}{0.32\textwidth}
        \centering
        \includegraphics[width=1.0\textwidth]{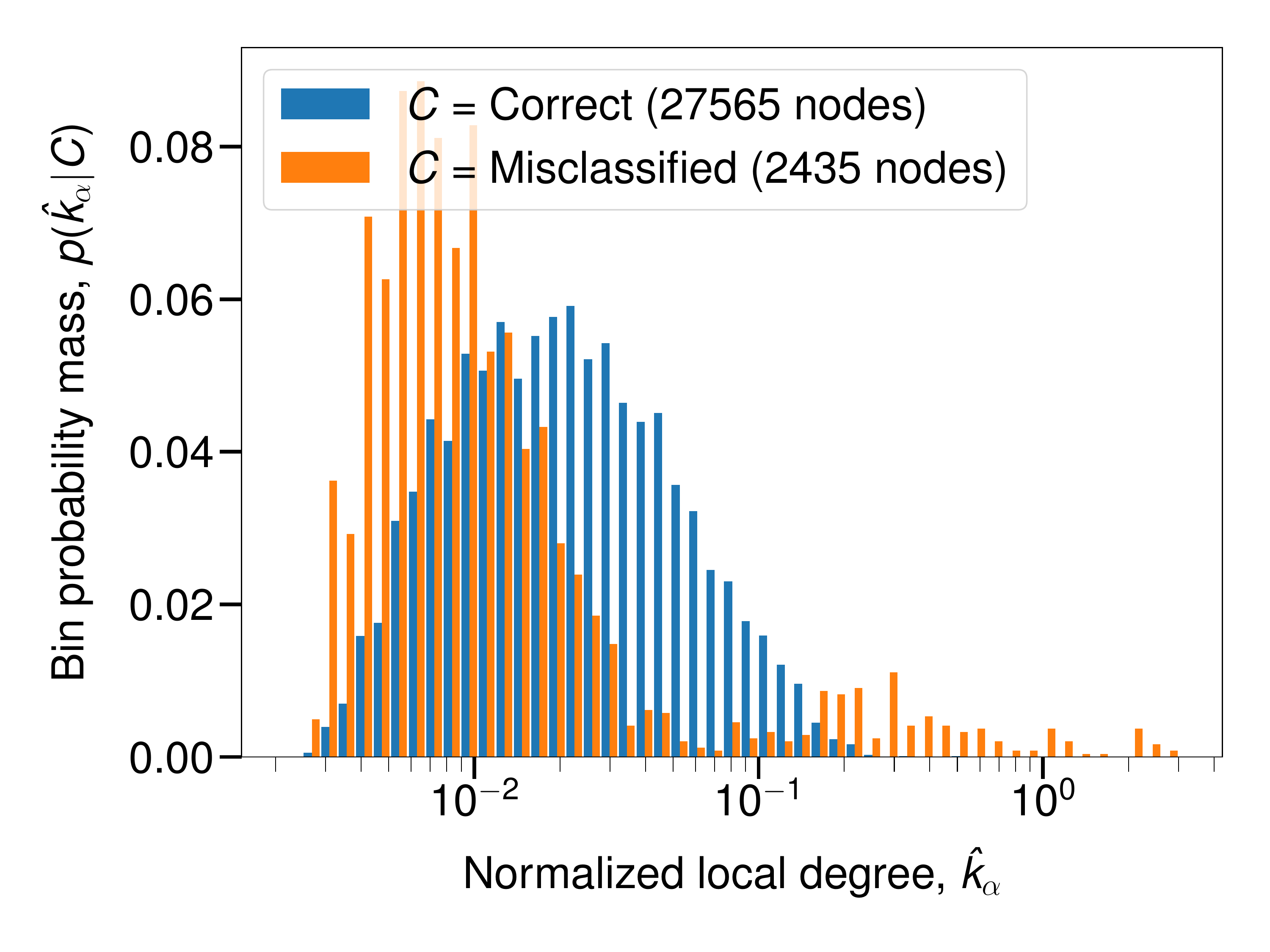}
        \caption{$\kavg = 15$, $n = 300$, $\mu = 0.45$.}
    \end{subfigure}
    \begin{subfigure}{0.32\textwidth}
        \centering
        \includegraphics[width=1.0\textwidth]{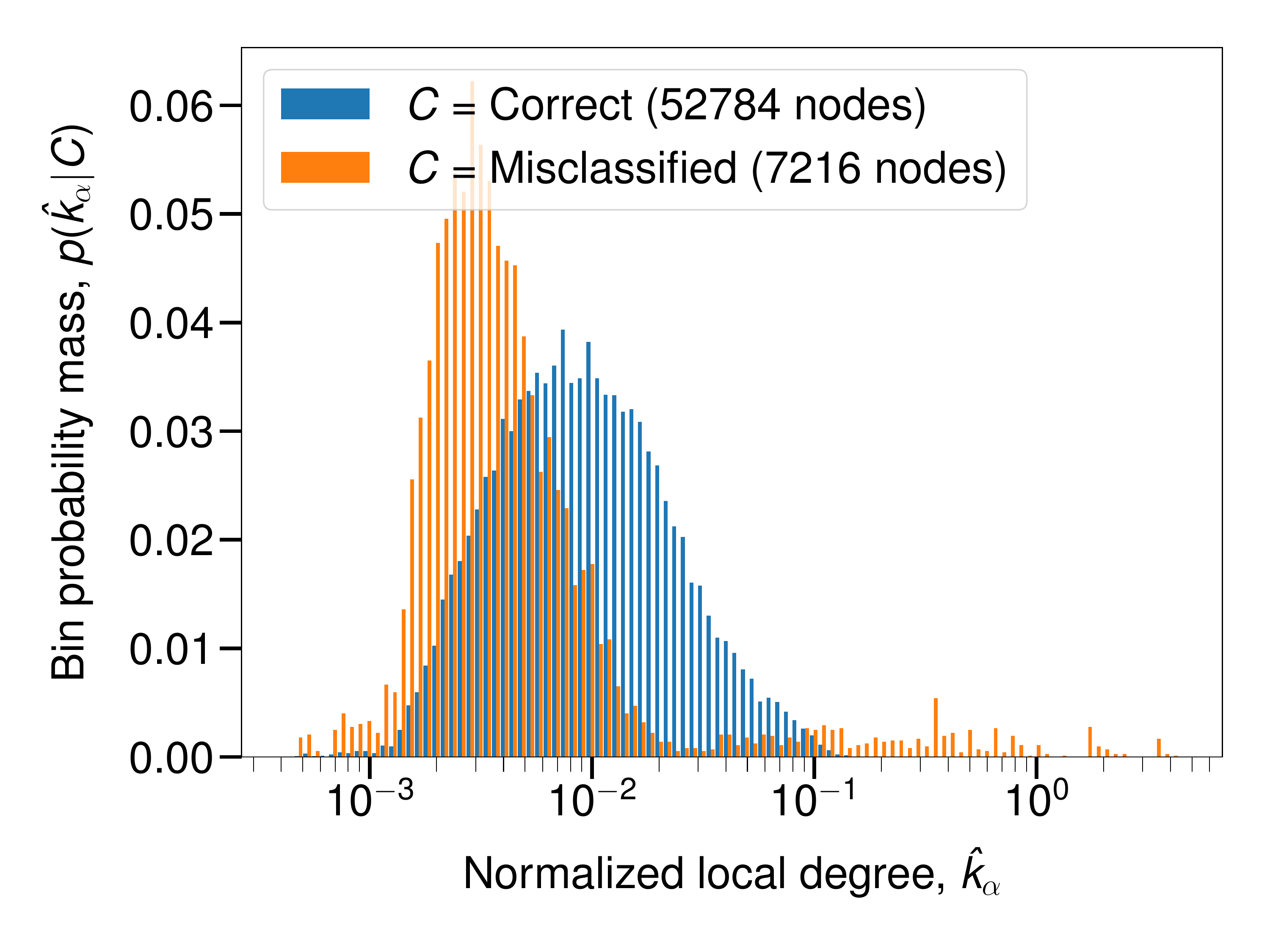}
        \caption{$\kavg = 25$, $n = 600$, $\mu = 0.55$.}
    \end{subfigure}
    \begin{subfigure}{0.32\textwidth}
        \centering
        \includegraphics[width=1.0\textwidth]{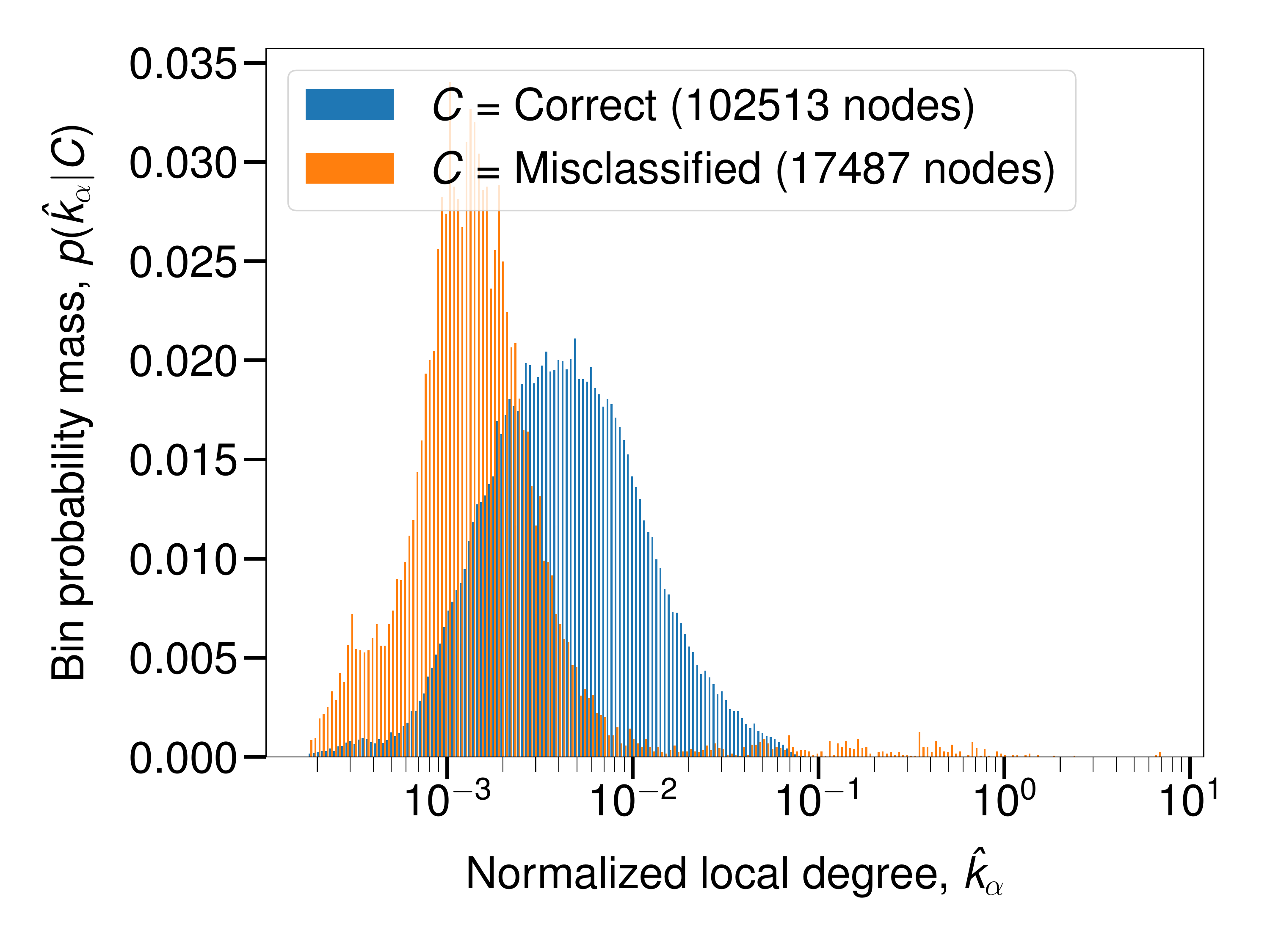}
        \caption{$\kavg = 50$, $n = 1200$, $\mu = 0.63$.}
    \end{subfigure}
    \caption{Degree distributions for correctly classified and misclassified nodes in Walktrap results, obtained from 100 different LFR networks with common average degree, network size and mixing parameter. The top row shows the distributions of the node degrees, the middle row shows the distribution of the normalized local degrees w.r.t.\ the ground truth communities and the bottom row shows the distributions of the normalized local degrees w.r.t.\ the predicted communities. Walktrap tends to misclassify nodes with low normalized local degree (w.r.t.\ the ground truth communities) and/or low absolute node degree. The statistics of the normalized local degrees w.r.t.\ predicted communities resemble the ones w.r.t.\ the ground truth communities.}\label{fig:nodestat_walktrap}
\end{figure*}

The degree and NLD distributions are visualized in Figs.~\ref{fig:nodestat_synwalk} and ~\ref{fig:nodestat_walktrap}. Although in Figs.~\ref{fig:lfr_ami_vs_mu} and~\ref{fig:lfr_ami_vs_n} we saw an apparently strong dependence of the AMI performance on the average degree for Synwalk and Walktrap, for Synwalk a significant dependence is not visible in the class distributions (Fig.~\ref{fig:nodestat_synwalk}, top row). Nevertheless, the distributions of the NLDs w.r.t.\ the ground truth communities (Fig.~\ref{fig:nodestat_synwalk}, middle row) reveal that misclassified nodes are more likely to exhibit a low NLD than correctly classified ones.

In contrast, although the latter observation holds for Walktrap as well (Fig.~\ref{fig:nodestat_walktrap}, middle row), the node degrees of its misclassified nodes appear to be smaller than those of correctly classified nodes (Fig.~\ref{fig:nodestat_walktrap}, top row). This behavior appears plausible when considering the mechanics of Walktrap: nodes are grouped based on cluster/node distances that are computed by considering random walks of a specified length $T$ (in our setup $T = 4$). For low values of $T$, low-degree nodes are rarely visited, resulting in frequent ties in distance calculations, whereas in the limit of $T\rightarrow\infty$ the distances are determined by the proportionality of the stationary distribution to the node degrees. Hence, it is necessary to make a trade-off regarding the random walk length $T$, which is typically chosen heuristically.

These differing properties of Synwalk and Walktrap manifest in contrasting detection behaviors on the LFR networks (cf. Fig.~\ref{fig:lfr_sample}). Synwalk identifies smaller communities with greater accuracy than larger ones (dependence on the normalized local degree), i.e., the majority of misclassified nodes occur in the largest communities. While Walktrap follows this trend, misclassified nodes occur in smaller communities with increasing frequency (stronger dependence on node degree).

Another interesting difference appears when inspecting to which clusters misclassified nodes are assigned. Synwalk tends to place misclassified nodes in additional (i.e., clusters with no matching ground truth community), small clusters. Such behavior is indicated by the NLD distributions w.r.t.\ predicted communities as well, where misclassified nodes exhibit a significantly higher NLD than correctly classified ones (cp. Fig.~\ref{fig:nodestat_synwalk}, bottom row). This results in detected ground truth communities being \enquote{pure}, i.e., they do not contain nodes from other ground truth communities.
In contrast, Walktrap mainly confuses node memberships within clusters that do have a matching ground truth community. Again, these observations are supported by the NLD distributions w.r.t predicted communities as well, where misclassified nodes exhibit a similar NLD to correctly classified ones (cp. Fig.~\ref{fig:nodestat_walktrap}, bottom row). 

These behavioral difference between Synwalk and Walktrap are visible in the AMI performance as well: whereas both methods misclassify approximately the same amount of nodes in the sample network in Fig.~\ref{fig:lfr_sample}, Synwalk achieves a significantly higher AMI value.

\begin{figure*}[t]
    \centering
    \begin{subfigure}{0.45\textwidth}
        \centering
        \includegraphics[width=1.0\textwidth]{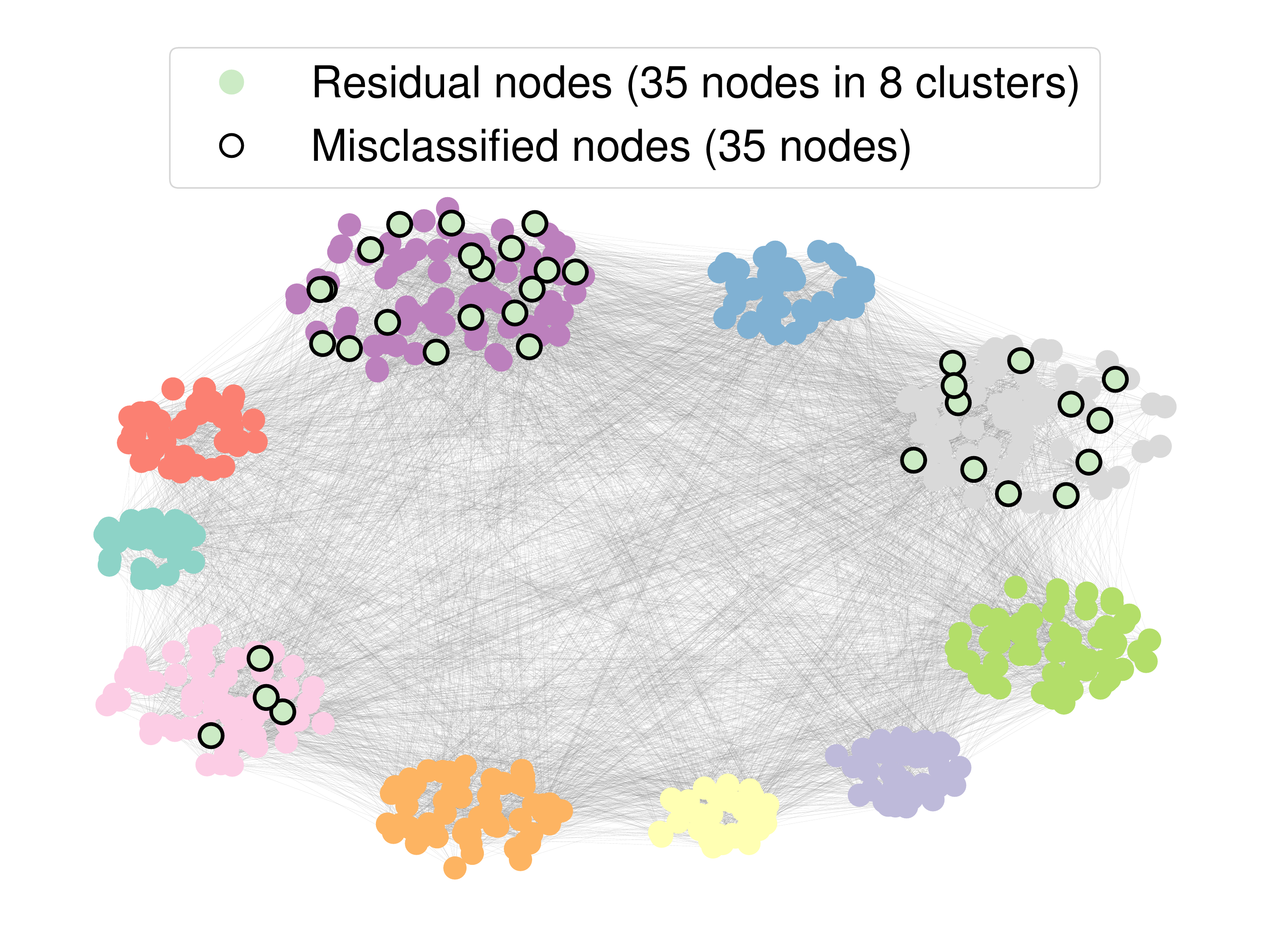}
        \caption{Synwalk, $\ami{\dom{Y}}{\dom{Y}_{true}} = 0.94$.}
    \end{subfigure}
    \begin{subfigure}{0.45\textwidth}
        \centering
        \includegraphics[width=1.0\textwidth]{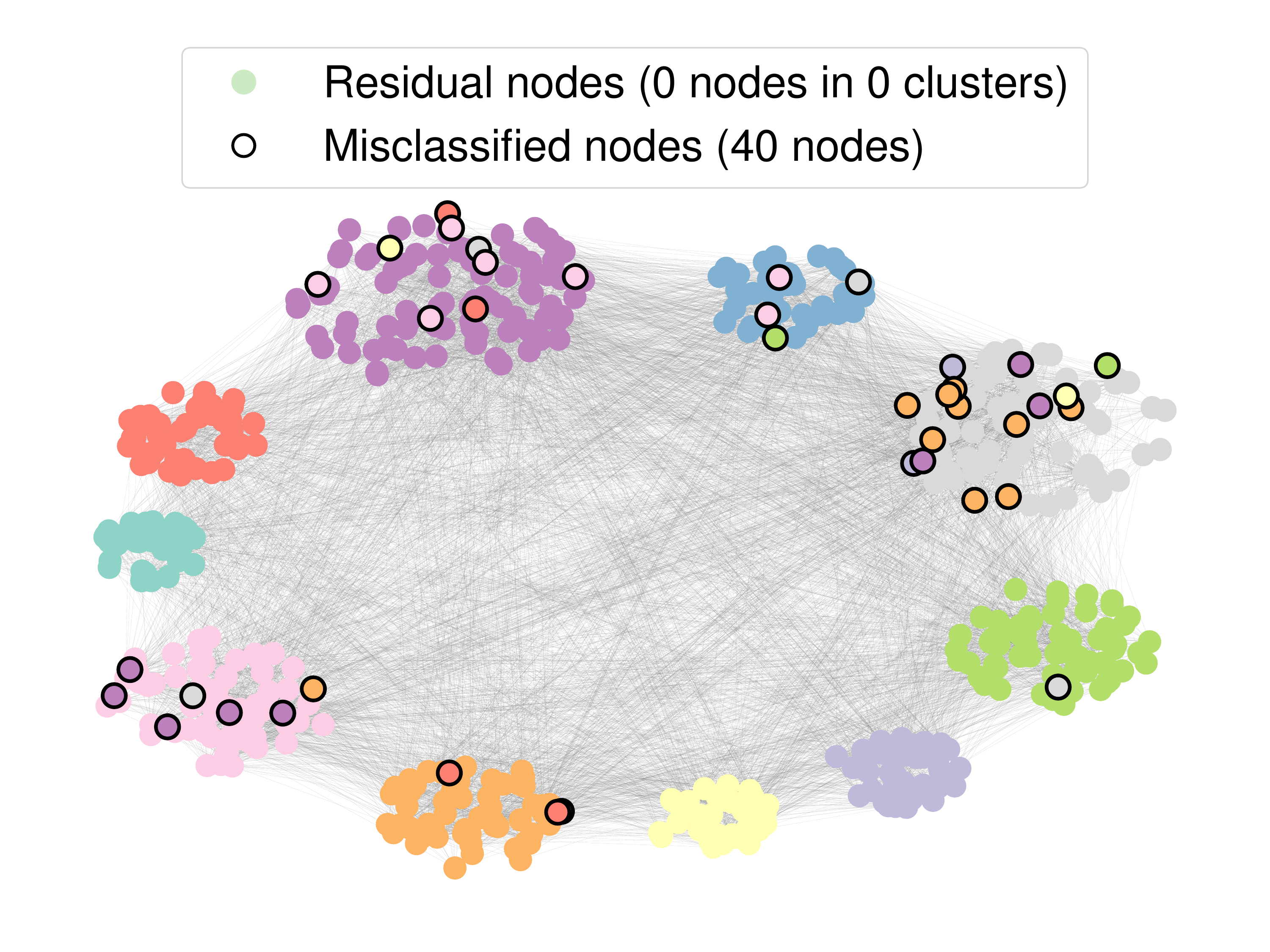}
        \caption{Walktrap, $\ami{\dom{Y}}{\dom{Y}_{true}} = 0.87$.}
    \end{subfigure}
    \caption{A sample LFR graph with communities as detected by Synwalk and Walktrap. The network has $N=600$ nodes, an average degree of $\kavg = 25$ and a mixing parameter of $\mu = 0.55$. Nodes are grouped according to their ground truth communities and share the same color if they belong to the same detected cluster. Misclassified nodes are highlighted with a black border. We aggregated nodes from predicted clusters that have no matching ground truth community into a single residual cluster. Synwalk places misclassified nodes into additional clusters, whereas Walktrap confuses node memberships between ground truth communities.}\label{fig:lfr_sample}
\end{figure*}

The above insights make apparent two advantages of our method. First, whereas there is no general answer on how to determine the random walk length $T$ for Walktrap, Synwalk does not require the tuning of any hyper-parameter. Secondly, consider a network with many small communities and low average degree. Following our earlier observations, Walktrap will have many misclassified nodes due to the low average degree. In contrast, the assumption of small communities implies a reasonably high normalized local degree for the majority of nodes and thus suggests a better performance of Synwalk when compared to Walktrap. Indeed, the results in Fig.~\ref{fig:lfr_set_B} support this intuition. The benchmark networks underlying these results were generated with parameter set B (see Table~\ref{tab:lfr}), effectively lowering the average community size for a given average degree compared to networks generated with parameter set A.

\begin{figure*}[t]
    \centering
    \begin{subfigure}{0.32\textwidth}
        \centering
        \includegraphics[width=1.0\textwidth]{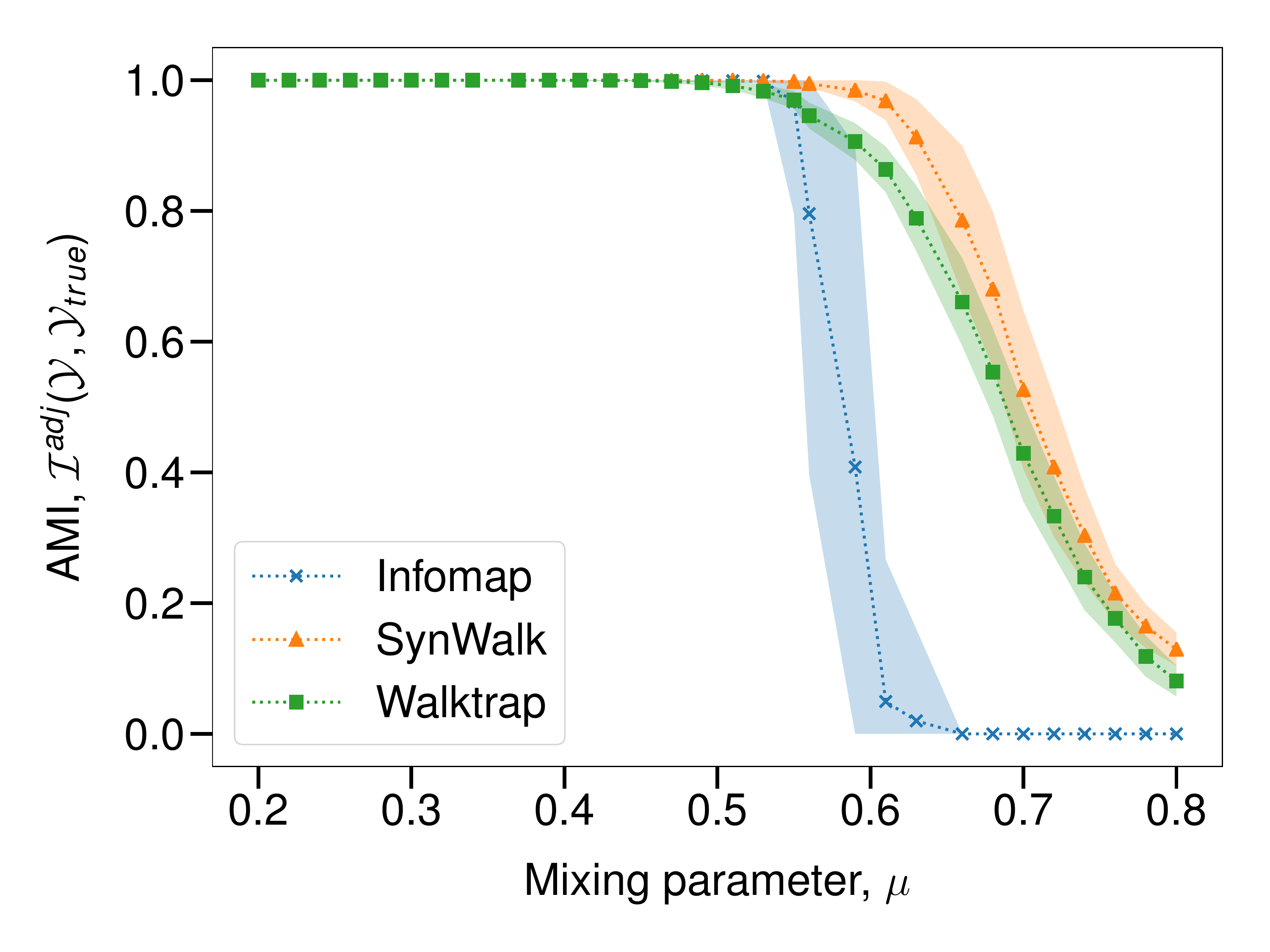}
        \caption{$\kavg = 20$, $N = 600$.}
    \end{subfigure}
    \begin{subfigure}{0.32\textwidth}
        \centering
        \includegraphics[width=1.0\textwidth]{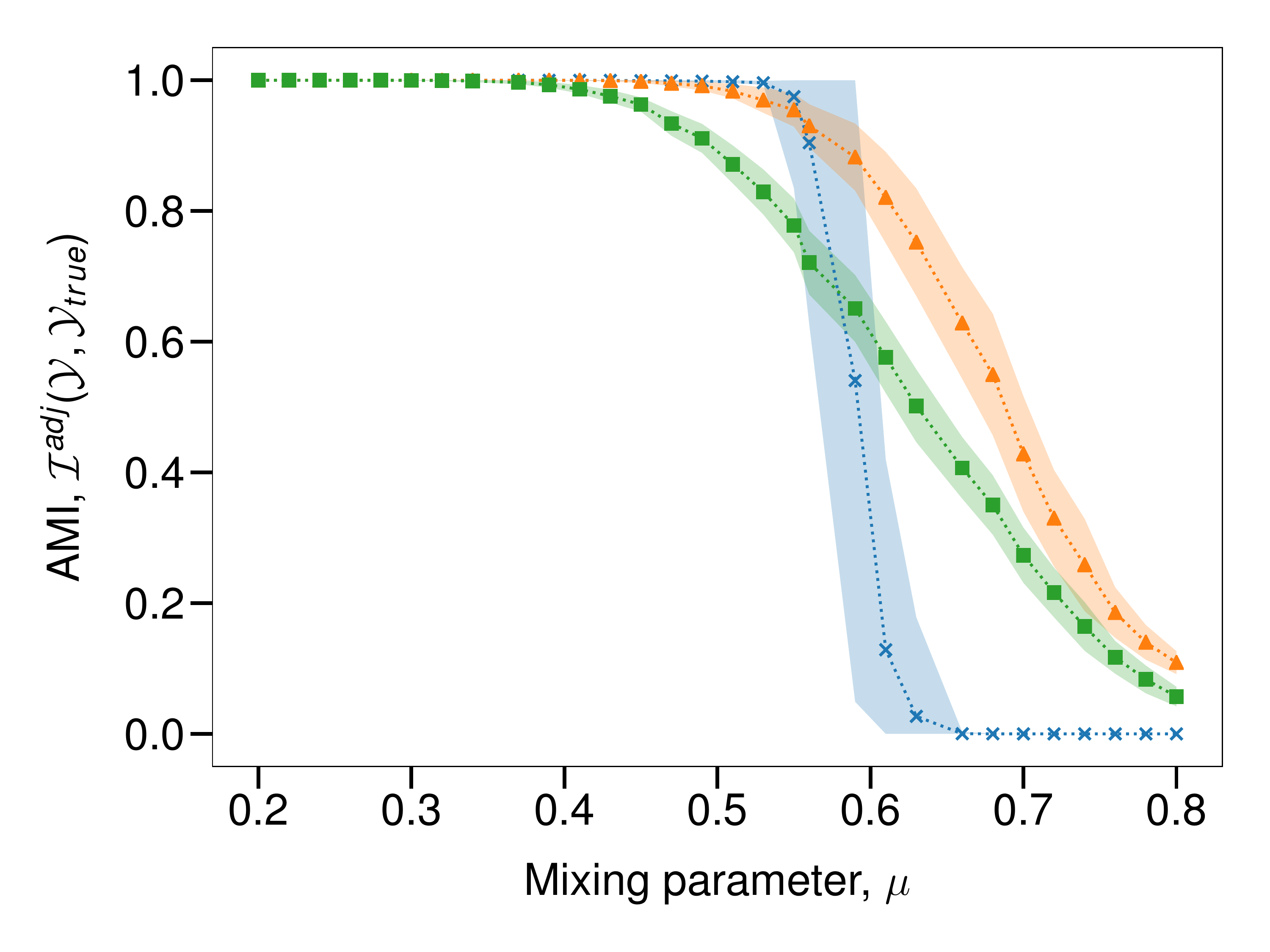}
        \caption{$\kavg = 20$, $N = 1200$.}
    \end{subfigure}
    \begin{subfigure}{0.32\textwidth}
        \centering
        \includegraphics[width=1.0\textwidth]{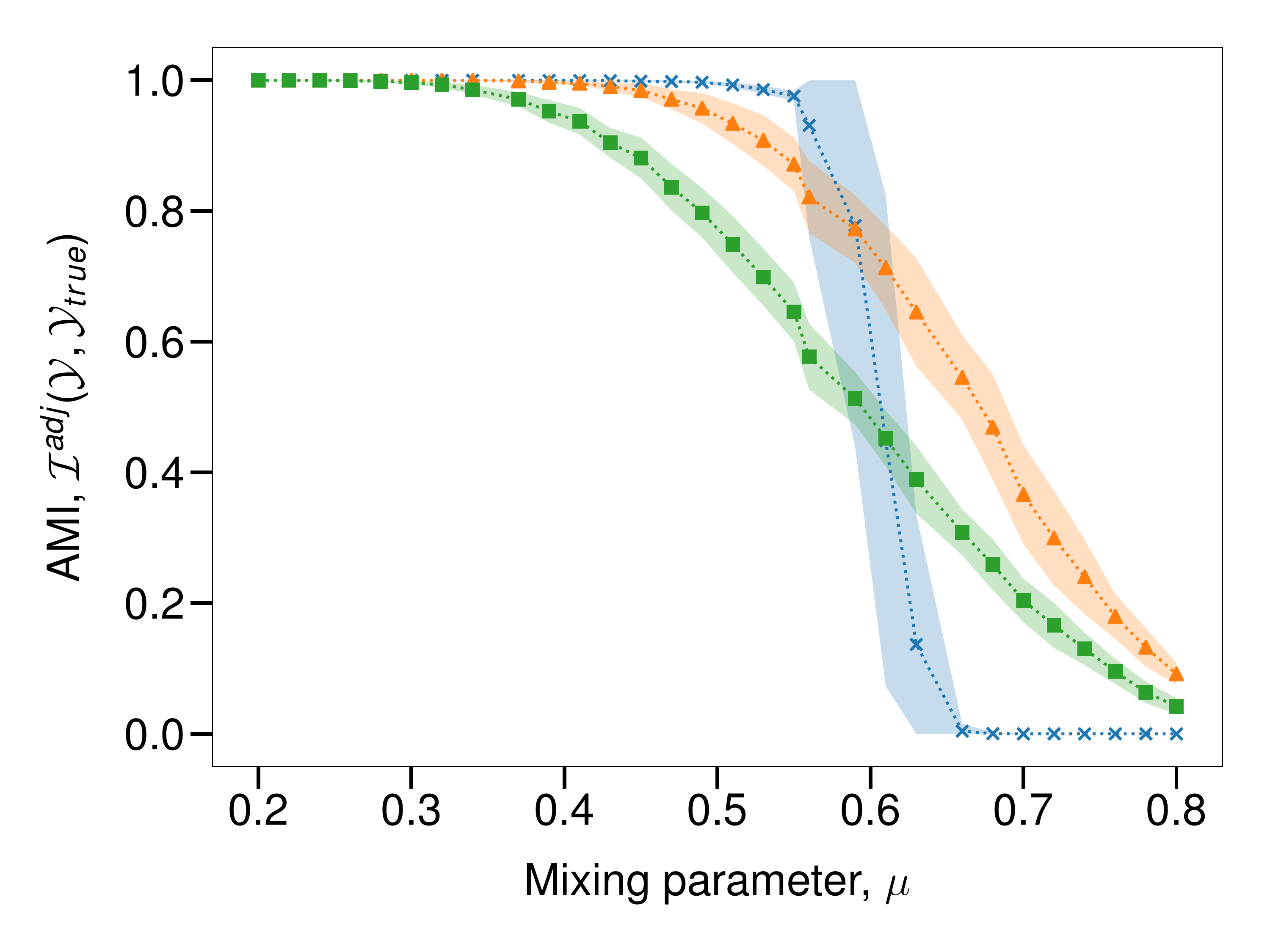}
        \caption{$\kavg = 20$, $N = 2400$.}
    \end{subfigure}
    \caption{Comparison of Infomap, Synwalk and Walktrap on LFR benchmark networks with given average degree and network size. 
    The lines and shaded areas show the mean and standard deviation of AMI as a function of themixing parameter, obtained from $100$ different network realizations. The networks were generated with parameter set B (see Table~\ref{tab:lfr}), simulating smaller communities with higher normalized local degrees. Synwalk outperforms Walktrap in this setup.}\label{fig:lfr_set_B}
\end{figure*}

	\begin{table}[!bh]
\centering
\caption{Properties of the examined real-world networks.} \label{tab:networks}
\resizebox{\textwidth}{!}{%
    \sisetup{table-number-alignment=center,table-figures-integer=7,table-figures-decimal=0}
    \begin{tabular}{lSSS[table-format=2.2]c}
    \toprule
    Network             & {Nodes}   & {Links}   & {$\bar{k}$}  & Source          \\ \midrule
    dblp                &  317080   & 1049866   &  6.62     & \citep{Yang2015, snapnets} \\ 
    facebook            &   22470   &  170823   & 15.20     & \citep{Rozemberczki2019, snapnets} \\ 
    github              &   37700   &  289003   & 15.33     & \citep{Rozemberczki2019, snapnets} \\ 
    lastfm-asia         &    7624   &   27806   &  7.29     & \citep{Rozemberczki2020, snapnets} \\ 
    pennsylvania-roads  & 1088092   & 1541898   &  2.83     & \citep{Leskovec2008, snapnets} \\ 
    wordnet             &  146005    & 656999   &  9.00     & \citep{Miller1998, konect} \\
    \bottomrule
    \end{tabular}%
    }
\end{table}

\section{Illustration on Empirical Networks}\label{sec:realworld}

In this section we illustrate the applicability of Synwalk on a selection of empirical networks (see Table~\ref{tab:networks}) by comparing the detection results of Synwalk, Infomap and Walktrap. For this purpose, we report single-number properties of their detected clusterings in Table~\ref{tab:cluster_numbers}. Synwalk and Infomap behave similarly in terms of their single-number characteristics. An exception to this observation is the github network, where Synwalk detects a greater number of (non-trivial) clusters. Notably, Walktrap results consistently show a higher fraction of trivial clusters when compared to the other two methods.

Additionally, we look at the distributions of cluster sizes in Fig.~\ref{fig:empirical_sizes} and normalized local degrees in Fig.~\ref{fig:empirical_nlds}. The results for further cluster and node properties are given in Appendix~\ref{appx:additional_empirical} for the sake of completeness. For all distribution plots we consider clusters with less than three members as trivial and we do not include their statistics in the distributions. Infomap and Synwalk again behave similar given their cluster and node property distributions. A deviation from this pattern is apparent in the distribution of NLDs for the github network, where Synwalk exhibits higher NLDs when compared to Infomap and Walktrap. The cluster size distributions of Walktrap show a trend towards small clusters. As the empirical networks under consideration are significantly larger than the examined LFR networks in Section~\ref{sec:lfr_experiments}, a random walk length of $T = 4$ might not be the optimal hyperparameter choice and thus could explain the many trivial clusters detected (cp. Table~\ref{tab:cluster_numbers}).

Interestingly, whereas Synwalk achieved similar AMI performance as Walktrap in Section~\ref{sec:lfr_experiments}, Synwalk shows similar qualitative behavior to Infomap regarding cluster and node property statistics on empirical networks. However, the differences in the qualitative detection behavior of Synwalk and Walktrap that we discussed in Section~\ref{sssec:node_statistics} could explain the different results on the larger empirical networks. We further conjecture that the common search heuristic (see Section~\ref{sec:lfr_experiments}) of Synwalk and Infomap acts as a \enquote{regularizer} on larger networks, i.e., the properties of predicted clusterings become more similar the larger the networks.

\begin{table}[b]
\centering
\caption{Single-number characteristics of the detection results of Infomap, Synwalk and Walktrap on the examined empirical networks. We list the number of detected clusters, the number of non-trivial (less than three nodes) clusters including the fraction thereof in brackets, and the modularity score for each clustering. Synwalk produces similar figures to Infomap in terms of non-trivial clusters and modularity, except for the github network. Walktrap predicts a higher fraction of trivial clusters compared to the other methods.} \label{tab:cluster_numbers}
\resizebox{\textwidth}{!}{%
    \sisetup{table-number-alignment=center,table-figures-integer=1,table-figures-decimal=2, table-align-uncertainty=true}
    \begin{tabular}{l|S[table-format=5]S[table-format=5]S[table-format=5]|S[table-format=5]@{\,}cS[table-format=5]@{\,}cS[table-format=5]@{\,}c|SSS}
    \toprule
     & \multicolumn{3}{c|}{Detected Clusters} & \multicolumn{6}{c|}{Non-trivial Clusters} & \multicolumn{3}{c}{Modularity} \\
     Network  & {Infomap} & {Synwalk} & {Walktrap} & \multicolumn{2}{c}{Infomap} & \multicolumn{2}{c}{Synwalk} & \multicolumn{2}{c|}{Walktrap} & {Infomap}  & {Synwalk}  & {Walktrap} \\ \midrule
    dblp                & 14544 & 14587 & 30425 & 14533 & (1.00) & 14576 & (1.00) & 25659 & (0.84) & 0.73 & 0.73 & 0.67 \\
    facebook            &   860 &   837 &  1227 &   813 & (0.95) &   807 & (0.96) &   834 & (0.68) & 0.76 & 0.74 & 0.75 \\
    github              &  1644 &  3935 &  7161 &  1333 & (0.81) &  2228 & (0.57) &   630 & (0.09) & 0.39 & 0.27 & 0.36 \\
    lastfm-asia         &   422 &   445 &   328 &   374 & (0.87) &   411 & (0.92) &   202 & (0.62) & 0.74 & 0.70 & 0.77 \\
    pennsylvania-roads  & 47174 & 49095 & 15101 & 47004 & (1.00) & 48925 & (1.00) & 14931 & (0.99) & 0.90 & 0.89 & 0.95 \\
    wordnet             &  5786 &  6058 & 14063 &  5604 & (0.97) &  5861 & (0.97) &  9531 & (0.68) & 0.66 & 0.65 & 0.60 \\
    \bottomrule
    \end{tabular}
}
\end{table}

\begin{figure*}[t]
    \centering
    \begin{subfigure}{0.32\textwidth}
        \centering
        \includegraphics[width=1.0\textwidth]{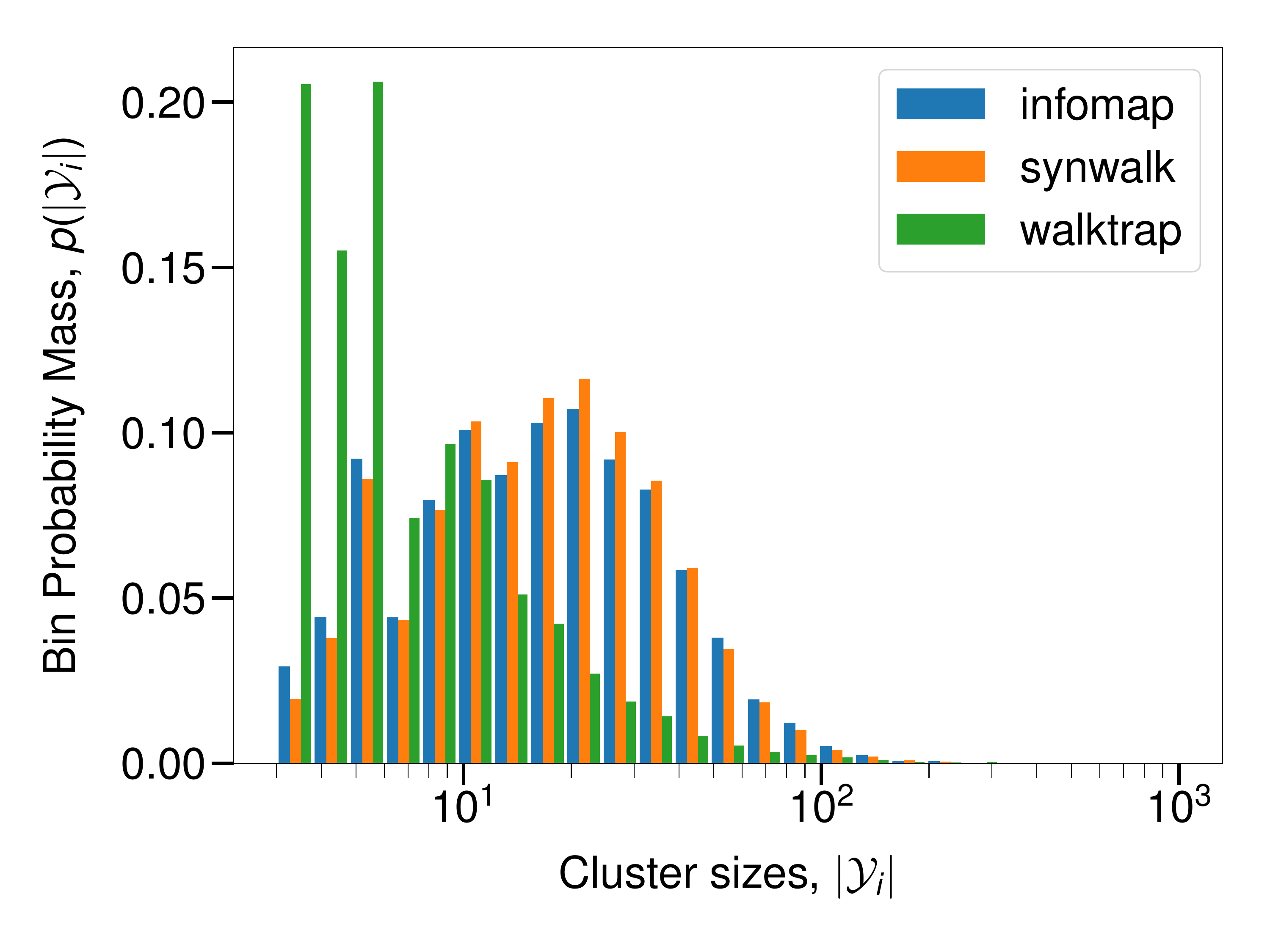}
        \caption{dblp.}
    \end{subfigure}
   \begin{subfigure}{0.32\textwidth}
        \centering
        \includegraphics[width=1.0\textwidth]{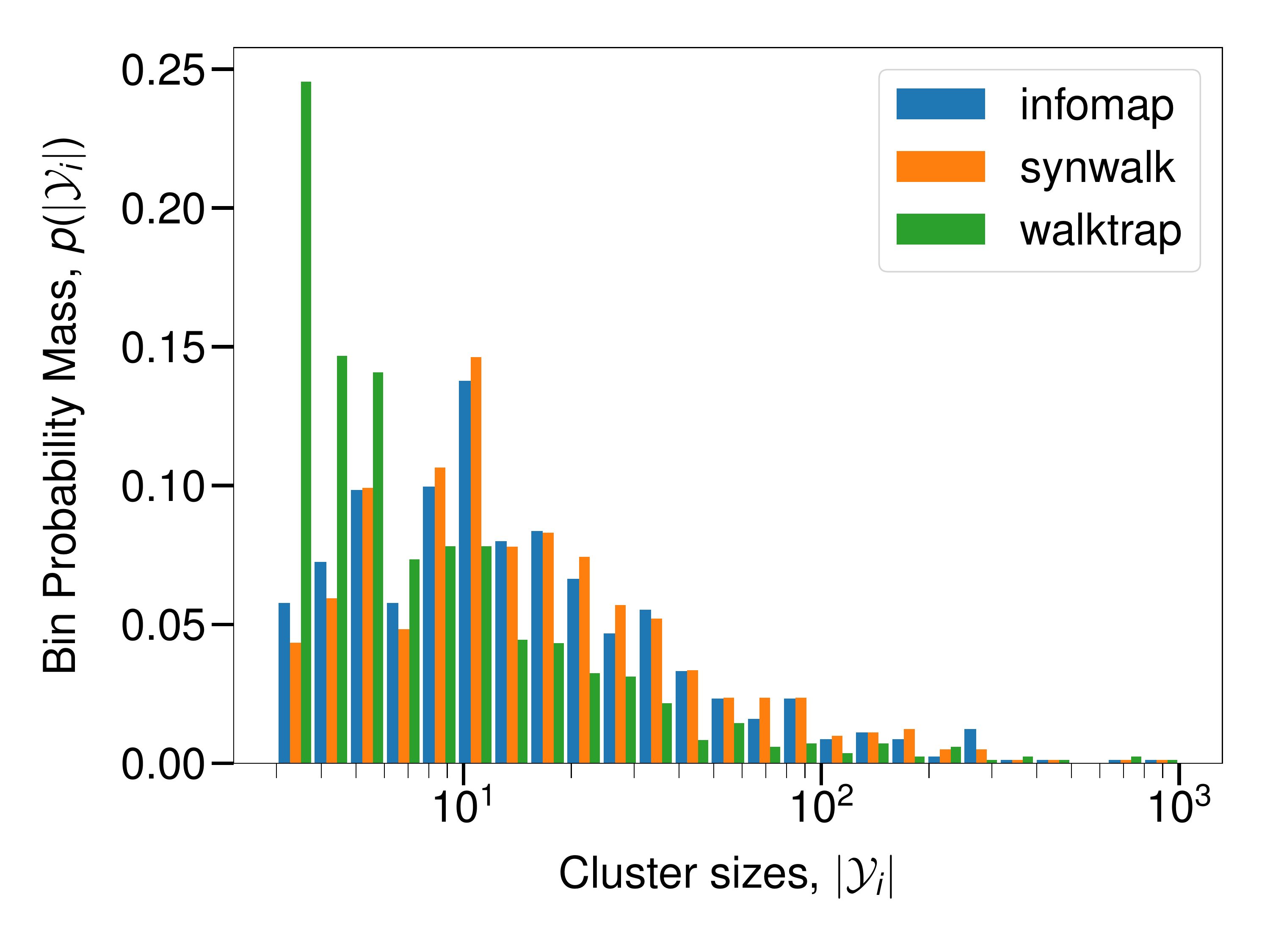}
        \caption{facebook.}
    \end{subfigure}
    \begin{subfigure}{0.32\textwidth}
        \centering
        \includegraphics[width=1.0\textwidth]{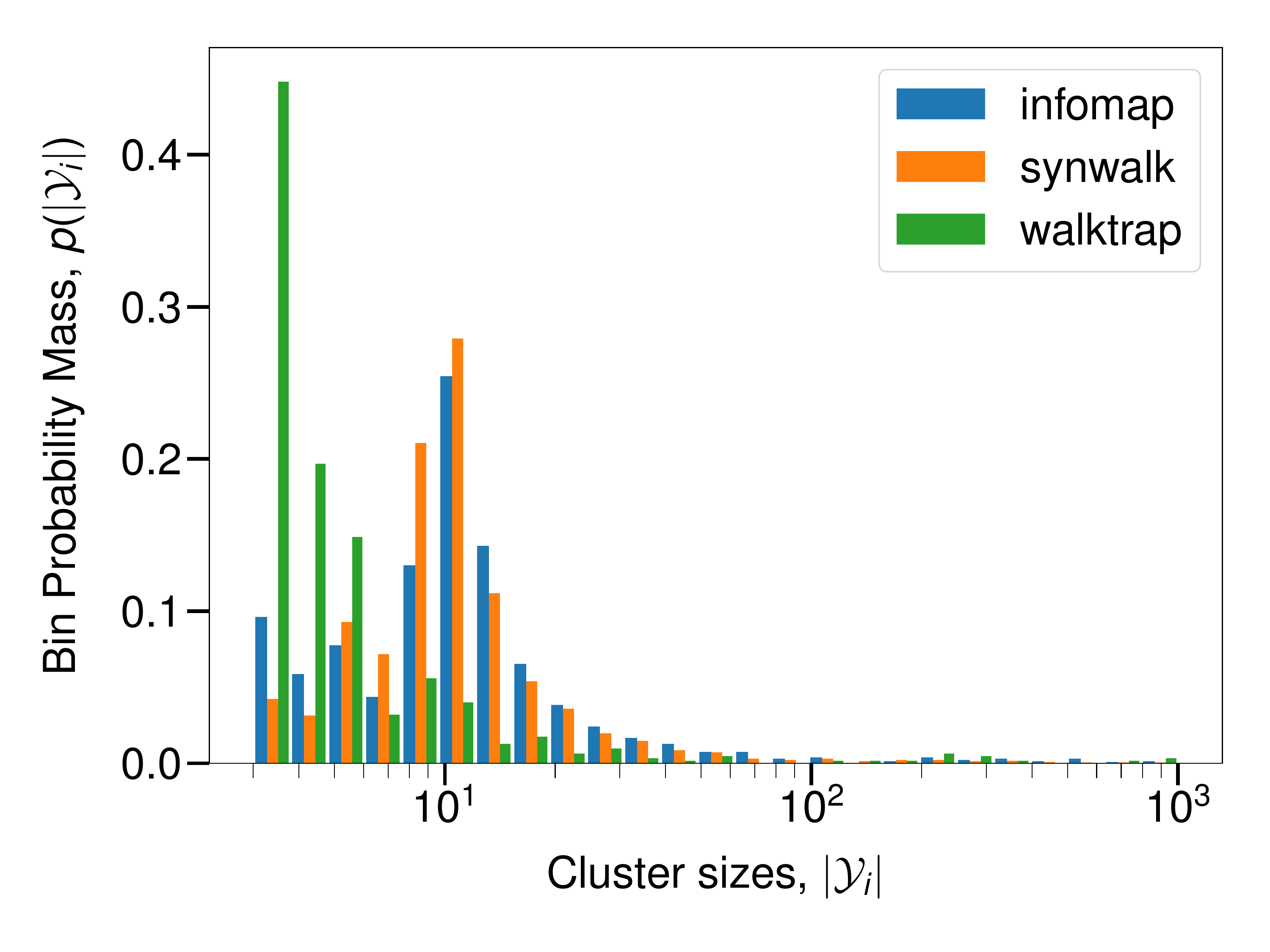}
        \caption{github.}
    \end{subfigure}
    \\
    \begin{subfigure}{0.32\textwidth}
        \centering
        \includegraphics[width=1.0\textwidth]{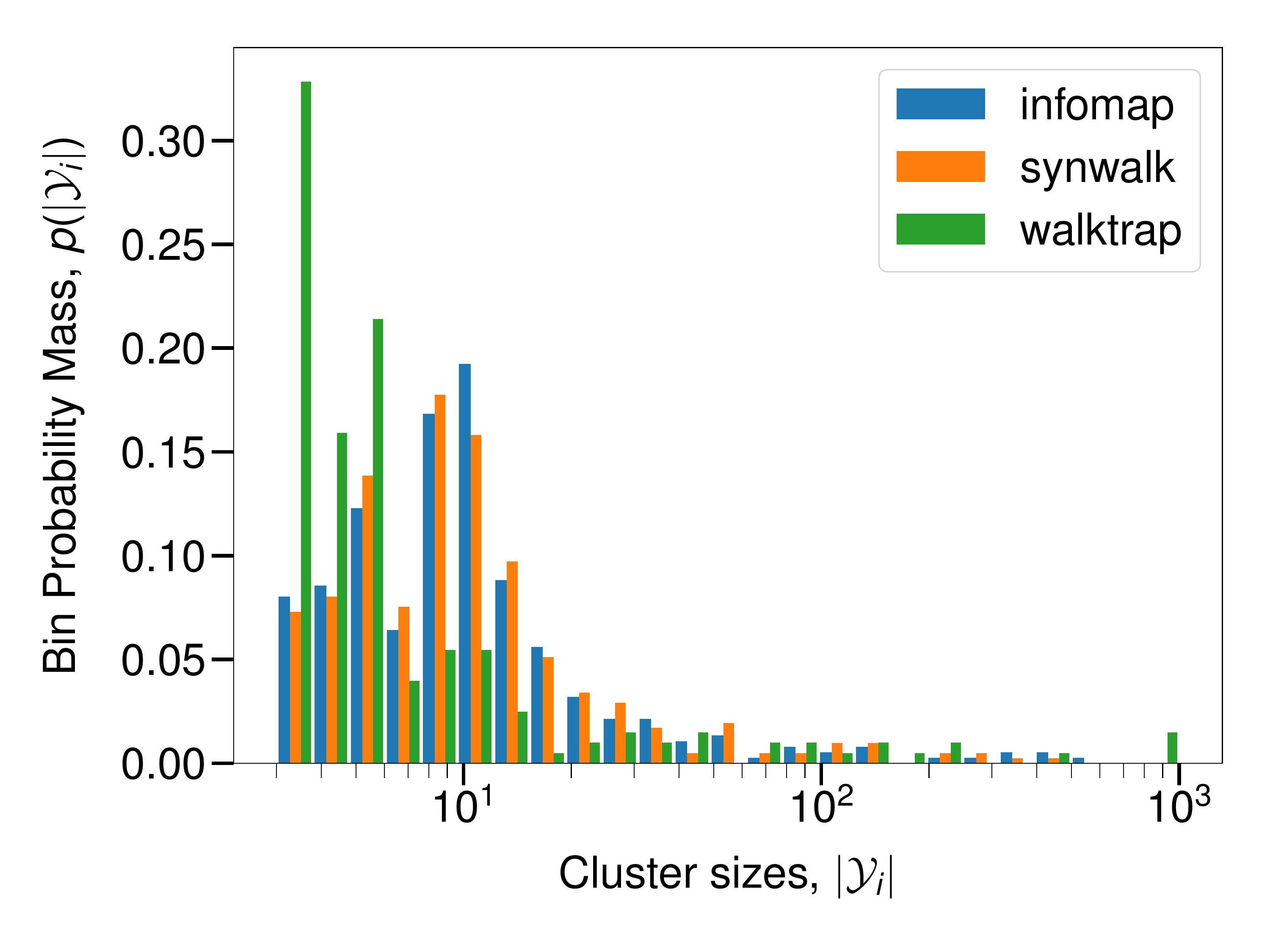}
        \caption{lastfm-asia.}
    \end{subfigure}
    \begin{subfigure}{0.32\textwidth}
        \centering
        \includegraphics[width=1.0\textwidth]{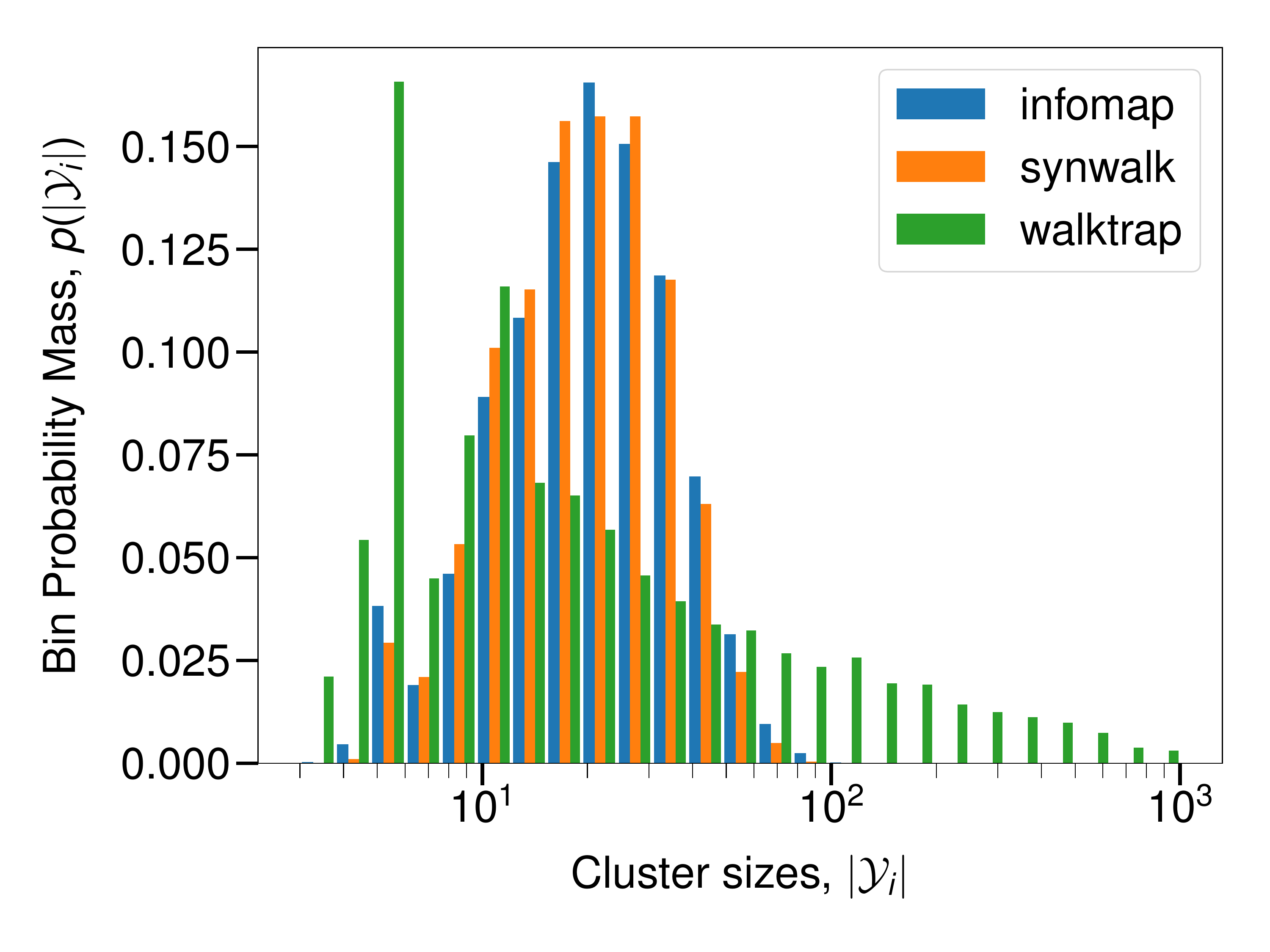}
        \caption{pennsylvania-roads.}
    \end{subfigure}
    \begin{subfigure}{0.32\textwidth}
        \centering
        \includegraphics[width=1.0\textwidth]{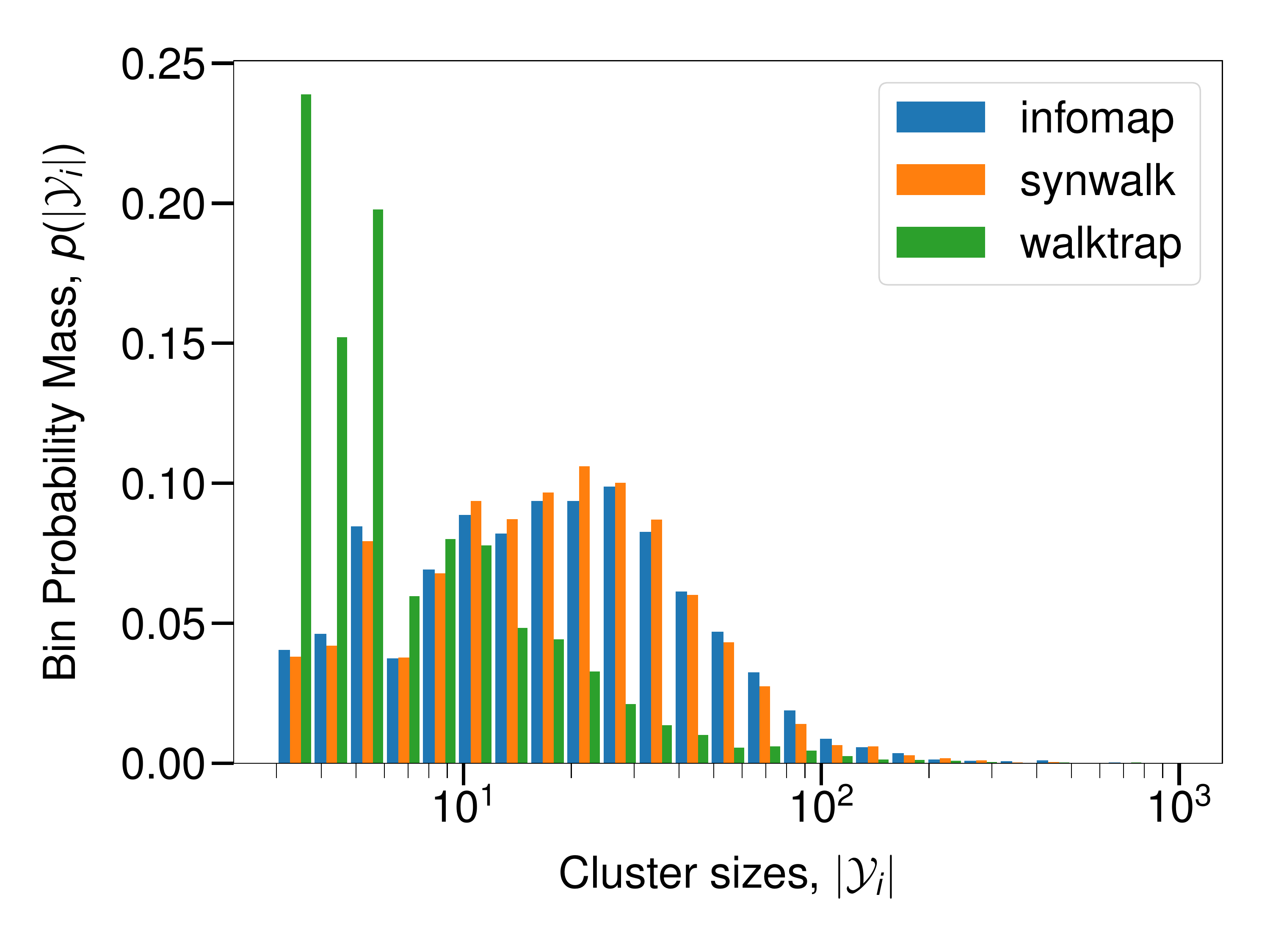}
        \caption{wordnet.}
    \end{subfigure}
    \caption{Distributions of cluster sizes for the detection results on empirical networks. Synwalk produces similar statistics to Infomap that are clearly distinguishable from Walktrap's results.}\label{fig:empirical_sizes}
\end{figure*}

\begin{figure*}[t]
    \centering
    \begin{subfigure}{0.32\textwidth}
        \centering
        \includegraphics[width=1.0\textwidth]{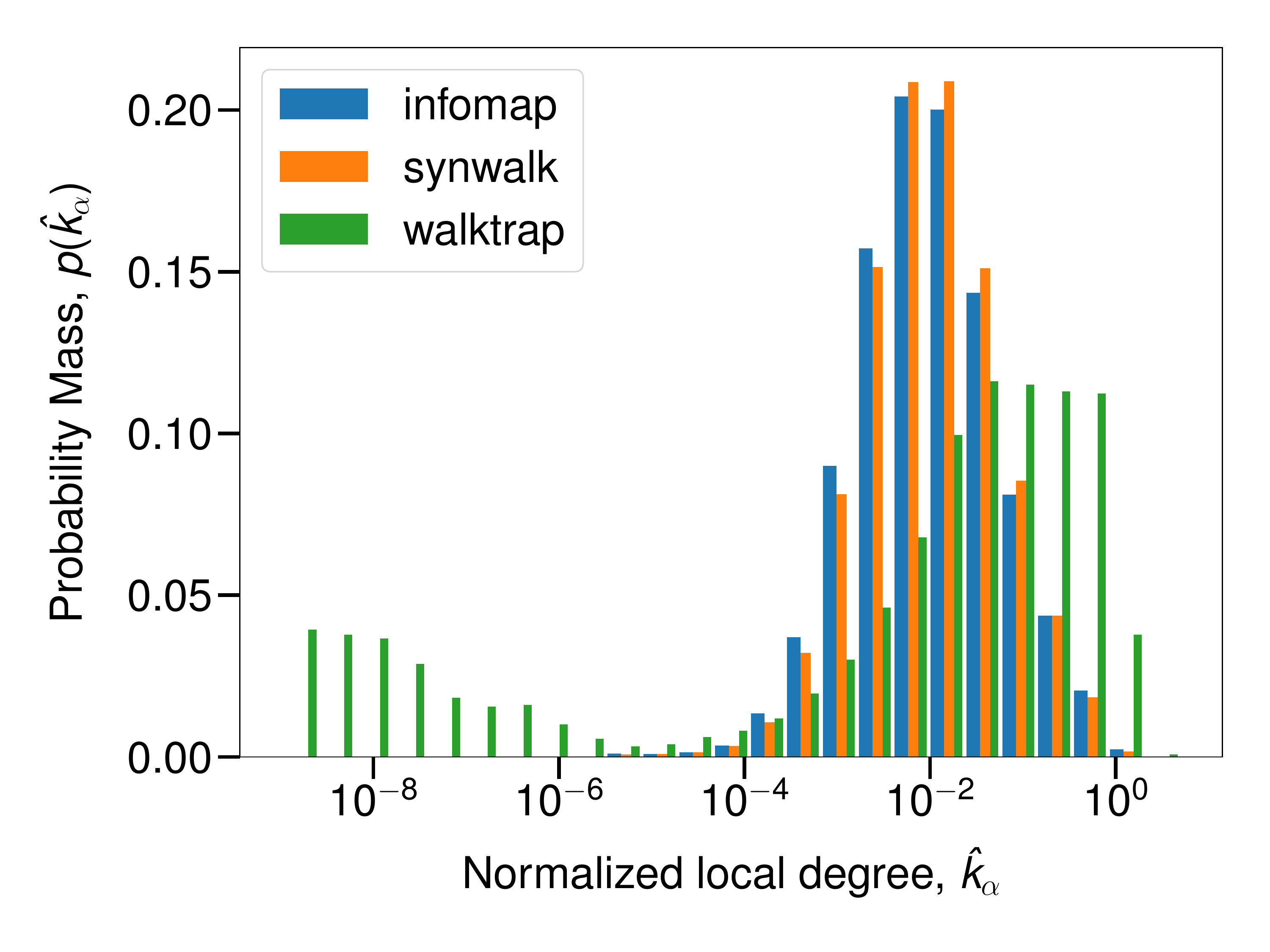}
        \caption{dblp.}
    \end{subfigure}
   \begin{subfigure}{0.32\textwidth}
        \centering
        \includegraphics[width=1.0\textwidth]{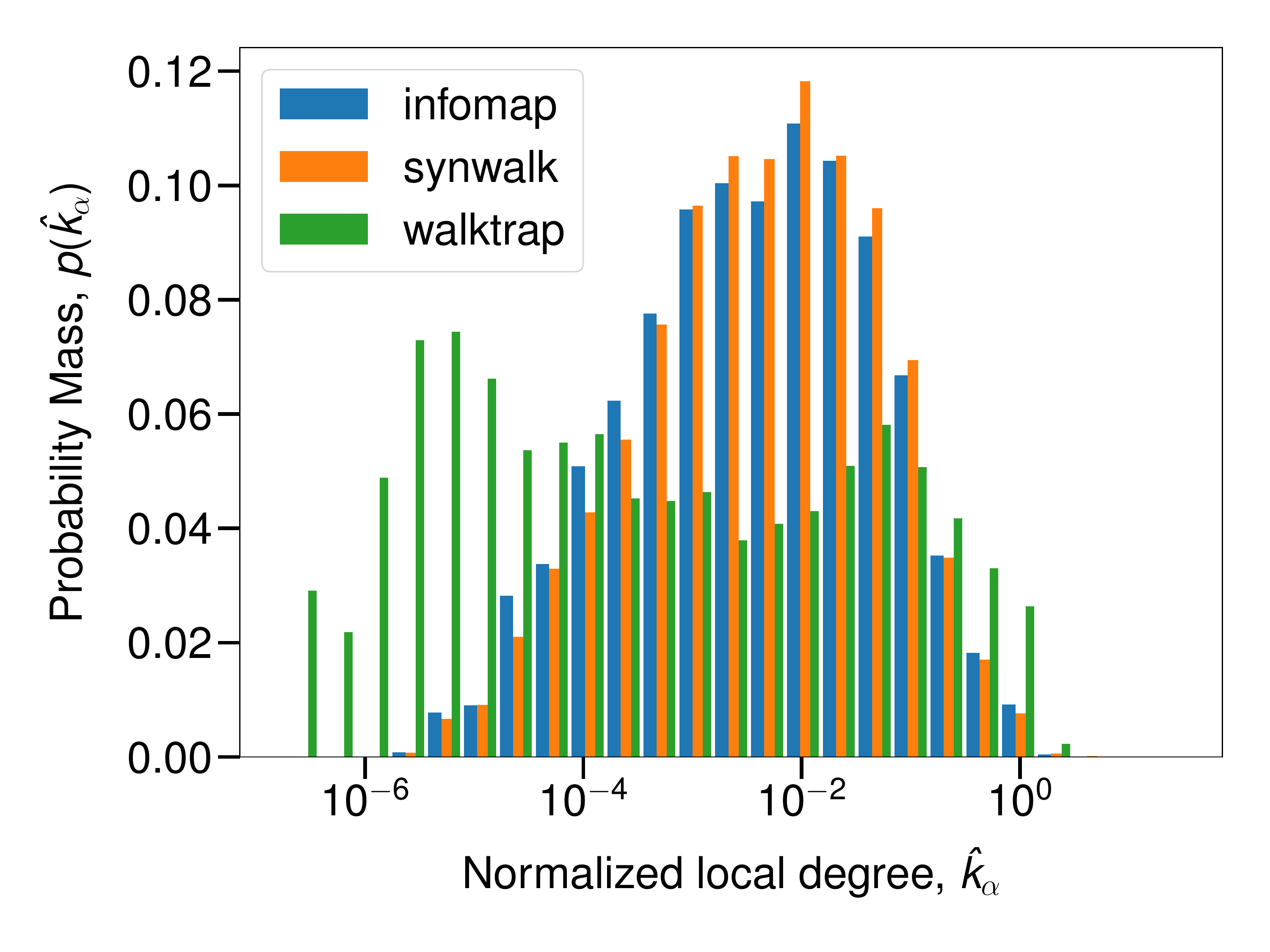}
        \caption{facebook.}
    \end{subfigure}
    \begin{subfigure}{0.32\textwidth}
        \centering
        \includegraphics[width=1.0\textwidth]{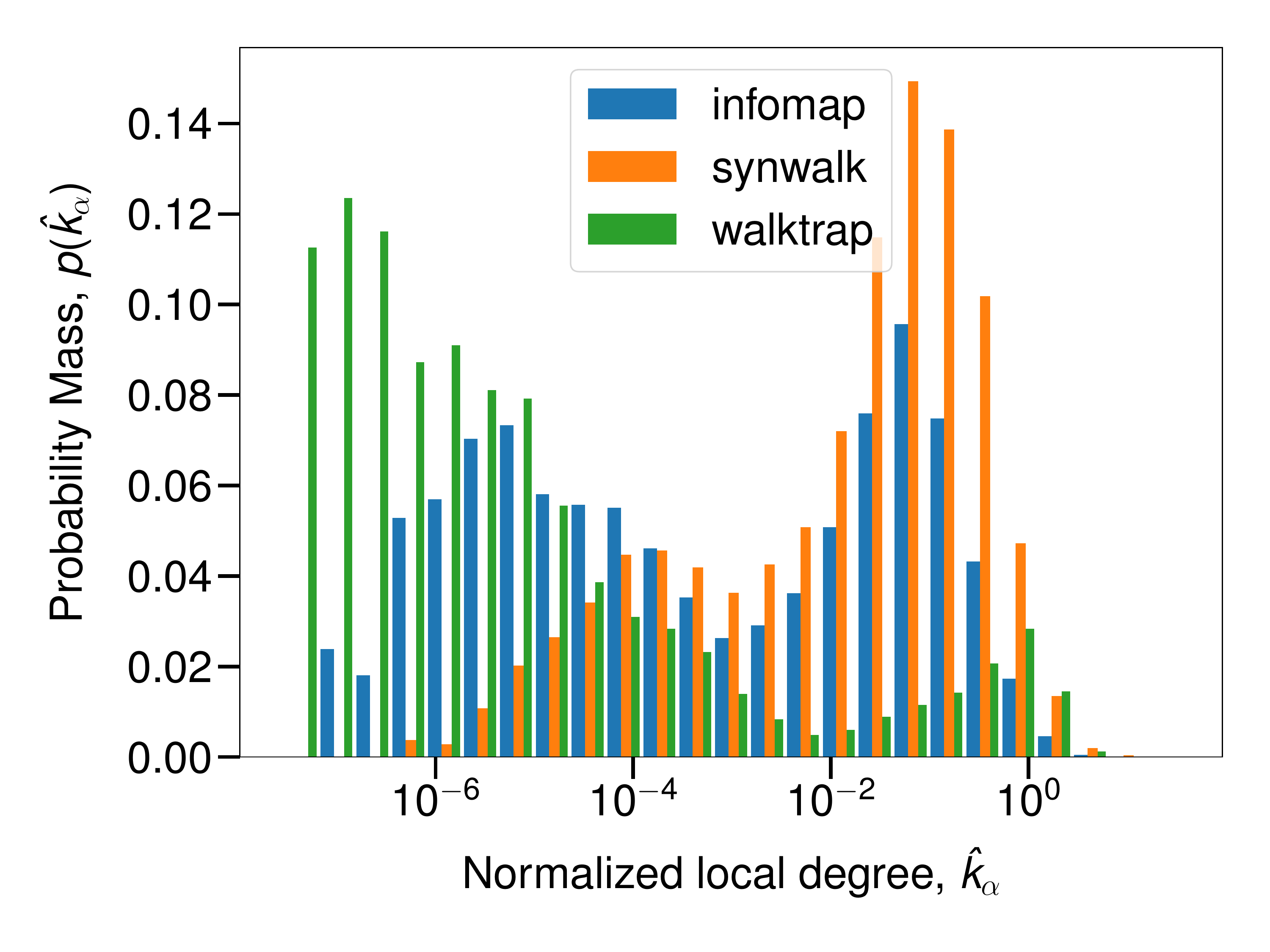}
        \caption{github.}
    \end{subfigure}
    \\
    \begin{subfigure}{0.32\textwidth}
        \centering
        \includegraphics[width=1.0\textwidth]{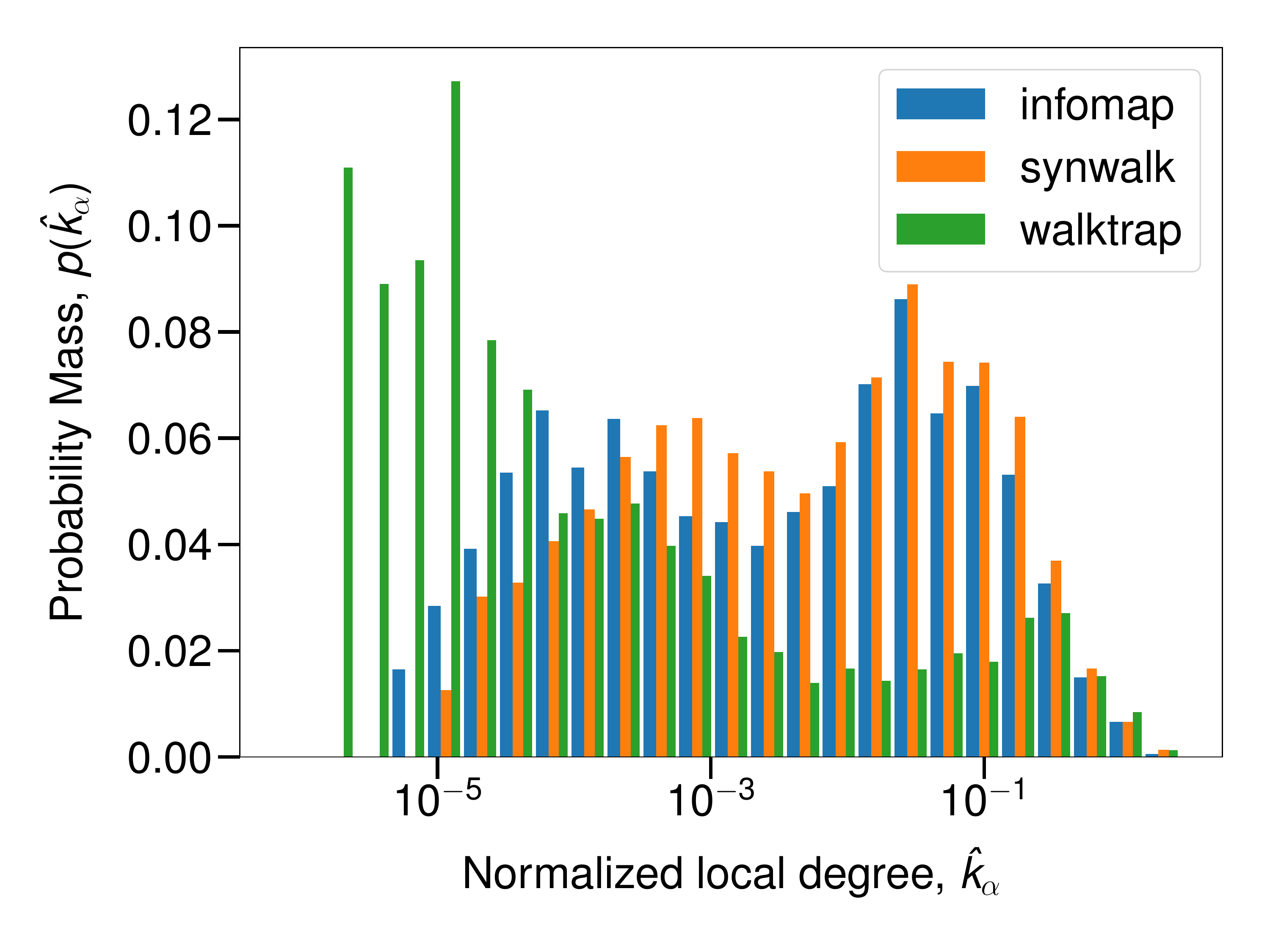}
        \caption{lastfm-asia.}
    \end{subfigure}
    \begin{subfigure}{0.32\textwidth}
        \centering
        \includegraphics[width=1.0\textwidth]{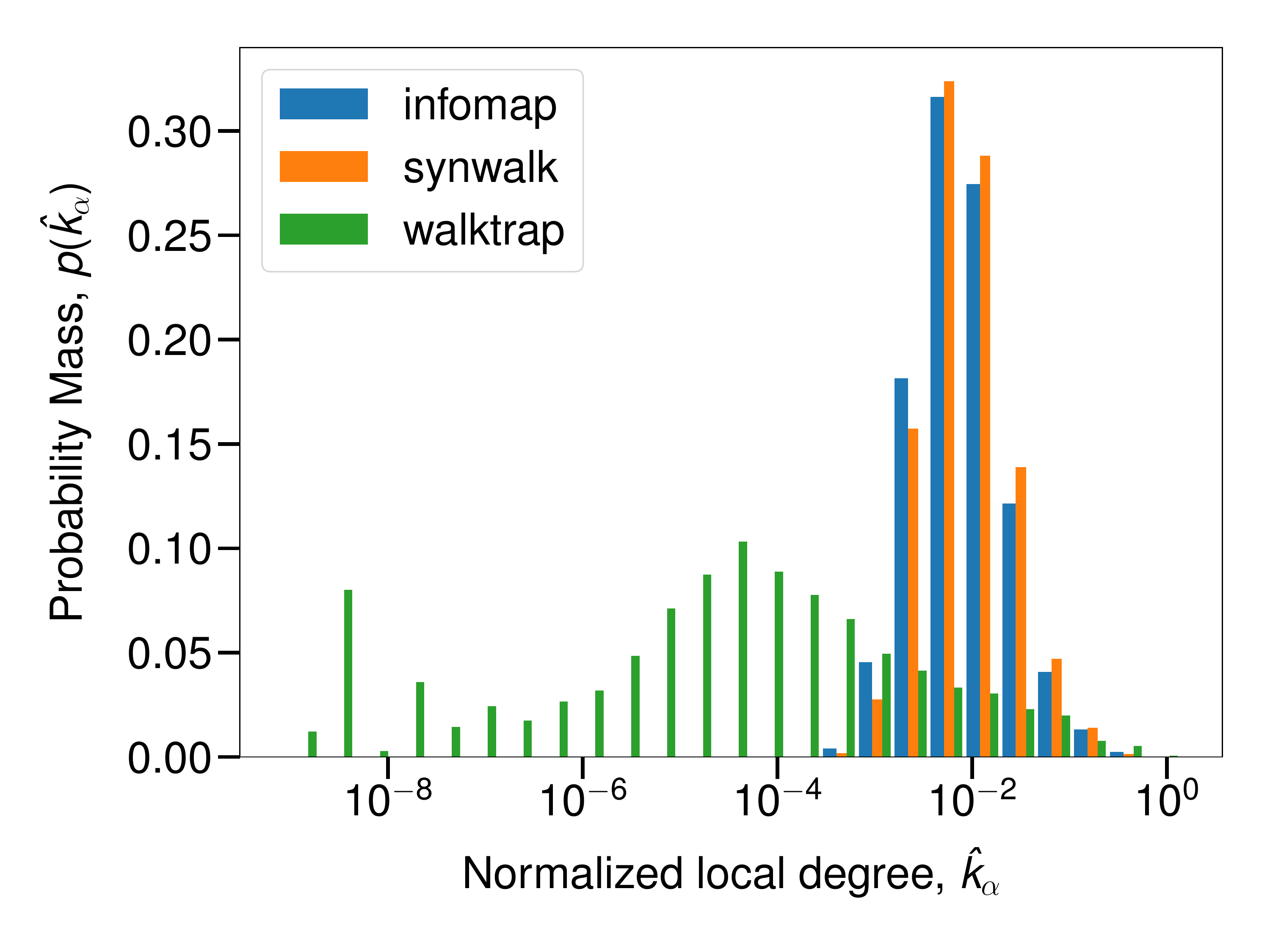}
        \caption{pennsylvania-roads.}
    \end{subfigure}
    \begin{subfigure}{0.32\textwidth}
        \centering
        \includegraphics[width=1.0\textwidth]{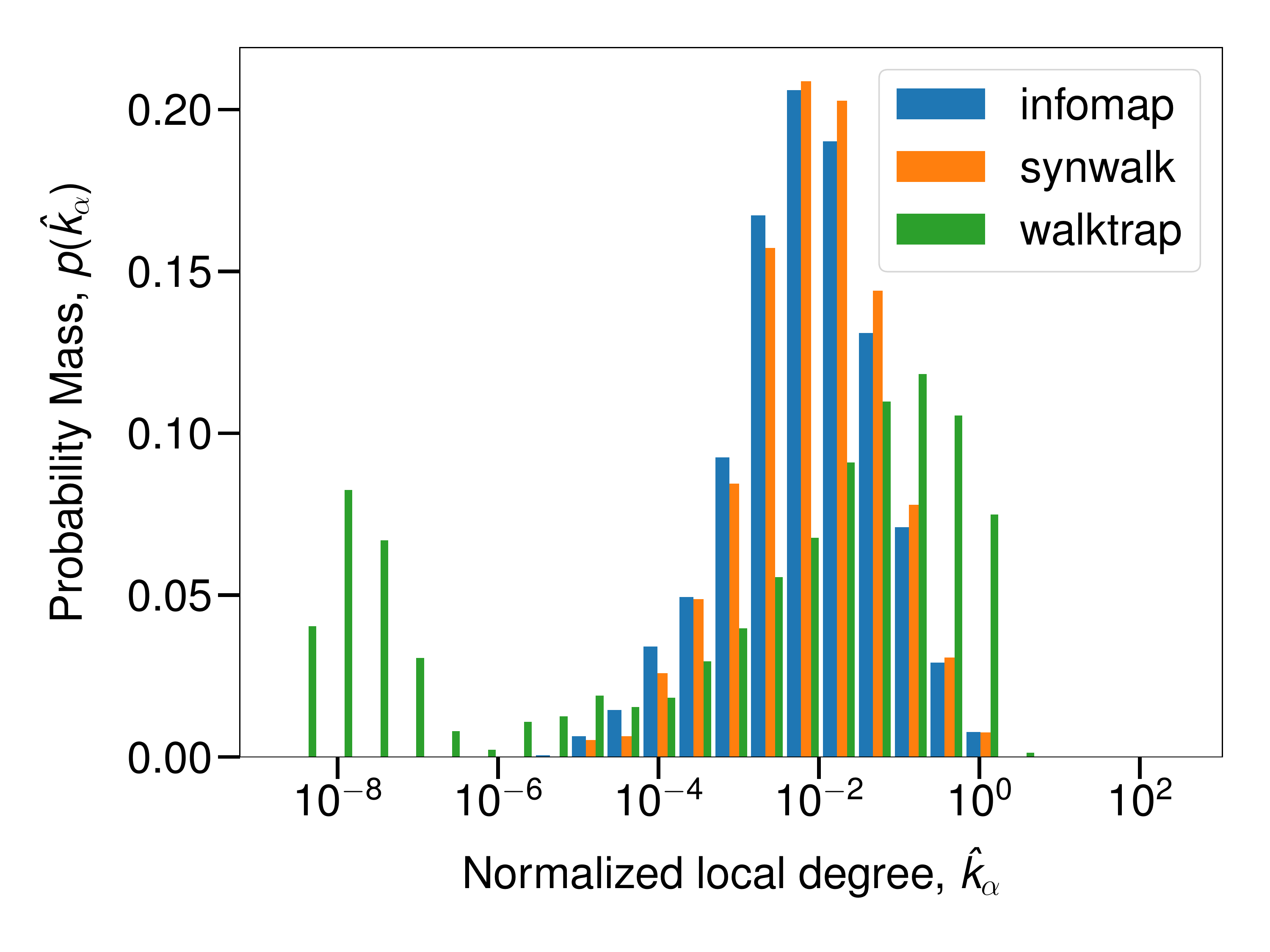}
        \caption{wordnet.}
    \end{subfigure}
    \caption{Distributions of normalized local degrees w.r.t.\ the discovered communities on empirical networks. The distributions generated by Synwalk resemble Infomap closely. An exception here is again the github network.}\label{fig:empirical_nlds}
\end{figure*}

	\section{Conclusion}\label{sec:conclusion}
In this work, we introduced Synwalk, a community detection method based on random walk modelling, that is characterized by an information-theoretic objective function. Our experiments underline the solid theoretical basis of synthetic random walk-based models and show that we can achieve robust performance across a wide range of problem setups. For specific networks, e.g., networks with many small communities and low average degree, Synwalk outperforms Infomap and Walktrap, at least on generated LFR benchmark graphs. 

We deem random walk modelling an interesting counterpart to (stochastic) block modelling for community detection that deserves more attention, as it opens up many interesting avenues for future research. 
For example, while the Synwalk objective is perfectly applicable to directed networks, our present study was limited to undirected networks only, suggesting a closer investigation of Synwalk in directed and/or weighted networks or networks with special properties (e.g., small worlds, etc.). Further, a deeper understanding of  the optimization landscape induced by the Synwalk objective and the influence of the optimization algorithm is required, as well as an extension of the approach to overlapping and hierarchical community structures. Finally, future research shall investigate random walk modelling approaches with structural assumptions different from those of the Synwalk objective, and whether such approaches can be tailored to detect communities of specific types or within specific network classes.

	%

	\bibliographystyle{spbasic}      
	\bibliography{library.bib}   

\begin{thebibliography}{37}
\providecommand{\natexlab}[1]{#1}
\providecommand{\url}[1]{{#1}}
\providecommand{\urlprefix}{URL }
\expandafter\ifx\csname urlstyle\endcsname\relax
  \providecommand{\doi}[1]{DOI~\discretionary{}{}{}#1}\else
  \providecommand{\doi}{DOI~\discretionary{}{}{}\begingroup
  \urlstyle{rm}\Url}\fi
\providecommand{\eprint}[2][]{\url{#2}}

\bibitem[{Blondel et~al.(2008)Blondel, Guillaume, Lambiotte, and
  Lefebvre}]{Blondel2008}
Blondel VD, Guillaume JL, Lambiotte R, Lefebvre E (2008) {Fast unfolding of
  communities in large networks}. Journal of Statistical Mechanics: Theory and
  Experiment 2008(10):P10008, \doi{10.1088/1742-5468/2008/10/P10008}

\bibitem[{Brandes et~al.(2008)Brandes, Delling, Gaertler, Gorke, Hoefer,
  Nikoloski, and Wagner}]{Brandes2008}
Brandes U, Delling D, Gaertler M, Gorke R, Hoefer M, Nikoloski Z, Wagner D
  (2008) {On Modularity Clustering}. IEEE Transactions on Knowledge and Data
  Engineering 20(2):172--188, \doi{10.1109/TKDE.2007.190689}

\bibitem[{Clauset et~al.(2004)Clauset, Newman, and Moore}]{Clauset2004}
Clauset A, Newman MEJ, Moore C (2004) {Finding community structure in very
  large networks}. Physical Review E 70(6):066111,
  \doi{10.1103/PhysRevE.70.066111}

\bibitem[{Cover and Thomas(2006)}]{Cover2006}
Cover TM, Thomas JA (2006) {Elements of Information Theory (Wiley Series in
  Telecommunications and Signal Processing)}, 2nd edn. Wiley-Interscience, USA

\bibitem[{Fellbaum(1998)}]{Miller1998}
Fellbaum C (ed)  (1998) {WordNet: An Electronic Lexical Database}. MIT Press,
  Cambridge

\bibitem[{Fortunato(2010)}]{Fortunato2010}
Fortunato S (2010) {Community detection in graphs}. Physics Reports
  486(3-5):75--174, \doi{10.1016/j.physrep.2009.11.002}

\bibitem[{Fortunato and Hric(2016)}]{Fortunato2016}
Fortunato S, Hric D (2016) {Community detection in networks: A user guide}.
  Physics Reports 659:1--44, \doi{10.1016/j.physrep.2016.09.002}

\bibitem[{Girvan and Newman(2002)}]{Girvan2002}
Girvan M, Newman MEJ (2002) {Community structure in social and biological
  networks}. Proceedings of the National Academy of Sciences 99(12):7821--7826,
  \doi{10.1073/pnas.122653799}

\bibitem[{{Hurley} and {Duriakova}(2015)}]{Hurley2015}
{Hurley} N, {Duriakova} E (2015) Reformulations of the map equation for
  community finding and blockmodelling. In: Proc. IEEE/ACM Int. Conf. on
  Advances in Social Networks Analysis and Mining (ASONAM), Paris, France, pp
  1606--1607, \doi{10.1145/2808797.2809356}

\bibitem[{{Hurley} and {Duriakova}(2016)}]{Hurley2016}
{Hurley} N, {Duriakova} E (2016) An information theoretic approach to
  generalised blockmodelling for the identification of meso-scale structure in
  networks. In: Proc. IEEE/ACM Int. Conf. on Advances in Social Networks
  Analysis and Mining (ASONAM), San Francisco, CA, pp 319--322,
  \doi{10.1109/ASONAM.2016.7752252}

\bibitem[{Kesidis and Walrand(1993)}]{Kesidis1993}
Kesidis G, Walrand J (1993) {Relative entropy between Markov transition rate
  matrices}. IEEE Transactions on Information Theory 39(3):1056--1057,
  \doi{10.1109/18.256516}

\bibitem[{Kunegis(2013)}]{konect}
Kunegis J (2013) {KONECT}. In: Proc. Int. Conf. on World Wide Web - WWW '13
  Companion, New York, NY, USA, pp 1343--1350, \doi{10.1145/2487788.2488173}

\bibitem[{Lancichinetti and Fortunato(2009{\natexlab{a}})}]{Lancichinetti2009b}
Lancichinetti A, Fortunato S (2009{\natexlab{a}}) {Benchmarks for testing
  community detection algorithms on directed and weighted graphs with
  overlapping communities}. Physical Review E 80(1):016118,
  \doi{10.1103/PhysRevE.80.016118}

\bibitem[{Lancichinetti and Fortunato(2009{\natexlab{b}})}]{Lancichinetti2009c}
Lancichinetti A, Fortunato S (2009{\natexlab{b}}) {Community detection
  algorithms: A comparative analysis}. Physical Review E 80(5):056117,
  \doi{10.1103/PhysRevE.80.056117}

\bibitem[{Lancichinetti et~al.(2008)Lancichinetti, Fortunato, and
  Radicchi}]{Lancichinetti2008}
Lancichinetti A, Fortunato S, Radicchi F (2008) {Benchmark graphs for testing
  community detection algorithms}. Physical Review E 78(4):046110,
  \doi{10.1103/PhysRevE.78.046110}

\bibitem[{Leskovec and Krevl(2014)}]{snapnets}
Leskovec J, Krevl A (2014) {SNAP Datasets: Stanford Large Network Dataset
  Collection}. http://snap.stanford.edu/data

\bibitem[{Leskovec et~al.(2008)Leskovec, Lang, Dasgupta, and
  Mahoney}]{Leskovec2008}
Leskovec J, Lang KJ, Dasgupta A, Mahoney MW (2008) {Community Structure in
  Large Networks: Natural Cluster Sizes and the Absence of Large Well-Defined
  Clusters}. Internet Mathematics 6(1):29--123

\bibitem[{Leskovec et~al.(2010)Leskovec, Lang, and Mahoney}]{Leskovec2010}
Leskovec J, Lang KJ, Mahoney M (2010) {Empirical comparison of algorithms for
  network community detection}. In: Proc. Int. Conf. on World Wide Web - WWW
  '10, New York, NY, USA, p 631, \doi{10.1145/1772690.1772755}

\bibitem[{Newman(2006)}]{Newman2006}
Newman MEJ (2006) {Modularity and community structure in networks}. Proceedings
  of the National Academy of Sciences 103(23):8577--8582,
  \doi{10.1073/pnas.0601602103}

\bibitem[{Newman and Girvan(2004)}]{Newman2004}
Newman MEJ, Girvan M (2004) {Finding and evaluating community structure in
  networks}. Physical Review E 69(2):026113, \doi{10.1103/PhysRevE.69.026113}

\bibitem[{Orman and Labatut(2009)}]{Orman2009}
Orman GK, Labatut V (2009) {A Comparison of Community Detection Algorithms on
  Artificial Networks}. In: Gama J, Costa VS, Jorge AM, Brazdil PB (eds)
  Discovery Science, Springer Berlin Heidelberg, Berlin, Heidelberg, pp
  242--256

\bibitem[{Peel et~al.(2017)Peel, Larremore, and Clauset}]{Peel2017}
Peel L, Larremore DB, Clauset A (2017) {The ground truth about metadata and
  community detection in networks}. Science Advances 3(5):e1602548,
  \doi{10.1126/sciadv.1602548}

\bibitem[{Pons and Latapy(2005)}]{Pons2005}
Pons P, Latapy M (2005) {Computing Communities in Large Networks Using Random
  Walks}. In: Yolum p, G{\"u}ng{\"o}r T, G{\"u}rgen F, {\"O}zturan C (eds)
  Computer and Information Sciences - ISCIS, Istanbul, Turkey, vol 3733 LNCS,
  pp 284--293, \doi{10.1007/11569596_31}

\bibitem[{Rached et~al.(2004)Rached, Alajaji, and Campbell}]{Rached2004}
Rached Z, Alajaji F, Campbell L (2004) {The Kullback–Leibler Divergence Rate
  Between Markov Sources}. IEEE Transactions on Information Theory
  50(5):917--921, \doi{10.1109/TIT.2004.826687}

\bibitem[{Radicchi et~al.(2004)Radicchi, Castellano, Cecconi, Loreto, Parisi,
  and Paris}]{Radicchi2004a}
Radicchi F, Castellano C, Cecconi F, Loreto V, Parisi D, Paris D (2004)
  {Defining and identifying communities in networks}. Proceedings of the
  National Academy of Sciences 101(9):2658--2663, \doi{10.1073/pnas.0400054101}

\bibitem[{Raghavan et~al.(2007)Raghavan, Albert, and Kumara}]{Raghavan2007}
Raghavan UN, Albert R, Kumara S (2007) Near linear time algorithm to detect
  community structures in large-scale networks. Phys Rev E 76:036106,
  \doi{10.1103/PhysRevE.76.036106}

\bibitem[{Reichardt and Bornholdt(2006)}]{Reichardt2006}
Reichardt J, Bornholdt S (2006) Statistical mechanics of community detection.
  Phys Rev E 74:016110, \doi{10.1103/PhysRevE.74.016110}

\bibitem[{Rosvall and Bergstrom(2008)}]{Rosvall2008}
Rosvall M, Bergstrom CT (2008) {Maps of random walks on complex networks reveal
  community structure}. Proceedings of the National Academy of Sciences
  105(4):1118--1123, \doi{10.1073/pnas.0706851105}

\bibitem[{Rosvall et~al.(2009)Rosvall, Axelsson, and Bergstrom}]{Rosvall2009}
Rosvall M, Axelsson D, Bergstrom CT (2009) {The map equation}. The European
  Physical Journal Special Topics 178(1):13--23,
  \doi{10.1140/epjst/e2010-01179-1}

\bibitem[{Rosvall et~al.(2019)Rosvall, Delvenne, Schaub, and
  Lambiotte}]{Rosvall2019}
Rosvall M, Delvenne JC, Schaub MT, Lambiotte R (2019) {Different Approaches to
  Community Detection}. In: Doreian P, Batagelj V, Ferligoj A (eds) Advances in
  Network Clustering and Blockmodeling, Wiley, pp 105--119,
  \doi{10.1002/9781119483298.ch4}

\bibitem[{Rozemberczki and Sarkar(2020)}]{Rozemberczki2020}
Rozemberczki B, Sarkar R (2020) Characteristic functions on graphs: Birds of a
  feather, from statistical descriptors to parametric models. In: d'Aquin M,
  Dietze S, Hauff C, Curry E, Cudr{\'{e}}{-}Mauroux P (eds) {CIKM} '20: The
  29th {ACM} Int. Conf. on Information and Knowledge Management, Virtual Event,
  Ireland, pp 1325--1334, \doi{10.1145/3340531.3411866}

\bibitem[{Rozemberczki et~al.(2019)Rozemberczki, Allen, and
  Sarkar}]{Rozemberczki2019}
Rozemberczki B, Allen C, Sarkar R (2019) {Multi-scale Attributed Node
  Embedding}. CoRR abs/1909.13021

\bibitem[{Toth(2020)}]{Toth2020}
Toth C (2020) {Synthesizing Infomap}. Master's thesis, Graz University of
  Technology, \doi{10.5281/zenodo.4446856}

\bibitem[{Toth(2021)}]{Toth2021}
Toth C (2021) A collection of {LFR} benchmark graphs.
  \doi{10.5281/zenodo.4450167}

\bibitem[{Vinh et~al.(2010)Vinh, Epps, and Bailey}]{Vinh2010}
Vinh NX, Epps J, Bailey J (2010) {Information Theoretic Measures for
  Clusterings Comparison: Variants, Properties, Normalization and Correction
  for Chance}. The Journal of Machine Learning Research 11:2837--2854

\bibitem[{Yang and Leskovec(2015)}]{Yang2015}
Yang J, Leskovec J (2015) {Defining and evaluating network communities based on
  ground-truth}. Knowledge and Information Systems 42(1):181--213,
  \doi{10.1007/s10115-013-0693-z}

\bibitem[{Yang et~al.(2016)Yang, Algesheimer, and Tessone}]{Yang2016}
Yang Z, Algesheimer R, Tessone CJ (2016) {A Comparative Analysis of Community
  Detection Algorithms on Artificial Networks}. Scientific Reports 6(1):30750,
  \doi{10.1038/srep30750}

\end{thebibliography}
	
	\begin{acknowledgements}
		The work of Bernhard C. Geiger was supported by the HiDALGO project and has been funded by the European Commission’s ICT activity of the H2020 Programme under grant agreement number 824115. The Know-Center is funded within the Austrian COMET Program - Competence Centers for Excellent Technologies - under the auspices of the Austrian Federal Ministry for Climate Action, Environment, Energy, Mobility, Innovation and Technology, the Austrian Federal Ministry for Digital and Economic Affairs, and the State of Styria. COMET is managed by the Austrian Research Promotion Agency FFG.
	\end{acknowledgements}
	
	\appendix
	
\section{Proofs}\label{appx:proofs}

\subsection{Derivation of the Synwalk Objective}\label{appx:mainProp}

\begin{proposition}\label{prop:mainProp}
    Let $\{X_t\}$ be a stationary Markov chain with transition probability matrix $P$ derived from the weight matrix $W$ of the network $\graph$ and let $Q^\dom{Y}$ be a transition probability matrix of the same size parameterized as in~\eqref{eq:qmat}. Then, for every partition $\dom{Y}$, it holds that
    \begin{align}\label{eq:prop:statements}
    \min_{\{[r_\alpha^i]_{\alpha\in \dom{Y}_i},\ s_i,\ u_i\}_{i\in\idxset{Y}}} \kldr{P}{Q^\dom{Y}} \leq \mutinf{X_t;X_{t-1}} - \sum_{i\in\idxset{Y}} p_i\kld{p_{i\to i}}{p_i}
    \end{align}
    where $p_i$ and $p_{i\to i}$ are defined as in~\eqref{eq:cluster-probabilities}.
\end{proposition}

Since $ \mutinf{X_t;X_{t-1}}$ is independent of the candidate partition $\dom{Y}$, minimizing the right-hand side of~\eqref{eq:prop:statements} over the partition $\dom{Y}$ is equivalent to a corresponding maximization of $\objective{\dom{Y}}$.

\begin{proof}\label{proof:mainProp}
    Let  $p_{\alpha\to i} := \sum_{\beta\in \dom{Y}_i} p_{\alpha \to \beta}$,  $p_{\alpha\not\to i}:=\sum_{j\neq i} p_{\alpha \to j} = 1-p_{\alpha\to i}$, and let $\indicator{i}{\alpha}$ denote the indicator function for cluster $\dom{Y}_i$, i.e., 
    \begin{align}
        \indicator{i}{\alpha} = 
        \begin{cases}1 \quad\text{if}\quad \alpha\in\dom{Y}_i\\ 0 \quad\text{otherwise.}
        \end{cases}
    \end{align} Then from~\eqref{eq:qmat} follows that
    \begin{align}\label{eq:sums}  
    \kldr{P}{Q^\dom{Y}}
    &= \sum_{\alpha,\beta}  p_\alpha p_{\alpha \to \beta} \log \frac{p_{\alpha \to \beta}}{q_{\alpha \to \beta}}\notag\\
    &\begin{aligned}
        =\sum_{\alpha,\beta} p_\alpha p_{\alpha \to \beta} &\left[ \indicator{m(\alpha)}{\beta} \log \frac{p_{\alpha \to \beta}}{r_\beta^{m(\beta)} (1-s_{m(\alpha)})} \right. \\
        & \left. + (1 - \indicator{m(\alpha)}{\beta}) \log \frac{p_{\alpha \to \beta}}{r_\beta^{m(\beta)} s_{m(\alpha)} \frac{u_{m(\beta)}}{1 - u_{m(\alpha)}}} \right]
    \end{aligned}\notag\\ \addlinespace[7pt]
    &\begin{aligned}
        =\sum_j \sum_{\alpha}\sum_{\beta\in \dom{Y}_j}  p_\alpha p_{\alpha \to \beta} \log\frac{\frac{p_{\alpha \to \beta}}{p_{\alpha \to j}}}{r_\beta^j} + &\sum_i \sum_{\alpha\in \dom{Y}_i} p_\alpha p_{\alpha \to i}\log\frac{p_{\alpha \to i}}{1-s_i}\\
        + &\sum_i\sum_{j\neq i}\sum_{\alpha\in \dom{Y}_i} p_\alpha p_{\alpha \to j}\log\frac{p_{\alpha \to j}}{s_i\frac{u_j}{1 - u_i}}
    \end{aligned}\notag\\ \addlinespace[7pt]
    &\begin{aligned}
    =\sum_j \sum_{\alpha}\sum_{\beta\in \dom{Y}_j}  p_\alpha p_{\alpha \to \beta} \log\frac{\frac{p_{\alpha \to \beta}}{p_{\alpha \to j}}}{r_\beta^j}
    &+ \sum_i\sum_{j\neq i}\sum_{\alpha\in \dom{Y}_i} p_\alpha p_{\alpha \to j}\log\frac{\frac{p_{\alpha \to j}}{p_{\alpha\not\to i}}}{\frac{u_j}{1 - u_i}} \notag\\
    + \sum_i \sum_{\alpha\in \dom{Y}_i} p_\alpha p_{\alpha \to i}\log\frac{p_{\alpha \to i}}{1-s_i} &+ \sum_i\sum_{\alpha\in \dom{Y}_i} p_\alpha p_{\alpha\not\to i}\log\frac{p_{\alpha\not\to i}}{s_i}
    \end{aligned}\notag\\ \addlinespace[7pt]
    &\begin{aligned}
    =\sum_j &\sum_{\alpha}\sum_{\beta\in \dom{Y}_j}  p_\alpha p_{\alpha \to \beta} \log\frac{\frac{p_{\alpha \to \beta}}{p_{\alpha \to j}}}{r_\beta^j}
    + \sum_i\sum_{j\neq i}\sum_{\alpha\in \dom{Y}_i} p_\alpha p_{\alpha \to j}\log\frac{\frac{p_{\alpha \to j}}{p_{\alpha\not\to i}}}{\frac{u_j}{1 - u_i}}\\
    + &\sum_i \sum_{\alpha\in \dom{Y}_i} \left[ p_\alpha p_{\alpha \to i}\log\frac{p_\alpha p_{\alpha \to i}}{p_\alpha (1-s_i)} + p_\alpha p_{\alpha\not\to i}\log\frac{p_\alpha p_{\alpha\not\to i}}{p_\alpha s_i} \right]
    \end{aligned}\notag \\ \addlinespace[7pt]    
    &\begin{aligned}
    =\sum_j &\sum_{\alpha}\sum_{\beta\in \dom{Y}_j}  p_\alpha p_{\alpha \to \beta} \log\frac{\frac{p_{\alpha \to \beta}}{p_{\alpha \to j}}}{r_\beta^j}
    + \sum_i\sum_{j\neq i}\sum_{\alpha\in \dom{Y}_i} p_\alpha p_{\alpha \to j}\log\frac{\frac{p_{\alpha \to j}}{p_{\alpha\not\to i}}}{\frac{r_j}{1 - r_i}} \\
    + &\sum_i \kld{p_\alpha p_{\alpha \to i}}{p_\alpha (1-s_i)}.
    \end{aligned}  
    \end{align}
    We can now independently minimize the first summation term w.r.t $\{[r_\beta^j]_{\beta\in \dom{Y}_j}\}_{j\in\idxset{Y}}$, the second summation term w.r.t $[u_i]_{i\in\idxset{Y}}$, and the last summation term w.r.t $[s_i]_{i\in\idxset{Y}}$. Considering the latter, one can show~\citep[Lemma~10.8.1]{Cover2006} that the Kullback-Leibler divergence $\kld{p_\alpha p_{\alpha \to i}}{p_\alpha (1-s_i)}$ is minimized for
    \begin{align}
    1-s_i = \sum_{\alpha\in \dom{Y}_i} p_\alpha p_{\alpha \to i} = p_{i\to i}
    \end{align}
    and hence
    \begin{align}
    s_i &= 1 - p_{i\to i} = p_{i\not\to i}.
    \end{align}

    The minimizer regarding the distribution over nodes can be found along similar lines. For the first summation term in~\eqref{eq:sums} we observe that
    \begin{align*}
     \sum_j \sum_{\alpha}\sum_{\beta\in \dom{Y}_j}  p_\alpha p_{\alpha \to \beta} \log\frac{\frac{p_{\alpha \to \beta}}{p_{\alpha \to j}}}{r_\beta^j} 
    &= \sum_j \sum_{\alpha} p_\alpha p_{\alpha \to j}   \sum_{\beta\in \dom{Y}_j}  \frac{p_{\alpha \to \beta}}{p_{\alpha \to j}}\log\frac{\frac{p_{\alpha \to \beta}}{p_{\alpha \to j}}}{r_\beta^j} \\ \addlinespace[7pt]
    &= \sum_j p_j \sum_{\alpha} \frac{p_\alpha p_{\alpha \to j}}{p_j}   \sum_{\beta\in \dom{Y}_j}  \frac{p_{\alpha \to \beta}}{p_{\alpha \to j}}\log\frac{\frac{p_{\alpha \to \beta}}{p_{\alpha \to j}}}{r_\beta^j}  \\ \addlinespace[7pt]
    &= \sum_j p_j \sum_{\alpha} \frac{p_\alpha p_{\alpha \to j}}{p_j}   \sum_{\beta\in \dom{Y}_j}  \frac{p_{\alpha \to \beta}}{p_{\alpha \to j}}\log\frac{\frac{p_\alpha p_{\alpha \to \beta}}{p_j}}{\frac{p_\alpha p_{\alpha \to j}}{p_j}r_\beta^j}\\ \addlinespace[7pt]
    &= \sum_j p_j \sum_{\alpha} \sum_{\beta\in \dom{Y}_j} \frac{p_\alpha p_{\alpha \to j}}{p_j}  \frac{p_{\alpha \to \beta}}{p_{\alpha \to j}}\log\frac{\frac{p_\alpha p_{\alpha \to \beta}}{p_j}}{\frac{p_\alpha p_{\alpha \to j}}{p_j}r_\beta^j}\\ \addlinespace[7pt]
    &= \sum_j p_j \kld{\frac{p_\alpha p_{\alpha \to j}}{p_j} \cdot \frac{p_{\alpha \to \beta}}{p_{\alpha \to j}}}{\frac{p_\alpha p_{\alpha \to j}}{p_j} \cdot r_\beta^j} 
    \end{align*}
    and by again employing~\citep[Lemma~10.8.1]{Cover2006}, for a fixed cluster $j$ this quantity is minimized for
    \begin{align}
     r_\beta^j = \sum_\alpha \frac{p_\alpha p_{\alpha \to \beta}}{p_j} = \frac{p_\beta}{p_j} = \Pr{X_t=\beta|Y_t=j}.
    \end{align}
    Applying these minimizers yields
    \begin{align}\label{eq:intermediate_min}
     \min_{\{[r_\alpha^i]_{\alpha\in \dom{Y}_i},\ s_i,\ u_i\}_{i\in\idxset{Y}}} \kldr{P}{Q^\dom{Y}} = \min_{[u_i]_{i\in\idxset{Y}}} &\mutinf{X_t;X_{t-1}} - \ent{Y_t} + \ent{S_t|Y_{t}}  \notag\\ - &\sum_i\sum_{j\neq i} p_i p_{i \to j}\log\frac{u_j}{1 - u_i}
    \end{align}
    where $S_t = 1$ is a binary RV reflecting whether we stay in or leave a cluster at time $t$. Since for the distribution over clusters $[u_i]_{i\in\idxset{Y}}$ there exists no closed-form solution to the best of our knowledge, we choose $[u_i]_{i\in\idxset{Y}} = [p_i]_{i\in\idxset{Y}}$ as a sub-optimal solution. In other words, we utilize the stationary distribution of $\{Y_t\}$. Inserting this choice yields
    \begin{equation}\label{eq:finalequivalence}
    \min_{\{[r_\alpha^i]_{\alpha\in \dom{Y}_i},\ s_i,\ u_i\}_{i\in\idxset{Y}}} \kldr{P}{Q^\dom{Y}} \leq \mutinf{X_t;X_{t-1}} - \sum_i p_i \kld{p_{i\to i}}{p_i}.
    \end{equation}
    This completes the proof.
\end{proof}

\subsection{Bounds on the Synwalk Objective}\label{appx:bounds}

\begin{proposition}\label{prop:bounds}
 Let $\{X_t\}$ be a stationary Markov chain with transition probability matrix $P$ derived from the weight matrix $W$ of the network $\graph$ and let $\dom{Y}$ be any candidate partition. Then, we have
     \begin{equation}
    0 \leq \sum_{i\in\idxset{Y}} p_i\kld{p_{i\to i}}{p_i} \leq \mutinf{Y_t;Y_{t-1}} \leq \mutinf{X_t;X_{t-1}}
    \end{equation}
    where $\{Y_t\}$ is the process obtained by projecting $\{X_t\}$ through the partition $\dom{Y}$ (see Section~\ref{sec:preliminaries}).
\end{proposition}

\begin{proof}
 The first inequality follows immediately from the non-negativity of Kullback-Leibler divergence; the last inequality is the data processing inequality~\citep[Th.~2.8.1]{Cover2006}, which follows from the fact that $Y_t-X_t-X_{t-1}-Y_{t-1}$ is a Markov tuple. We are thus left with proving the second inequality. 
 
 To this end, consider a relaxation of optimization problem~\eqref{eq:original}, in which we let the distribution over clusters depend on the originating cluster; i.e., rather than a single distribution $[u_i]_{i\in\idxset{Y}}$, we now consider a set of distributions $\{[u_i^j]_{i\in\idxset{Y}}\}_{j\in\idxset{Y}}$. In other words, we consider, for every partition $\dom{Y}$, the minimization problem
 \begin{equation}\label{eq:alternative_min}
     \min_{\{[r_\alpha^i]_{\alpha\in \dom{Y}_i},\ s_i,\ [u_j^i]_{j\in\idxset{Y}}\}_{i\in\idxset{Y}}} \kldr{P}{Q^\dom{Y}}.
\end{equation}
It can be shown along the lines of Proposition~\ref{prop:mainProp} that the distributions $[r_\alpha^i]_{\alpha\in \dom{Y}_i}$ and $s_i$, for $i\in\idxset{Y}$, minimizing~\eqref{eq:original} also minimize~\eqref{eq:alternative_min}. Thus, we have with~\eqref{eq:intermediate_min} that
\begin{align}\label{eq:proof2:intermediate}
    \min_{\{[r_\alpha^i]_{\alpha\in \dom{Y}_i},\ s_i,\ [u_j^i]_{j\in\idxset{Y}}\}_{i\in\idxset{Y}}} \kldr{P}{Q^\dom{Y}}&=
    \min_{\{[u_j^i]_{j\in\idxset{Y}}\}_{i\in\idxset{Y}}}
    \mutinf{X_t;X_{t-1}} - \ent{Y_t} + \ent{S_t|Y_{t}}  \notag\\  &\quad\quad -\sum_i\sum_{j\neq i} p_i p_{i \to j}\log\frac{u_j^i}{1 - u_i^i}.
\end{align}
The right-hand side can be shown to be minimized by setting
\begin{equation}
 u_j^i = \frac{p_{i\to j}}{p_{i\not\to i}} = \Pr{Y_t=j|Y_{t-1}=i,Y_t\neq i}= \Pr{Y_t=j|Y_{t-1}=i,S_{t-1}=1}
\end{equation}
for $j\neq i$ and $u_i^i=0$ for every $i\in\idxset{Y}$. Thus, we have
\begin{align}
&-\sum_i\sum_{j\neq i} p_i p_{i \to j}\log\frac{u_j^i}{1 - u_i^i}\notag\\
 &=- \sum_i\sum_{j\neq i} \Pr{Y_t=j,Y_{t-1}=i}\log\Pr{Y_t=j|Y_{t-1}=i,S_{t-1}=1}\\
 &= -\sum_i\sum_{j} \Pr{Y_t=j,Y_{t-1}=i,S_{t-1}=1}\log\Pr{Y_t=j|Y_{t-1}=i,S_{t-1}=1}\\
 &= \Pr{S_{t-1}=1}\ent{Y_t|Y_{t-1},S_{t-1}=1}\\
 &= \ent{Y_t|Y_{t-1},S_{t-1}}
\end{align}
because $\ent{Y_t|Y_{t-1},S_{t-1}=0}=0$. Insertig this into~\eqref{eq:proof2:intermediate} yields
\begin{align}
    &\min_{\{[r_\alpha^i]_{\alpha\in \dom{Y}_i},\ s_i,\ [u_j^i]_{j\in\idxset{Y}}\}_{i\in\idxset{Y}}} \kldr{P}{Q^\dom{Y}}\notag\\
    &=
    \mutinf{X_t;X_{t-1}} - \ent{Y_t} + \ent{S_t|Y_{t}} + \ent{Y_t|Y_{t-1},S_{t-1}}\\
    &\stackrel{(a)}{=} \mutinf{X_t;X_{t-1}} - \ent{Y_t} + \ent{S_{t-1}|Y_{t-1}} + \ent{Y_t|Y_{t-1},S_{t-1}}\\
     &\stackrel{(b)}{=} \mutinf{X_t;X_{t-1}} - \ent{Y_t} + \ent{Y_t,S_{t-1}|Y_{t-1}}\\
     &\stackrel{(c)}{=} \mutinf{X_t;X_{t-1}} - \ent{Y_t} + \ent{Y_t|Y_{t-1}}\\
     &\stackrel{(d)}{=} \mutinf{X_t;X_{t-1}} - \mutinf{Y_t;Y_{t-1}}
\end{align}
where $(a)$ is due to stationarity, $(b)$ due to the chain rule of entropy, $(c)$ since $S_{t-1}$ is a function of $Y_t$ and $Y_{t-1}$, and $(d)$ by the definition of mutual information. Since, finally,~\eqref{eq:alternative_min} optimizes over a larger feasible set (over $[u_i^j]_{i,j\in\idxset{Y}}$ rather than $[u_i]_{i\in\idxset{Y}}$), the minimum of~\eqref{eq:alternative_min} cannot exceed the minimum of~\eqref{eq:original}. We thus have with~\eqref{eq:finalequivalence} that
\begin{equation}
    \mutinf{X_t;X_{t-1}} - \mutinf{Y_t;Y_{t-1}} \leq \mutinf{X_t;X_{t-1}} - \sum_i p_i \kld{p_{i\to i}}{p_i}.
\end{equation}
This completes the proof.
\end{proof}

\subsection{Synwalk objective attains global optimum on network of isolated cliques}\label{app:proof_optimal}

\begin{proposition}\label{prop:optimal}
 Let $\graph$ be an unweighted network of disconnected cliques, defined by the partition $\clutrue$. Then, $\clutrue$ is a global optimizer of~\eqref{eq:si_objective}.
\end{proposition} 

\begin{proof}
We prove Proposition~\ref{prop:optimal} by first deriving the upper bound of the Synwalk objective for the given type of network, and then showing that the ground truth partition achieves this upper bound.

Let $\graph$ be an unweighted network of disconnected cliques, i.e., there exists a partition function $\clutrue$ such that the link set $E$ of $\graph$ equals $\bigcup_{i\in\mathcal{K}^{\clutrue}} (\clutrue_i)^2$. Note that self-loops are included in the link set. Since the network is unweighted, movement within cliques and the invariant distribution within each clique are uniform. Since the resulting Markov chain $\{X_t\}$ is not irreducible, infinitely many invariant distributions exist for the entire alphabet $\dom{X}$. Specifically, for every distribution $p^\bullet=[p^\bullet_i]_{i\in\mathcal{K}^{\clutrue}}$ over the ground truth communities, the distribution on $\dom{X}$ defined by $p_\alpha=p^\bullet_{m(\alpha)}/|\clutrue_{m(\alpha)}|$ is invariant under $P$. Thus, we have that
 \begin{align}
  \ent{X_t} &= -\sum_{\alpha\in\dom{X}} \frac{p^\bullet_{m(\alpha)}}{|\clutrue_{m(\alpha)}|} \log \frac{p^\bullet_{m(\alpha)}}{|\clutrue_{m(\alpha)}|} = -\sum_{i\in\idxset{\dom{Y}^\bullet}} p_i^\bullet \log\frac{p_i^\bullet}{|\clutrue_i|}.
 \end{align}
Given that the random walker is currently at node $\alpha$ in cluster $\clutrue_i$, all nodes in this cluster (including $\alpha$) are equally likely to be visited in the next step. Thus, it follows that $\ent{X_t|X_{t-1}=\alpha} = \log |\clutrue_{m(\alpha)}|$. Combining this with $\ent{X_t}$ yields the upper bound
\begin{align}
    \mutinf{X_t;X_{t-1}} &= \ent{X_t} - \ent{X_t | X_{t-1}} \\
                         &= \ent{X_t} - \sum_{\alpha\in\dom{X}} p_\alpha \ent{X_t | X_{t-1} = \alpha} \\
                         &= \ent{Y^\bullet},
\end{align}
where $Y^\bullet$ is a random variable with distribution $p^\bullet$.
 
Since the clusters coincide with the cliques in $\graph$ there are no links connecting one cluster with another. Hence, we have $p_{i\to i}=1$ for every cluster and thus $\objective{\clutrue}=\ent{Y^\bullet}=\mutinf{X_t;X_{t-1}}$, i.e., the upper bound from~\eqref{eq:bounds} is achieved.
\end{proof}

Note that a coarsening\footnotemark of the partition $\clutrue$ does not achieve this optimum. Indeed, while for $\widetilde{\dom{Y}}$ being a coarsening of $\clutrue$ it still holds that $\widetilde{p}_{i\to i}=1$, we get 
 \begin{equation}
  \widetilde{p}_i=\sum_{j{:}  \clutrue_j\subseteq \widetilde{\dom{Y}}_i} p_j^\bullet 
 \end{equation}
and thus $\objective{\widetilde{\dom{Y}}} = \ent{\widetilde{Y}} < \ent{Y^\bullet}$ due to data processing~\citep[Problem~2.4]{Cover2006}.

\footnotetext{We call a partition $\widetilde{\dom{Y}} = \{\widetilde{\dom{Y}}_1,\dots,\widetilde{\dom{Y}}_M\}$ a \emph{coarsening} of another partition $\dom{Y} = \{\dom{Y}_1,\dots,\dom{Y}_K\}$ if each cluster in the coarsening $\widetilde{\dom{Y}}_i \in \widetilde{\dom{Y}}$ is comprised of one or more clusters of the original partition $\dom{Y}$ and $M < K$.}
	\section{Additional Results on the LFR benchmark}\label{appx:additional_lfr}

Here, we provide additional experiments on the LFR benchmark by investigating the AMI performance as a function of the network density.

In this experiment we again use parameter set A (see Table~\ref{tab:lfr}) to generate the LFR benchmark networks. We now fix the average degree and the mixing parameter of the generated LFR networks while varying their sizes in range $N \in [300, 19200]$. We then plot the AMI as a function of the network density $\rho$ in Figure~\ref{fig:lfr_ami_vs_n}. All three methods show improved AMI for increasing network densities, whereas Infomap's performance still mostly depends on the mixing parameter. Apparently, for Synwalk and Walktrap there exist transition phases from low to high AMI with increasing network density. Interestingly, the location of these transition phases shifts to higher network densities the higher the absolute value of the average degree (Figure~\ref{fig:lfr_ami_vs_n}, rows from left to right; cf. Section~\ref{sssec:ami_vs_mu}). In addition, the AMI results for Synwalk consistently rise above those of Walktrap during these transition phases. For very small network densities Synwalk performance drops faster when compared to Walktrap. In summary, for a given average degree and mixing parameter Synwalk performs better than or equal to Walktrap and Infomap, if the network size/density is sufficiently small/large.

\begin{figure*}[t]
    \centering
    \begin{subfigure}{0.32\textwidth}
        \centering
        \includegraphics[width=1.0\textwidth]{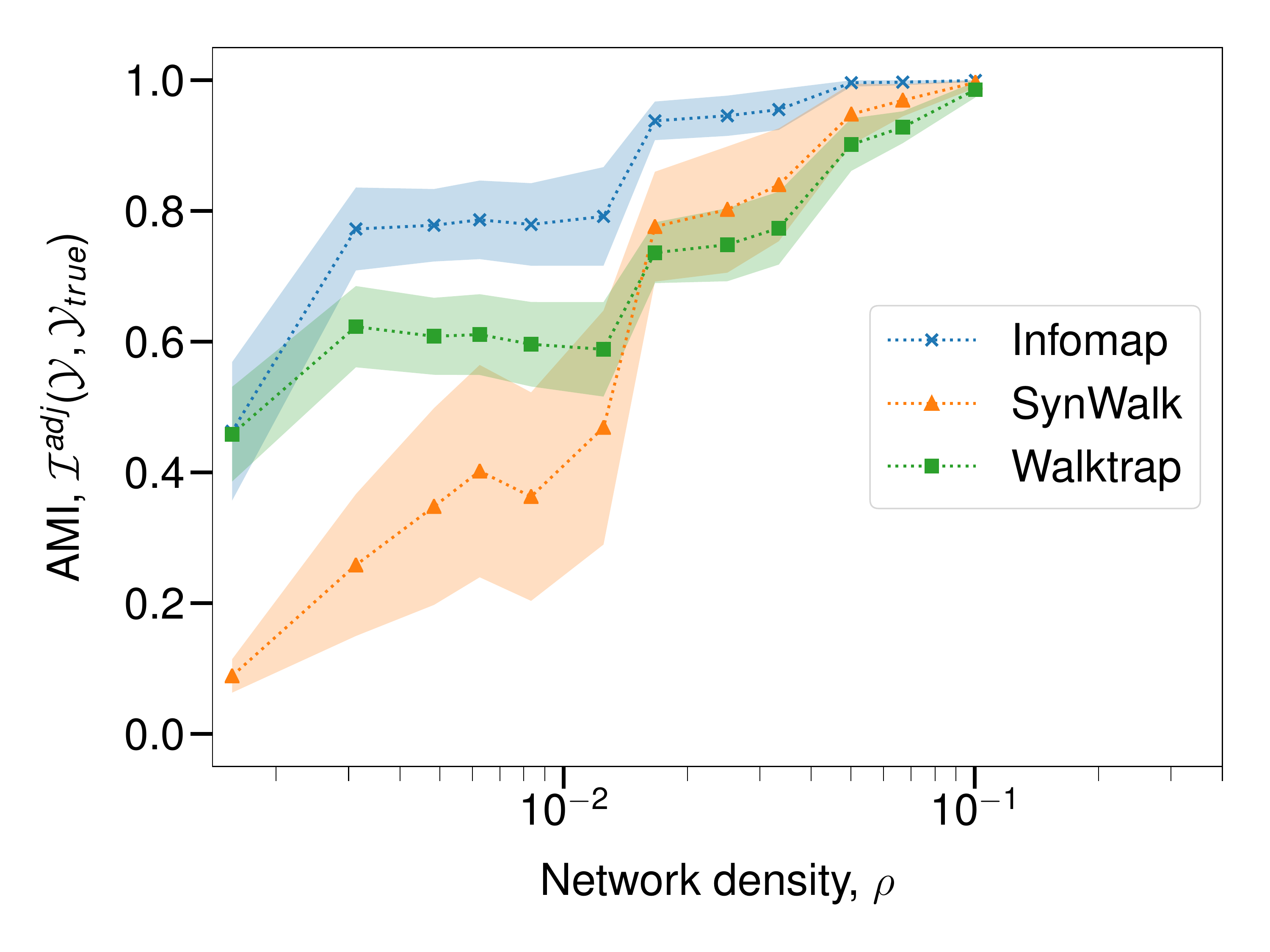}
        \caption{$\kavg = 15$, $\mu = 0.35$.}
    \end{subfigure}
   \begin{subfigure}{0.32\textwidth}
        \centering
        \includegraphics[width=1.0\textwidth]{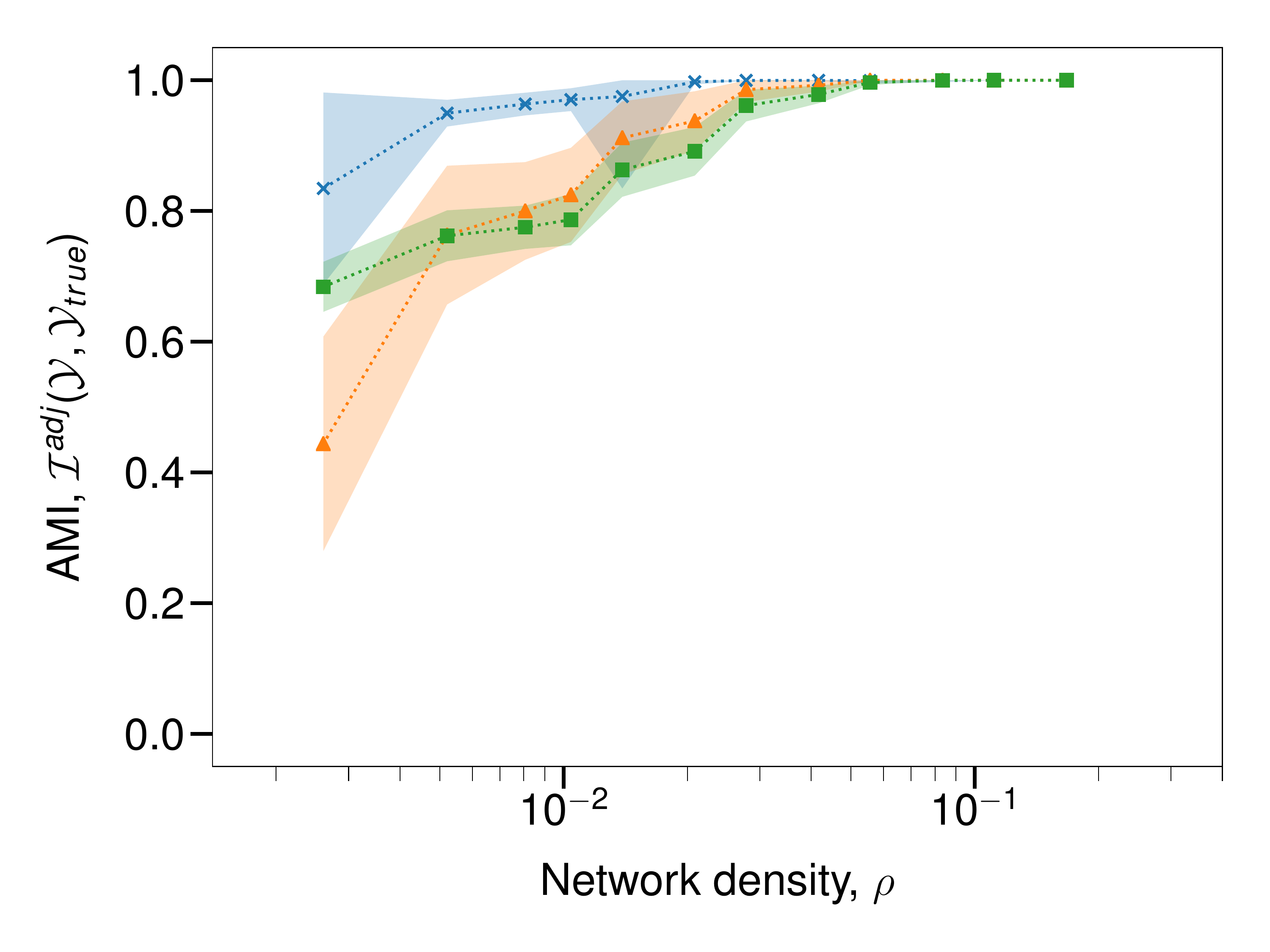}
        \caption{$\kavg = 25$, $\mu = 0.35$.}
    \end{subfigure}
    \begin{subfigure}{0.32\textwidth}
        \centering
        \includegraphics[width=1.0\textwidth]{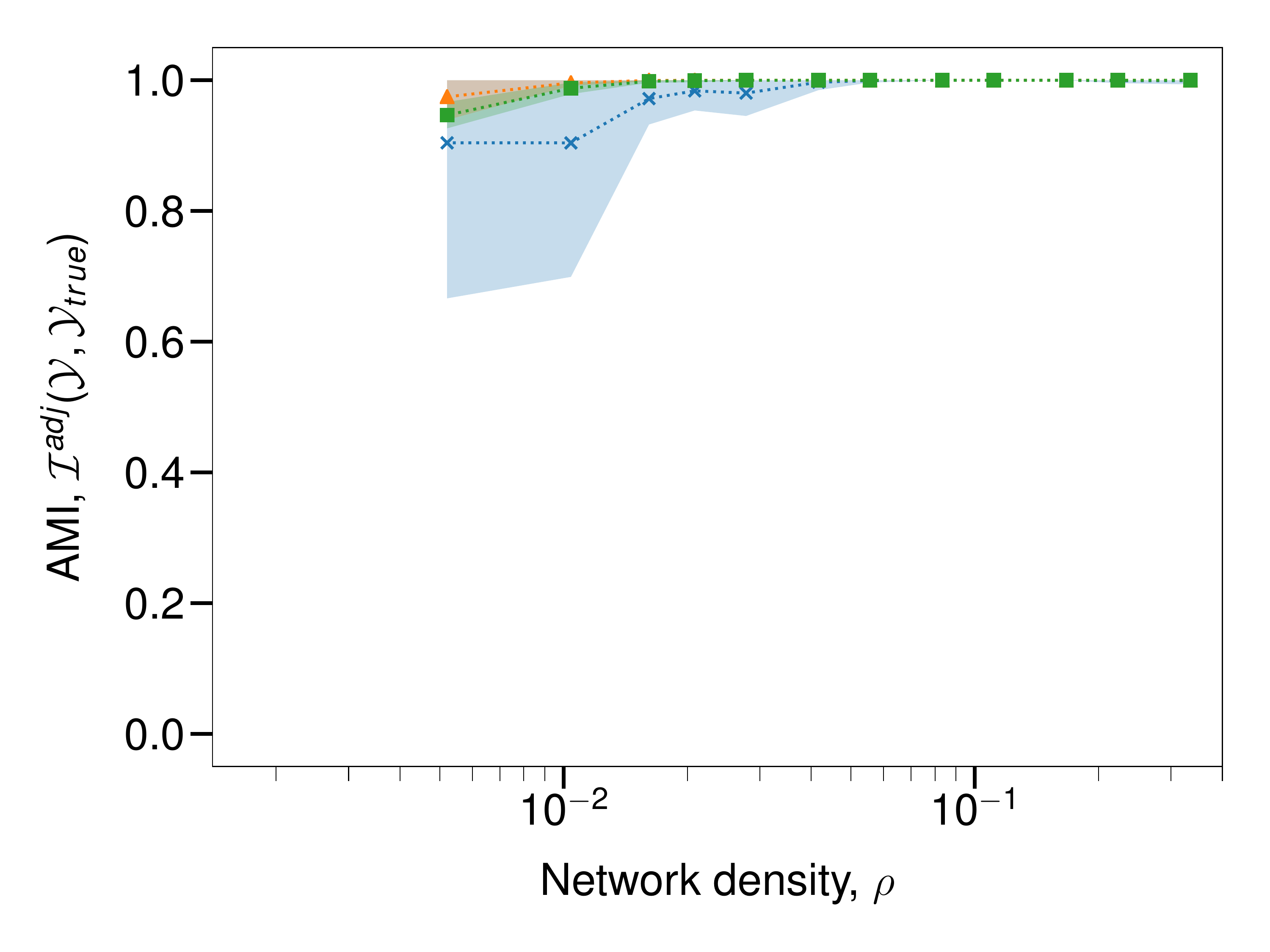}
        \caption{$\kavg = 50$, $\mu = 0.35$.}
    \end{subfigure}
    \\
    \begin{subfigure}{0.32\textwidth}
        \centering
        \includegraphics[width=1.0\textwidth]{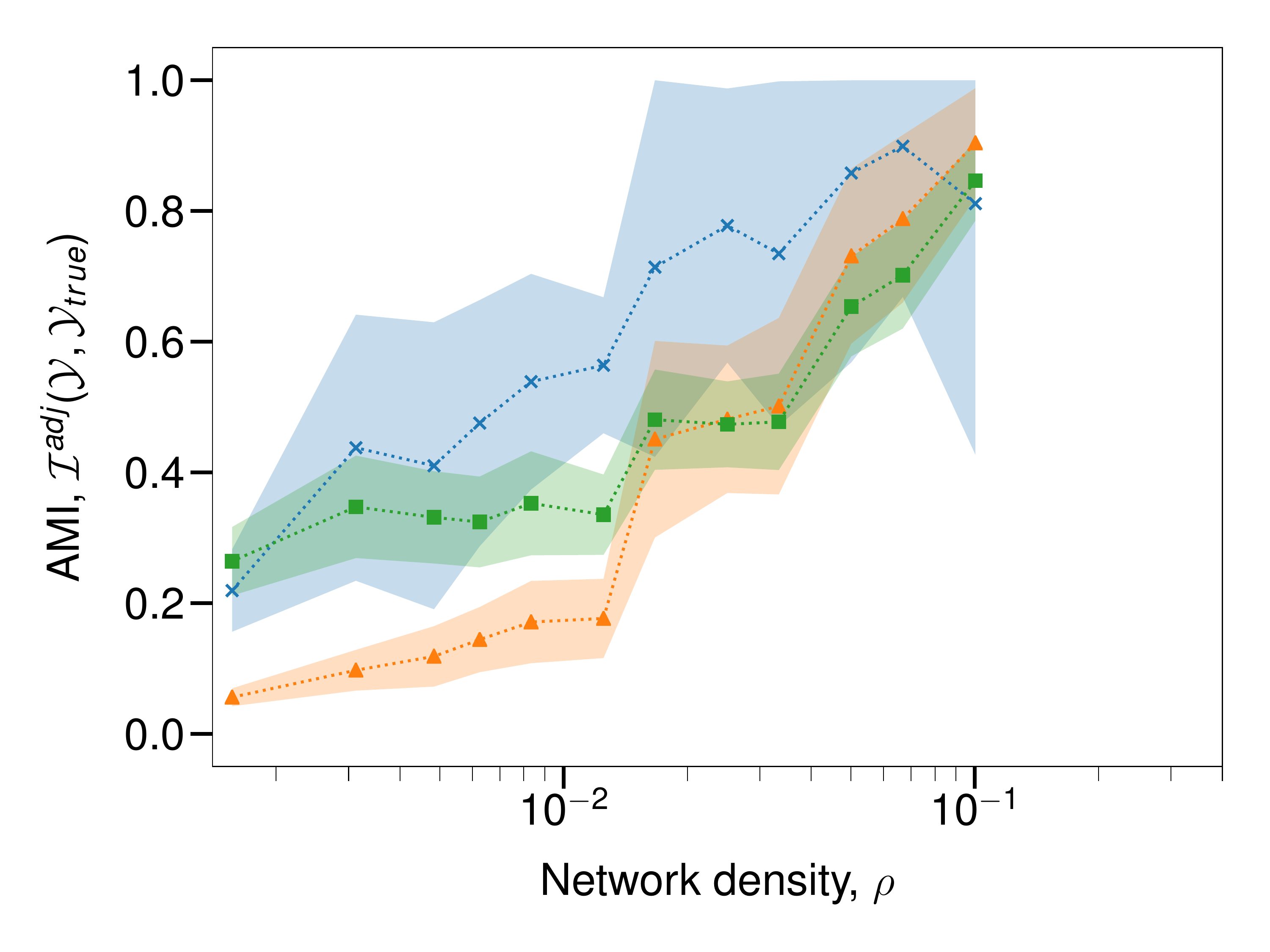}
        \caption{$\kavg = 15$, $\mu = 0.45$.}
    \end{subfigure}
    \begin{subfigure}{0.32\textwidth}
        \centering
        \includegraphics[width=1.0\textwidth]{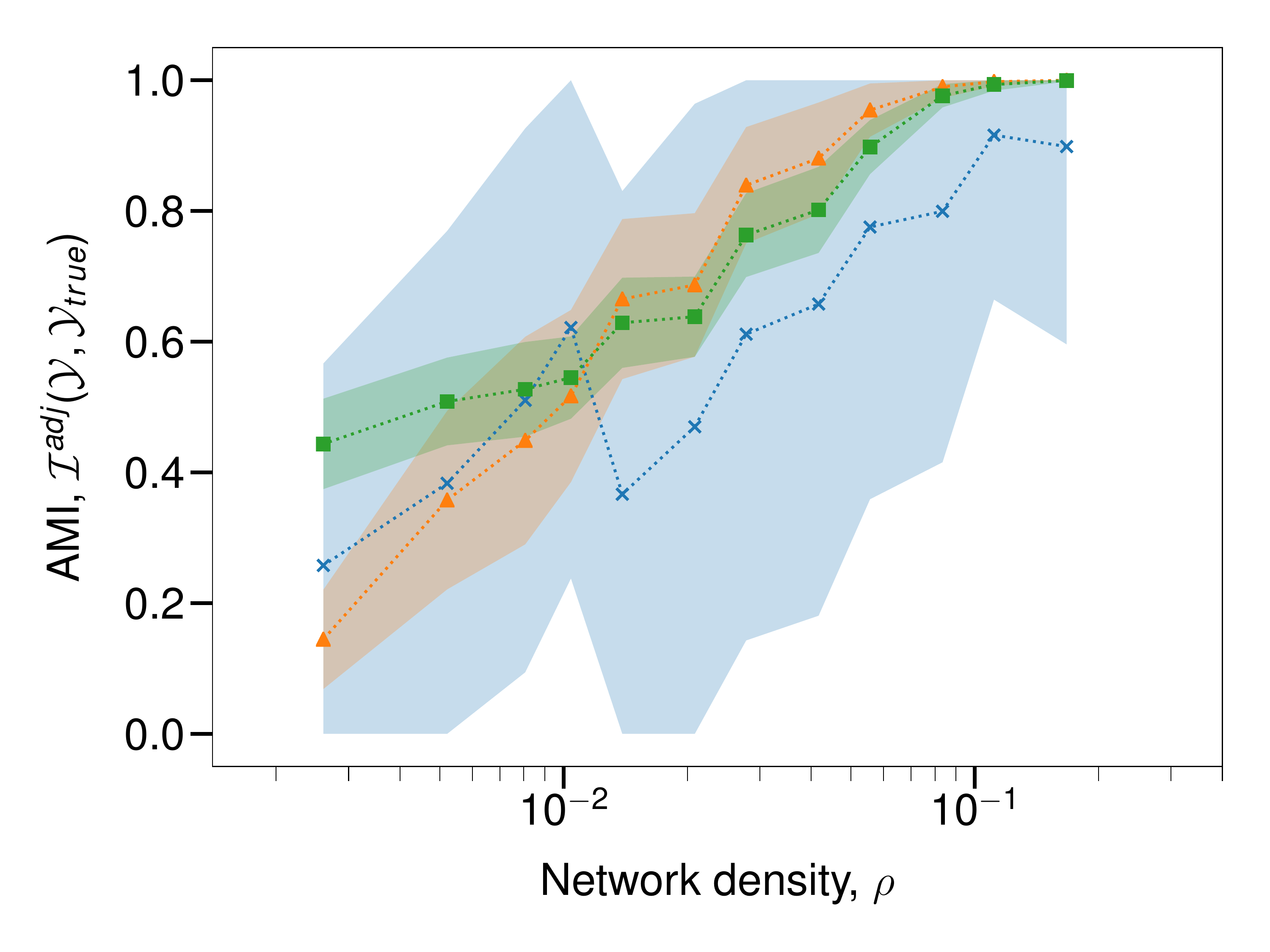}
        \caption{$\kavg = 25$, $\mu = 0.45$.}
    \end{subfigure}
    \begin{subfigure}{0.32\textwidth}
        \centering
        \includegraphics[width=1.0\textwidth]{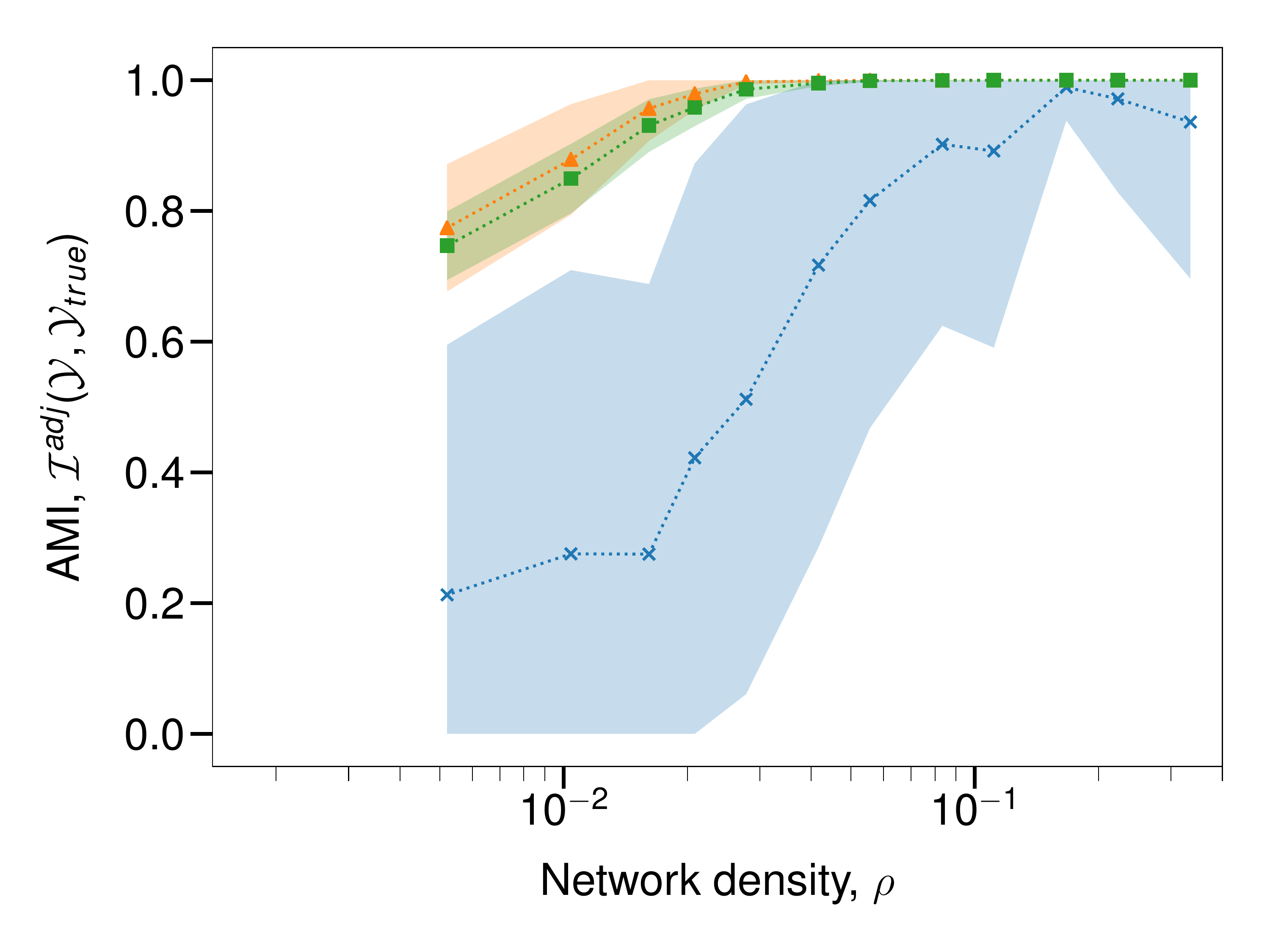}
        \caption{$\kavg = 50$, $\mu = 0.45$.}
    \end{subfigure}
    \\
    \begin{subfigure}{0.32\textwidth}
        \centering
        \includegraphics[width=1.0\textwidth]{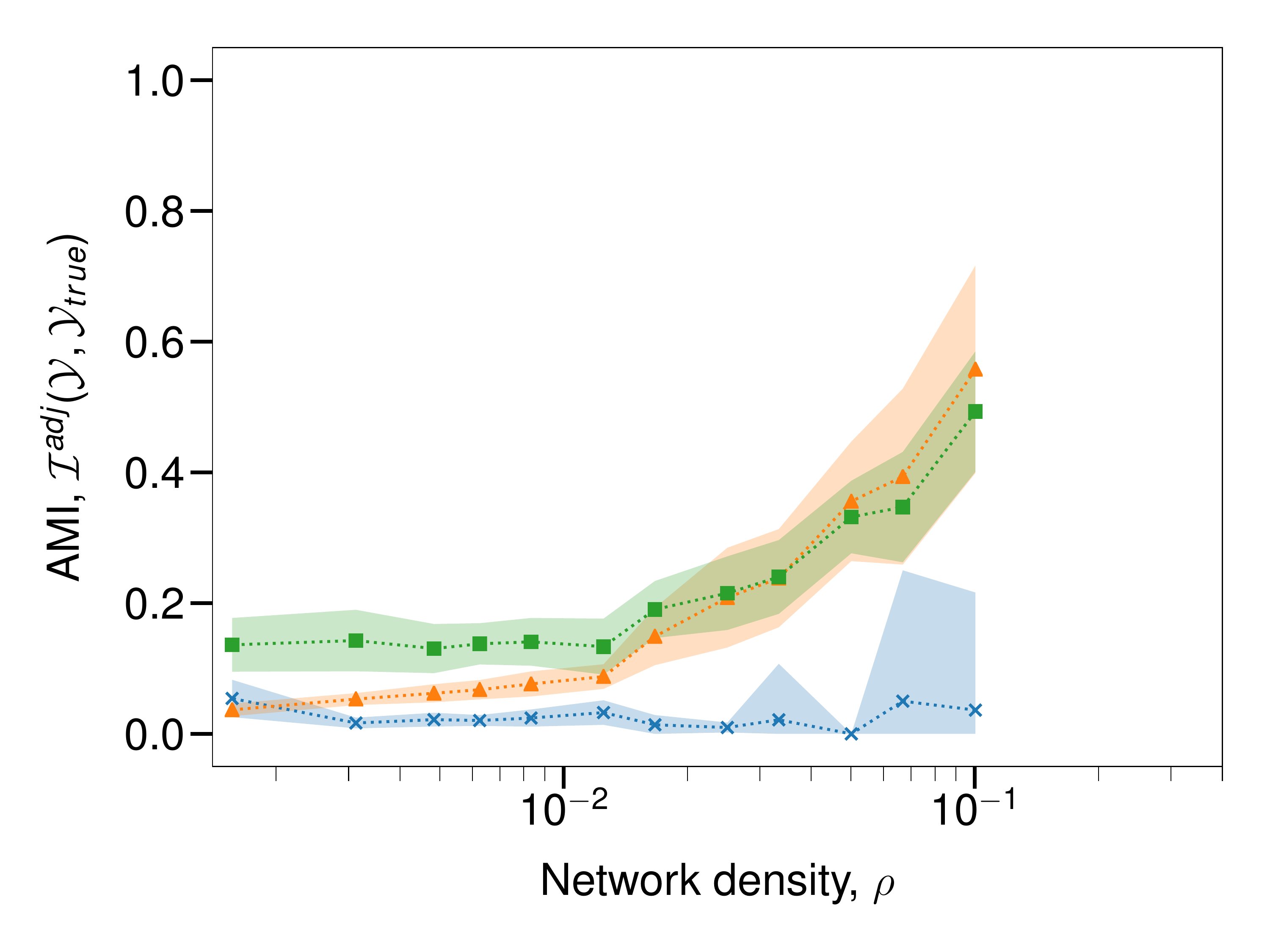}
        \caption{$\kavg = 15$, $\mu = 0.55$.}
    \end{subfigure}
    \begin{subfigure}{0.32\textwidth}
        \centering
        \includegraphics[width=1.0\textwidth]{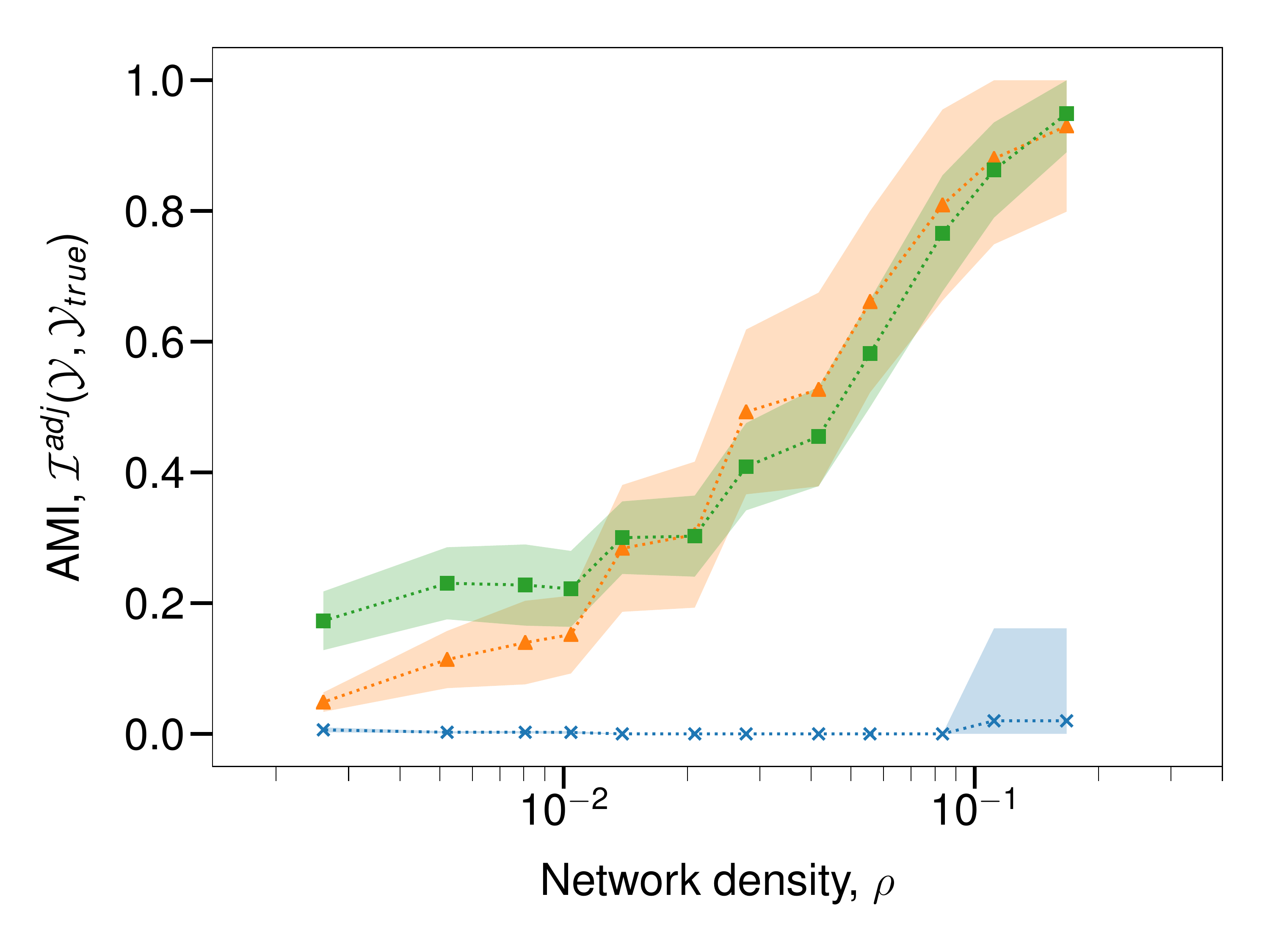}
        \caption{$\kavg = 25$, $\mu = 0.55$.}
    \end{subfigure}
    \begin{subfigure}{0.32\textwidth}
        \centering
        \includegraphics[width=1.0\textwidth]{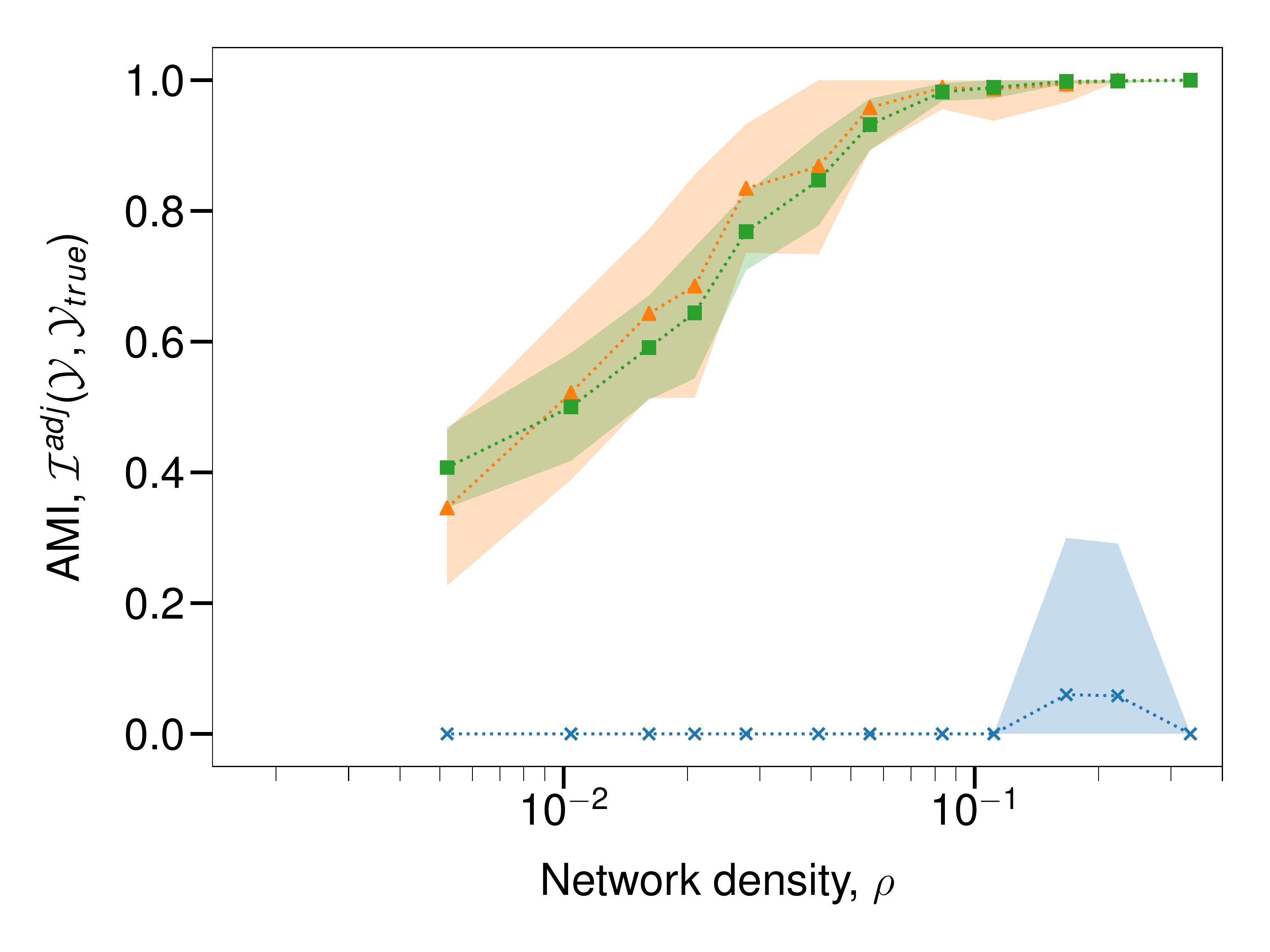}
        \caption{$\kavg = 50$, $\mu = 0.55$.}
    \end{subfigure}
    \caption{Comparison of Infomap, Synwalk and Walktrap on LFR benchmark networks with given average degree and mixing parameter. 
    The lines and shaded areas show the mean and standard deviation of AMI as a function of the network density (logarithmic scale), obtained from $50$ different network realizations. Performance of Synwalk and Walktrap rises with increasing network density, irrespective of the mixing parameter.}\label{fig:lfr_ami_vs_n}
\end{figure*}
	
\section{Cluster Matching Algorithm}\label{appx:cluster_matching}
Consider the ground truth clustering $\clutrue = \{\dom{Y}_i | i \in\idxset{\clutrue}\}$ and a prediction (by any community detection algorithm) $\clupred = \{\dom{Y}_i | i \in\idxset{\clupred}\}$ with index sets $\idxset{\clutrue}$ and $\idxset{\clupred}$ respectively. Then, in general, not only will the index sets differ in size but moreover they will encode equal or highly similar clusters in both clusterings differently. Hence, a matching between the index sets is necessary in our analyses (see Section~\ref{sssec:node_statistics}).

For this purpose we employ the following greedy matching algorithm. The matching is based on the contingency table between the two clusterings $\clutrue$, $\clupred$. In every step we add an entry to the cluster mapping by greedily picking the clusters with maximum overlap from the contingency table that do not already have a match. The procedure stops whenever all clusters of any input clustering have a match. The residual clusters obtain no mapping.

\begin{algorithm}[t]
\SetAlgoLined
\caption{Greedy Cluster Matching Algorithm}
\SetKwFunction{ctable}{contingency\_table}
\SetKwFunction{keys}{keys}
\SetKwFunction{values}{values}
\SetKw{or}{or}
\KwIn{True clustering $\clutrue$, predicted clustering $\clupred$}
\KwOut{Dictionary M: $\idxset{\clutrue} \mapsto \idxset{\clupred}$}
 Initialize the dictionary $M \leftarrow \emptyset$  \\
 Initialize the number of true clusters $K^{true} = \max \idxset{\clutrue}$ \\
 Initialize the number of predicted clusters $K^{pred} = \max \idxset{\clupred}$    \\
 Generate the contingency table $C \leftarrow \ctable{$\clutrue, \clupred$}$  \\
 \For{$i\leftarrow 1$ \KwTo $\min \{K^{true}, K^{pred}\}$}{
  $k^{true} \leftarrow -1$    \\
  $k^{pred} \leftarrow -1$    \\
  \While{$k^{true} < 0$ \or $k^{true} \in M.\keys{}$ \or $k^{pred} \in M.\values{}$}{
    $k^{true}, k^{pred} \leftarrow \argmax{C}$  \\
    $C[k^{true}, k^{pred}] \leftarrow -1$    \\
  }
  $M[k^{true}] \leftarrow k^{pred}$
 }
\end{algorithm}
	\section{Additional Results on the Empirical Networks}\label{appx:additional_empirical}

In this appendix we provide additional cluster and node statistics for the empirical networks considered in Section~\ref{sec:realworld}. The considered cluster properties are a subset of previously used measures for cluster characterization by~\citet{Leskovec2010} and~\citet{Yang2015}. Consider a cluster $S \subseteq \dom{X}$ within an undirected network $\graph = (\dom{X}, E)$ and let
\begin{align}
    m_s = \frac{1}{2}  \cdot \left | \{(\alpha, \beta)\in E \,|\, \alpha,\beta\in S\} \right |
\end{align}
denote the number of internal links of $S$, and
\begin{align}
    c_s = \left | \{(\alpha, \beta)\in E \,|\, \alpha\in S \wedge \beta\notin S\} \right |
\end{align}
denote the number of external links, i.e., links connecting nodes within $S$ to nodes outside of $S$. Then we define the following cluster properties.
\begin{itemize}
    \item \emph{Cluster density}. The cluster density
    \begin{align}
        \rho(S) = \frac{m_s}{\binom{|S|}{2}}
    \end{align} is the density of internal links in $S$.
    
    \item \emph{Clustering coefficient}. The clustering coefficient
    \begin{align}
        c(S) = \frac{1}{|S|} \sum_{\alpha\in S} c_\alpha
    \end{align} is the average of all node clustering coefficients $c_\alpha$. Let $neigh(\alpha)$ denote the neighborhood of node $\alpha$, i.e., all nodes that are connected to $\alpha$ by a link. Then the clustering coefficient for some node $\alpha$ is the fraction of realized triangles including $\alpha$, i.e.,
    \begin{align}
        c_\alpha = \frac{\frac{1}{2} \cdot \left | \{(\alpha, \beta)\in E \,|\, \alpha,\beta\in neigh(\alpha)\} \right |}{ \binom{|neigh(\alpha)|}{2}}.
    \end{align}

    \item \emph{Conductance}. The conductance
    \begin{align}
        \kappa(S) = \frac{c_s}{m_s + c_s}
    \end{align}
    gives the fraction of external links to the total number of cluster edges.

    \item \emph{Cut ratio}. The cut ratio
    \begin{align}
        \xi(S) = \frac{c_s}{|S| \cdot (|\dom{X}| - |S|)}
    \end{align}
    is the ratio of external links to all possible external links.
    
\end{itemize}

As already observed in Section~\ref{sec:realworld}, in general the cluster property statistics yielded by Synwalk and Infomap are highly similar, whereas Walktrap is mostly clearly distinguishable. The cluster densities (Figure~\ref{fig:empirical_densities}) are inverse proportionally distributed to the cluster sizes, e.g., Walktrap detects many small communities with high densities. Clustering coefficients (Figure~\ref{fig:empirical_clustering_coefficients}) are higher for Walktrap on the github and wordnet networks when compared to Synwalk and Infomap. Synwalk predicts highly conductive clusters for the github network (see conductance distributions in Figure~\ref{fig:empirical_conductances}). 
Walktrap yields significantly smaller cut ratios (Figure~\ref{fig:empirical_cut_ratios}) compared to Infomap and Synwalk on the lastfm-asia and pennsylvania-roads networks. Distributions of the mixing parameters are displayed in Figure~\ref{fig:empirical_mixing_parameters} and show no notable difference between the three methods. An exception is the github network, where Synwalk predicts a clustering such that most of the nodes exhibit a mixing parameter $\gtrsim 0.5$. This reflects the high conductance values in Figure~\ref{fig:empirical_conductances}.

\begin{figure*}[t]
    \centering
    \begin{subfigure}{0.32\textwidth}
        \centering
        \includegraphics[width=1.0\textwidth]{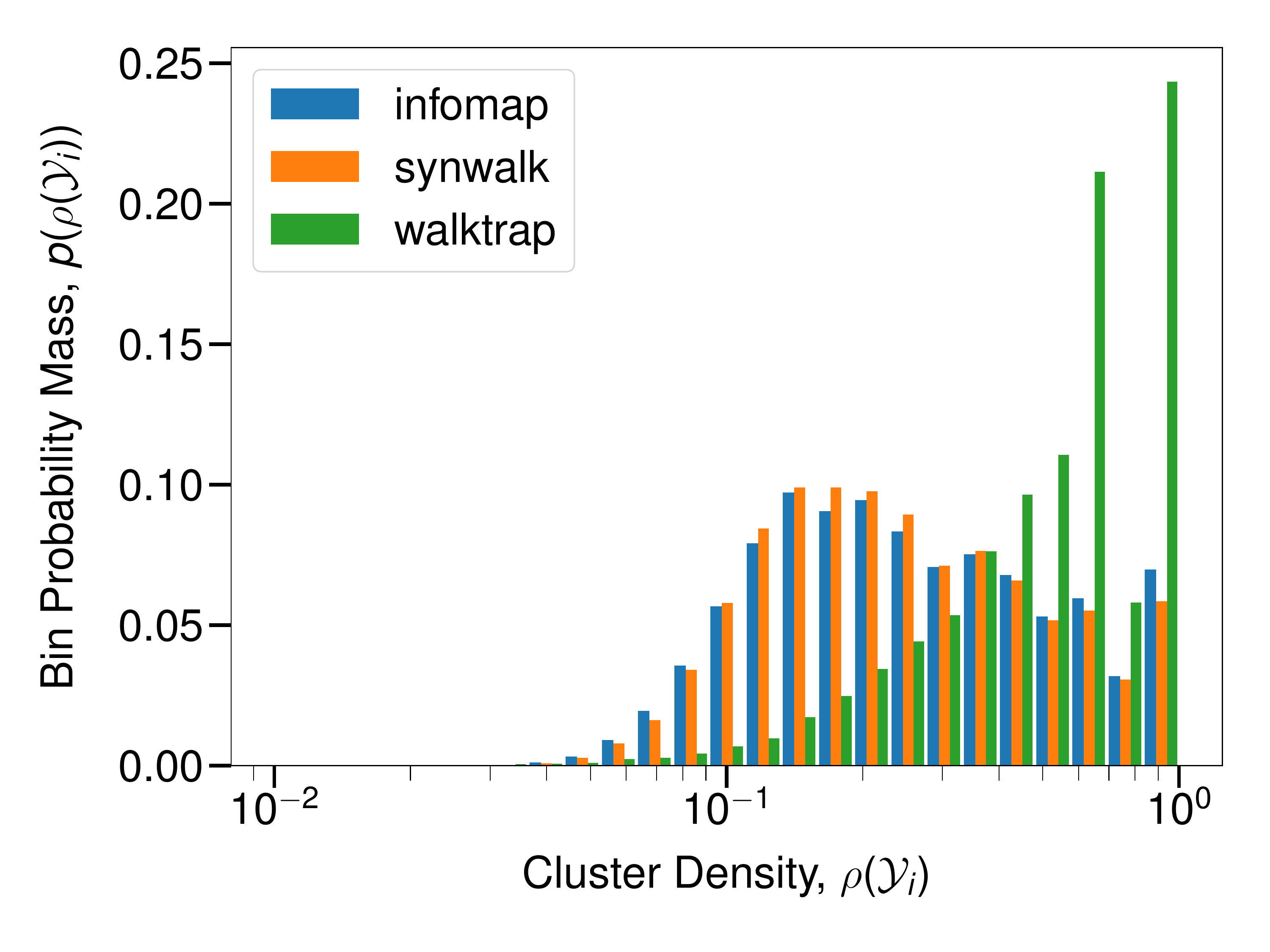}
        \caption{dblp.}
    \end{subfigure}
   \begin{subfigure}{0.32\textwidth}
        \centering
        \includegraphics[width=1.0\textwidth]{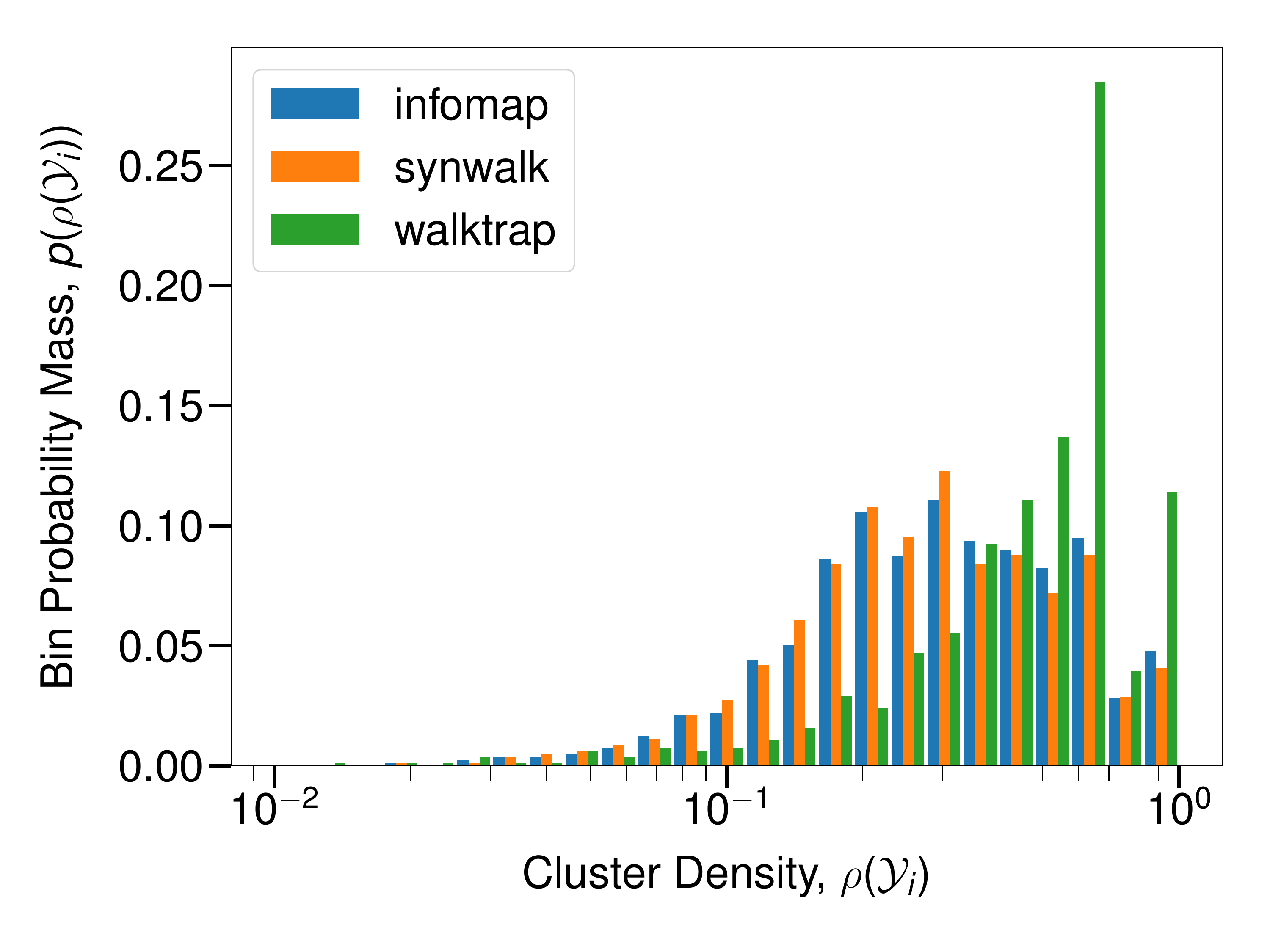}
        \caption{facebook.}
    \end{subfigure}
    \begin{subfigure}{0.32\textwidth}
        \centering
        \includegraphics[width=1.0\textwidth]{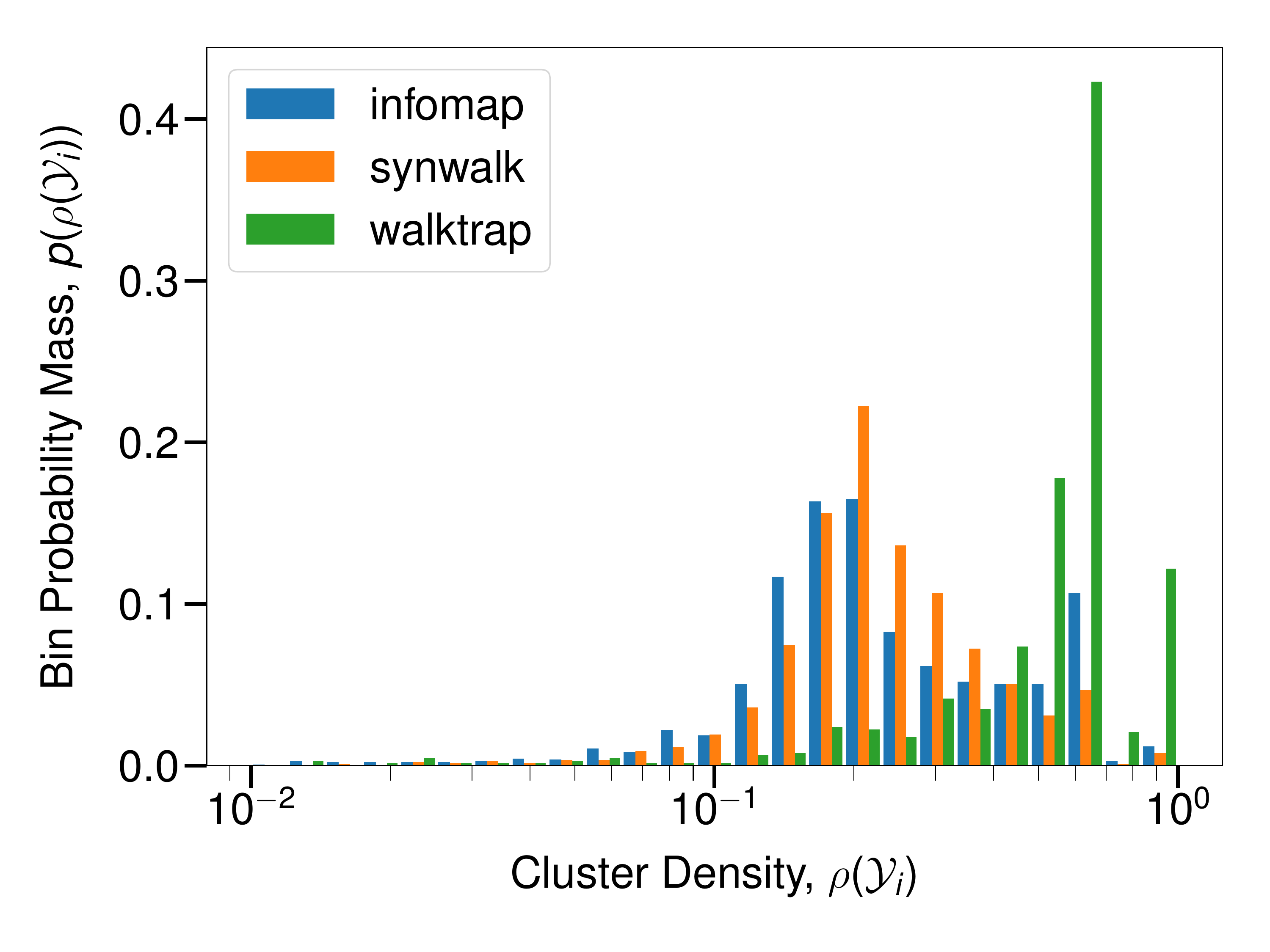}
        \caption{github.}
    \end{subfigure}
    \\
    \begin{subfigure}{0.32\textwidth}
        \centering
        \includegraphics[width=1.0\textwidth]{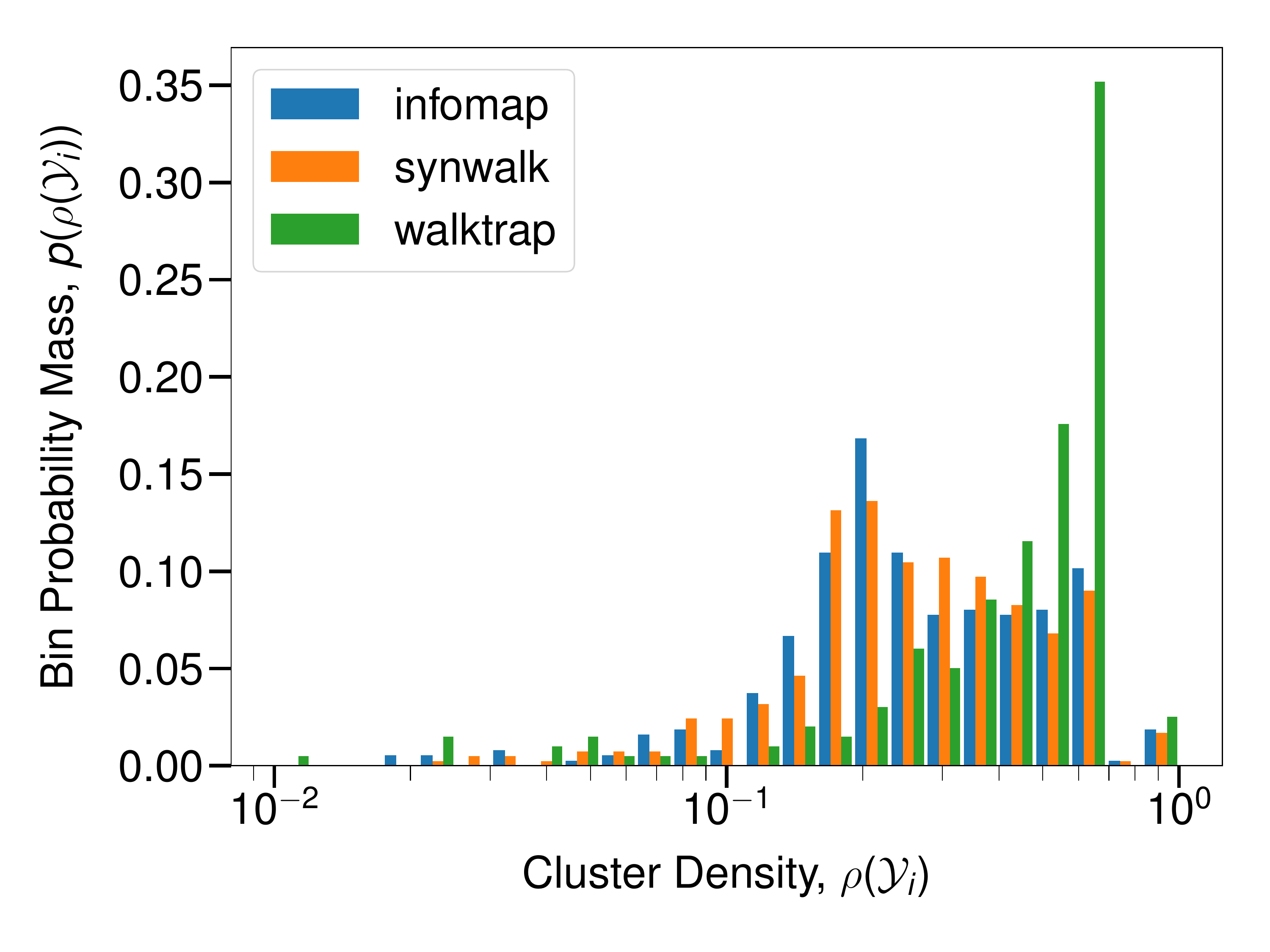}
        \caption{lastfm-asia.}
    \end{subfigure}
    \begin{subfigure}{0.32\textwidth}
        \centering
        \includegraphics[width=1.0\textwidth]{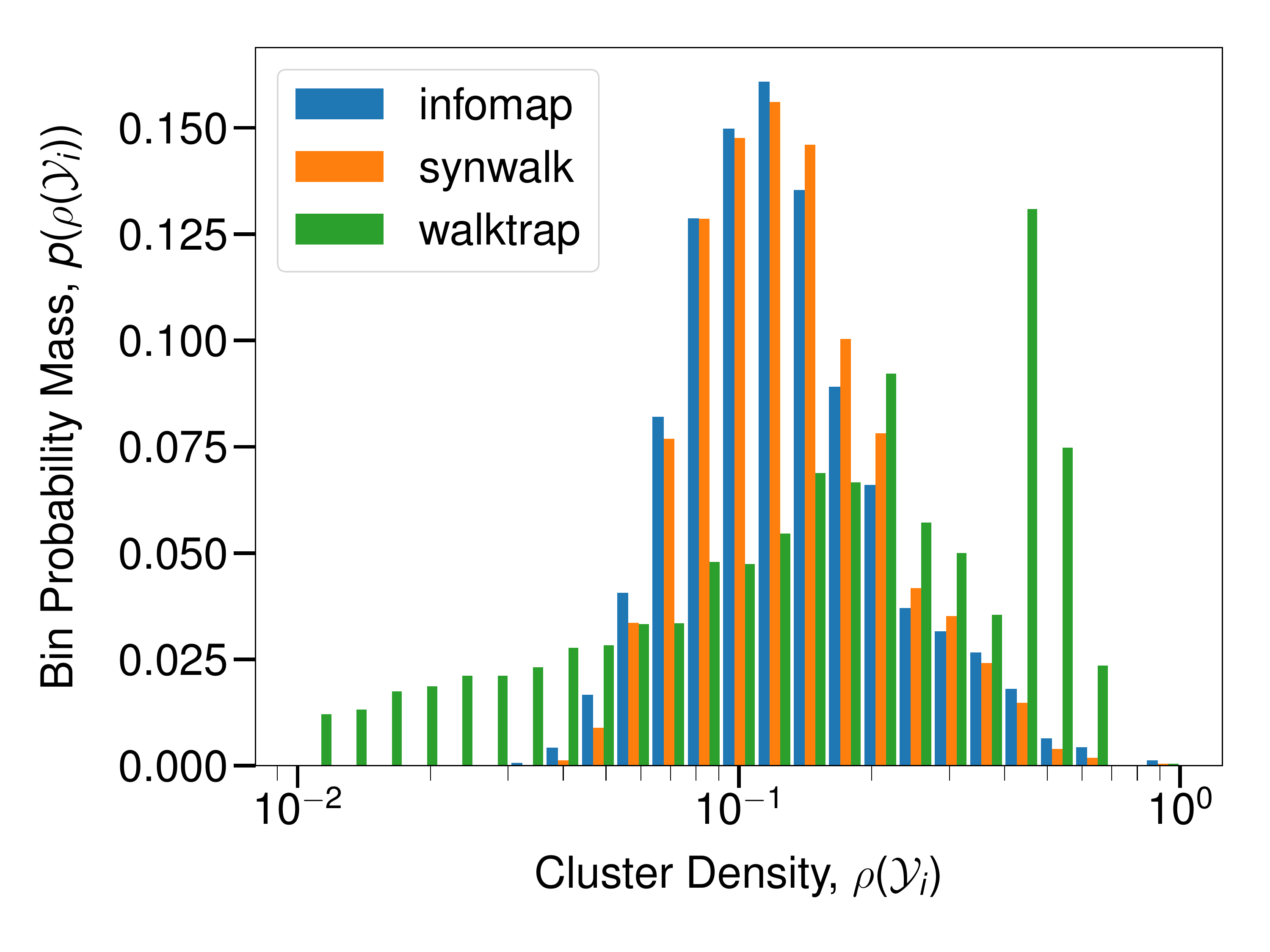}
        \caption{pennsylvania-roads.}
    \end{subfigure}
    \begin{subfigure}{0.32\textwidth}
        \centering
        \includegraphics[width=1.0\textwidth]{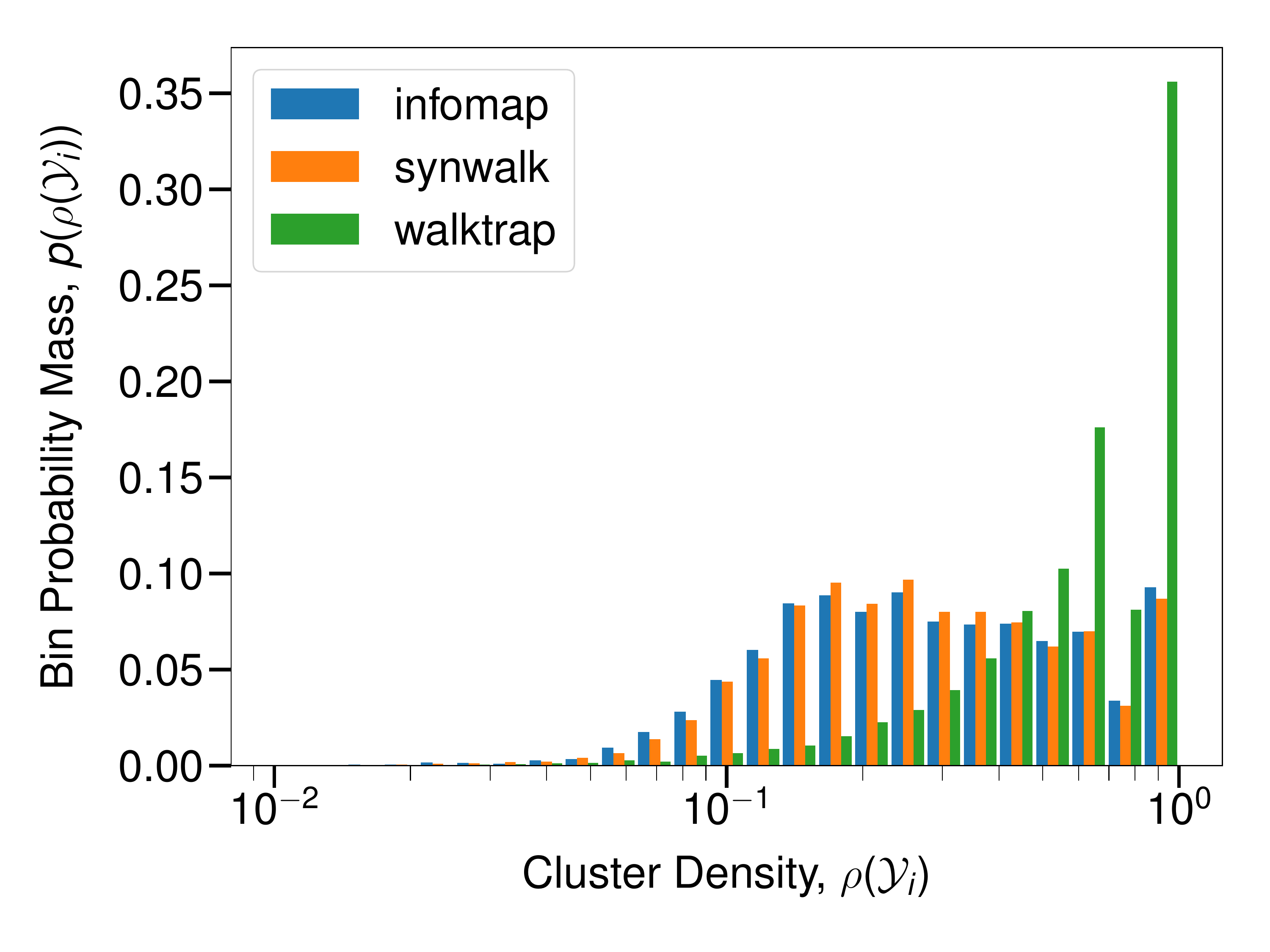}
        \caption{wordnet.}
    \end{subfigure}
    \caption{Distributions of cluster densities for the detection results on empirical networks.}\label{fig:empirical_densities}
\end{figure*}

\begin{figure*}[t]
    \centering
    \begin{subfigure}{0.32\textwidth}
        \centering
        \includegraphics[width=1.0\textwidth]{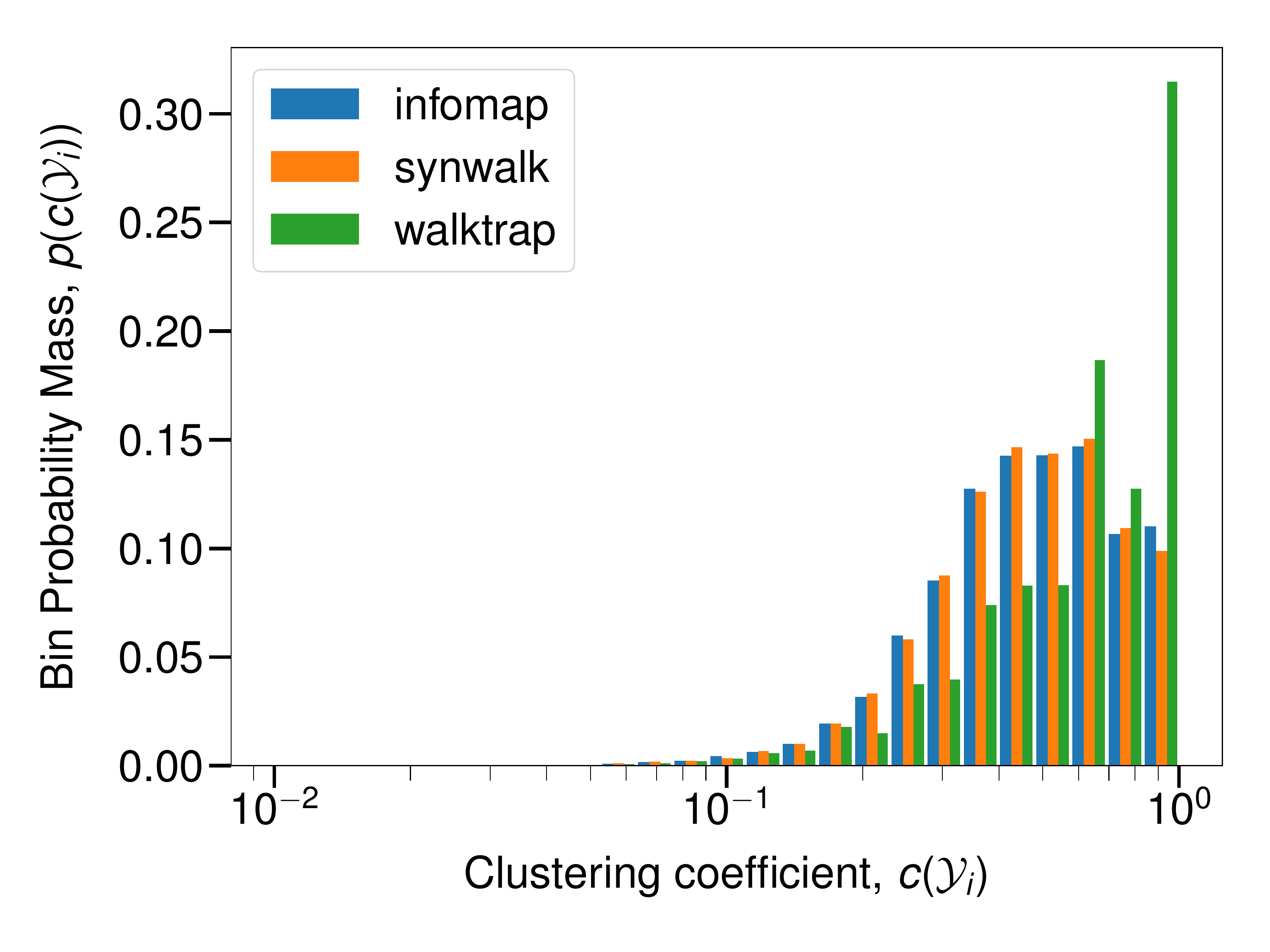}
        \caption{dblp.}
    \end{subfigure}
   \begin{subfigure}{0.32\textwidth}
        \centering
        \includegraphics[width=1.0\textwidth]{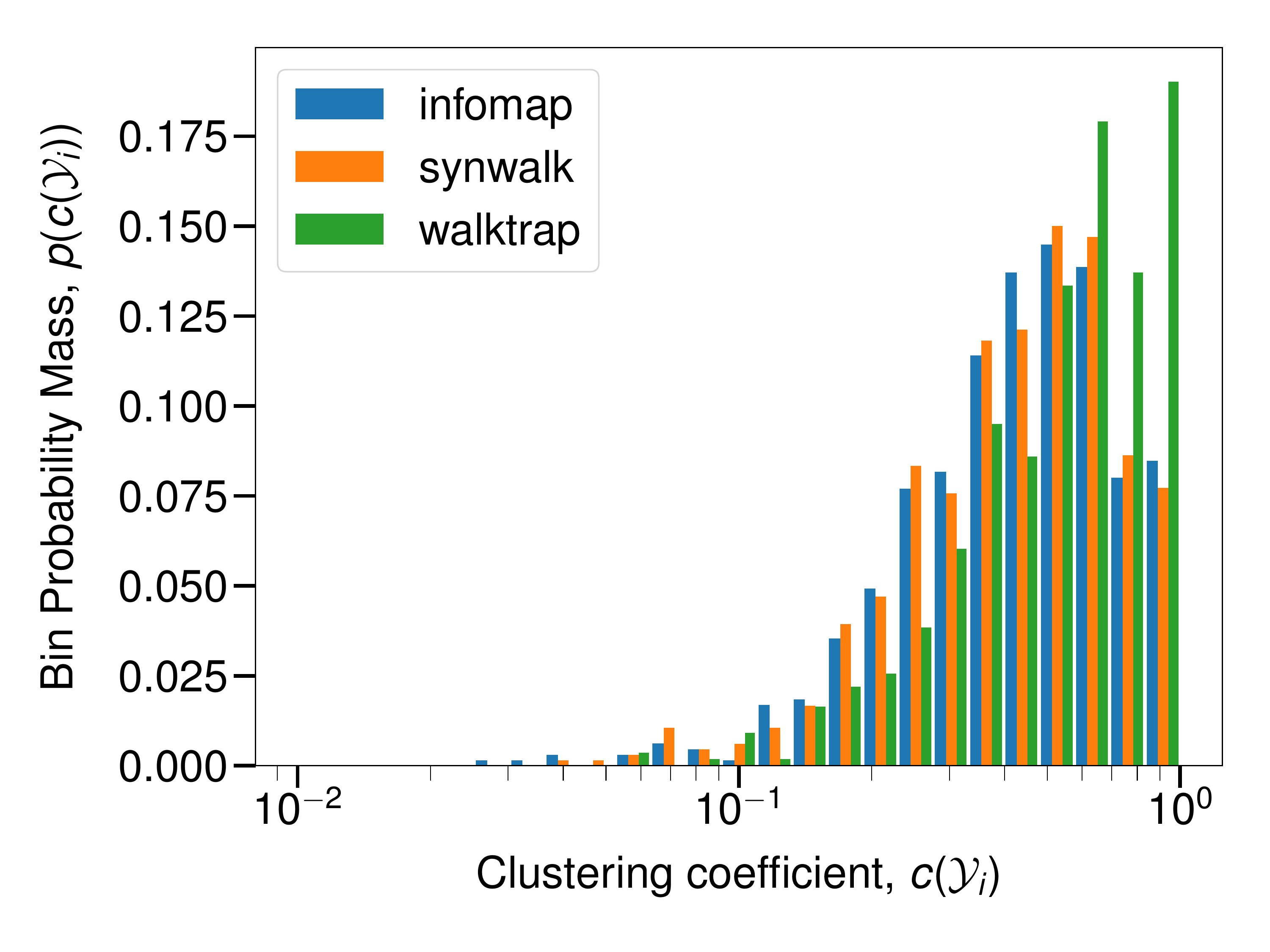}
        \caption{facebook.}
    \end{subfigure}
    \begin{subfigure}{0.32\textwidth}
        \centering
        \includegraphics[width=1.0\textwidth]{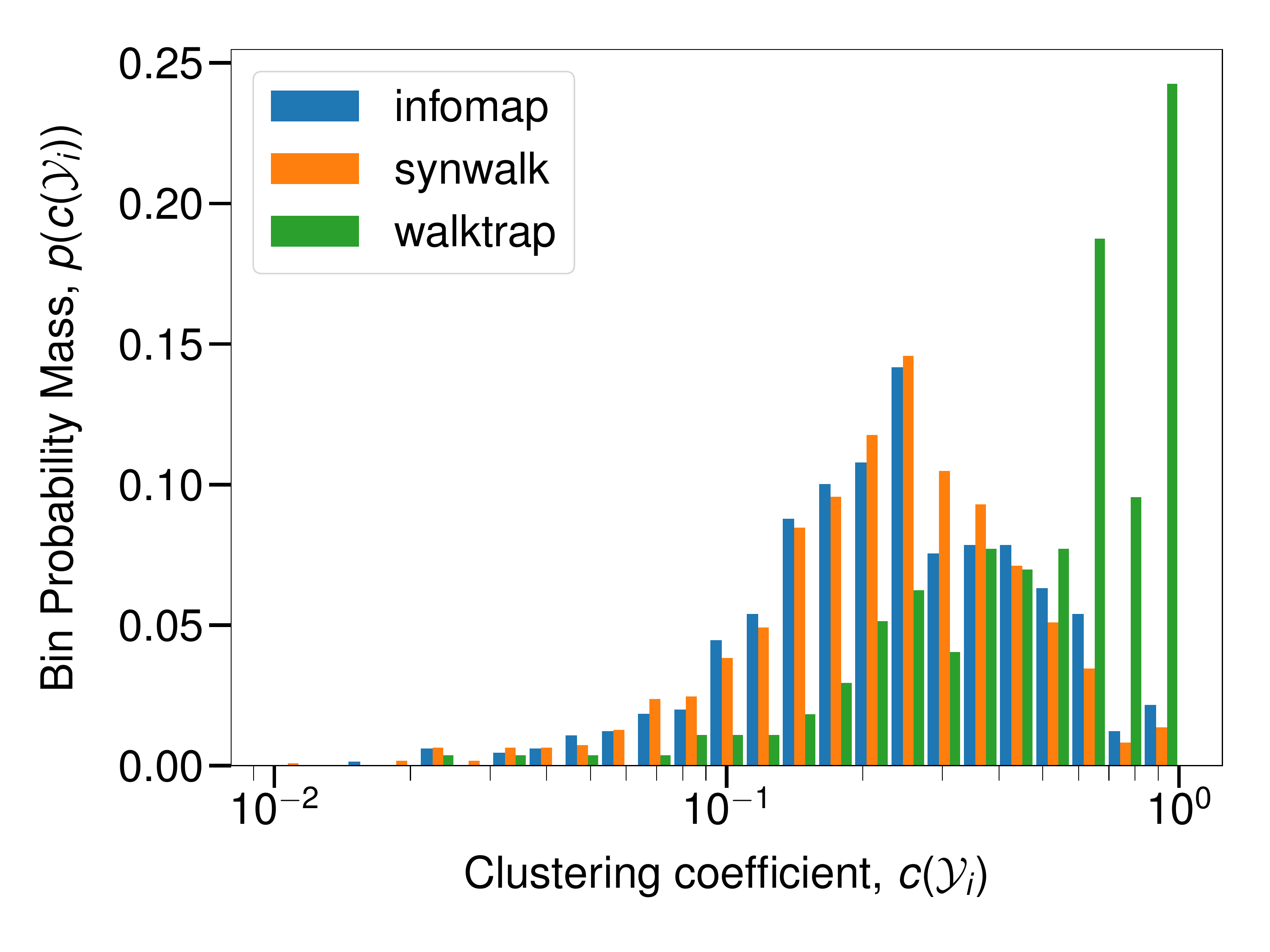}
        \caption{github.}
    \end{subfigure}
    \\
    \begin{subfigure}{0.32\textwidth}
        \centering
        \includegraphics[width=1.0\textwidth]{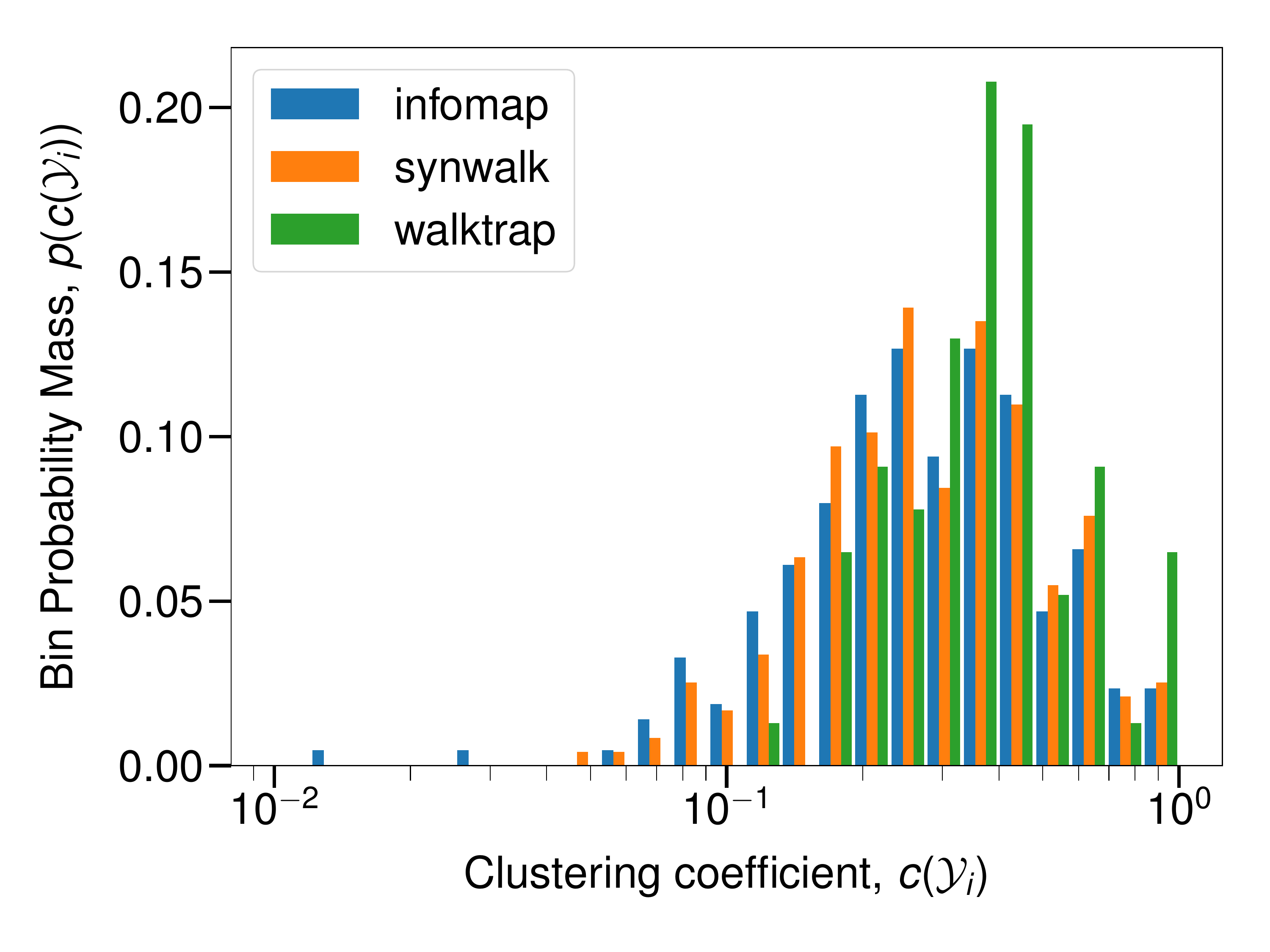}
        \caption{lastfm-asia.}
    \end{subfigure}
    \begin{subfigure}{0.32\textwidth}
        \centering
        \includegraphics[width=1.0\textwidth]{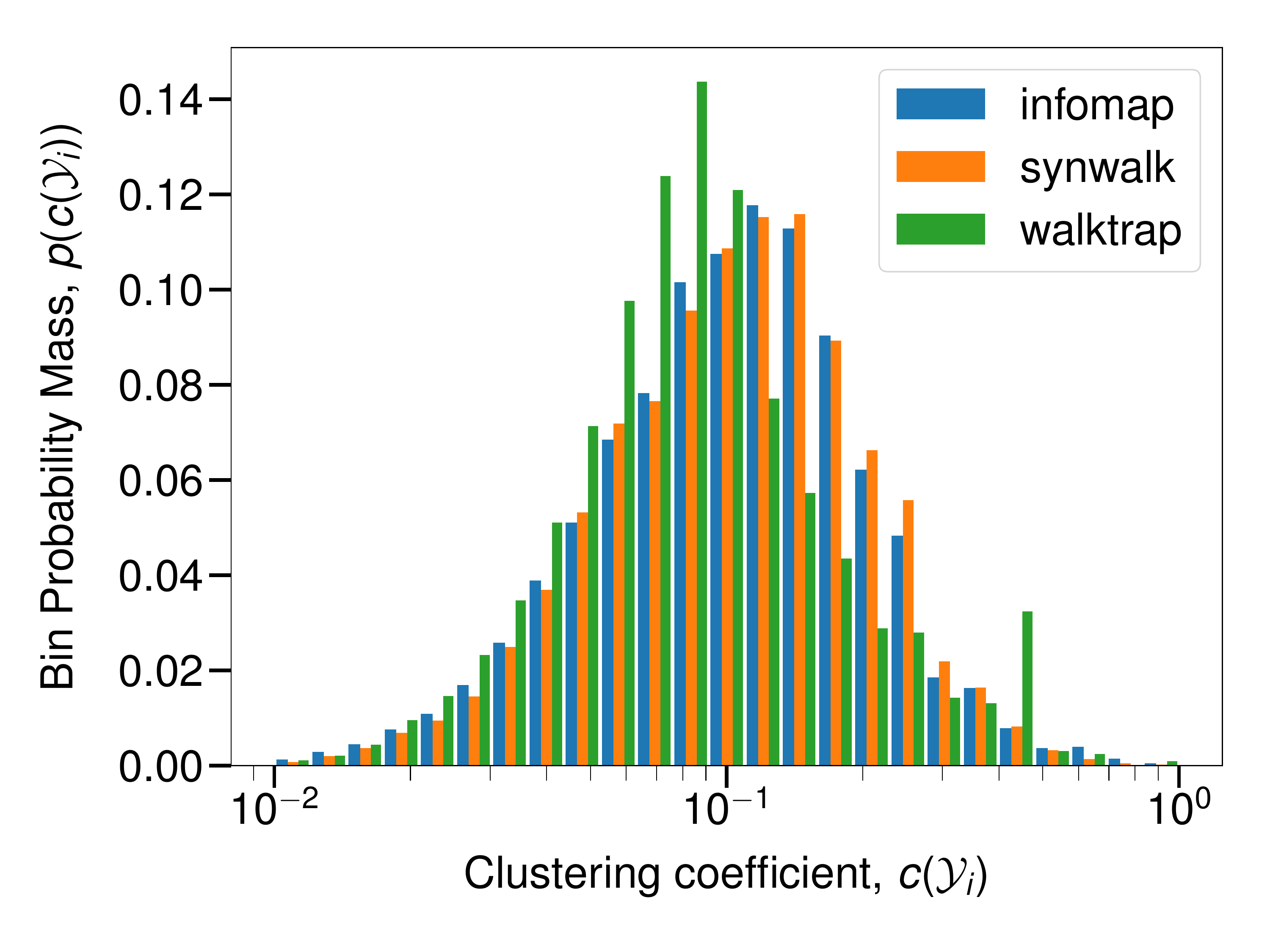}
        \caption{pennsylvania-roads.}
    \end{subfigure}
    \begin{subfigure}{0.32\textwidth}
        \centering
        \includegraphics[width=1.0\textwidth]{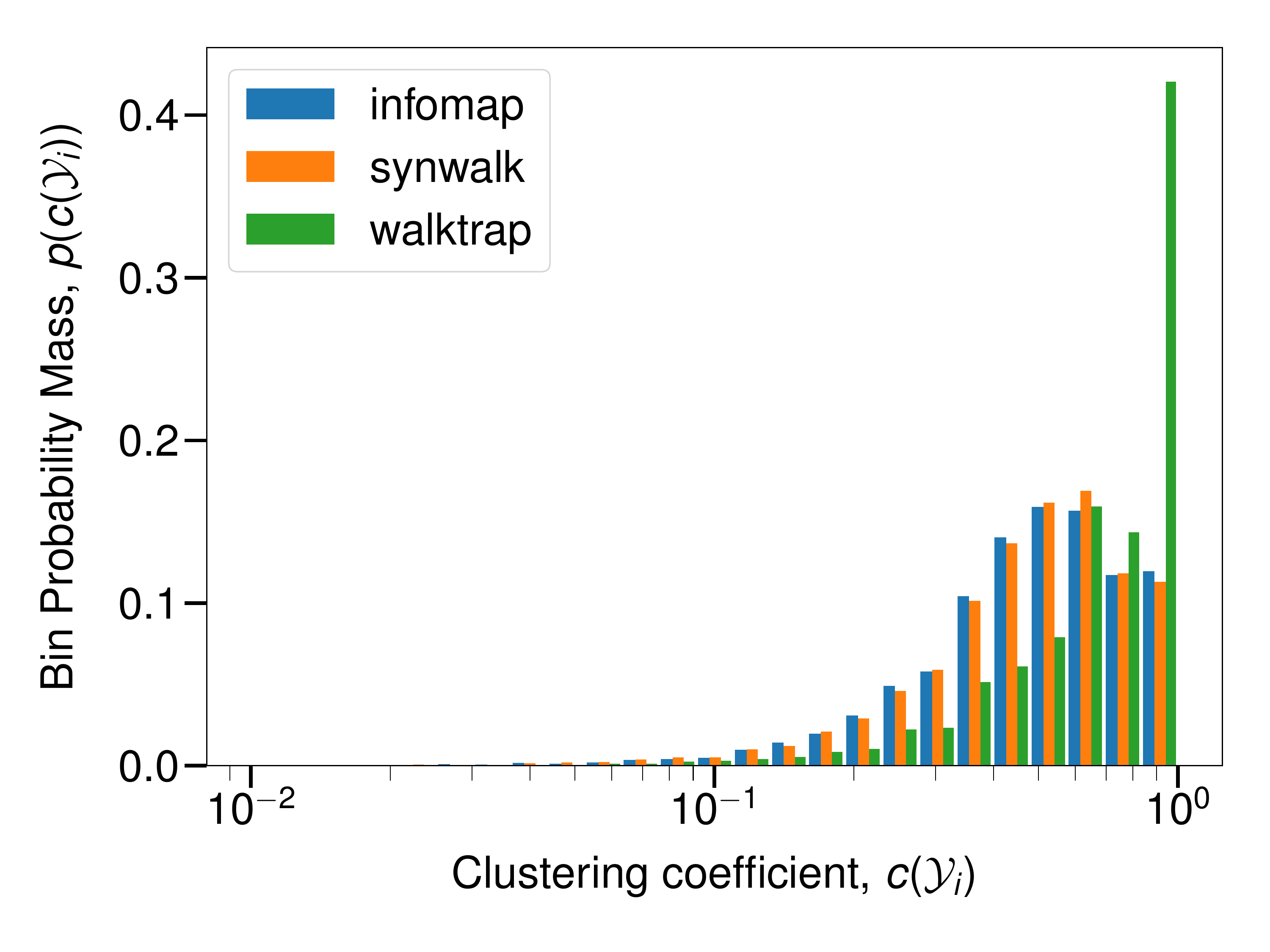}
        \caption{wordnet.}
    \end{subfigure}
    \caption{Distributions of cluster clustering coefficients for the detection results on empirical networks.}\label{fig:empirical_clustering_coefficients}
\end{figure*}

\begin{figure*}[t]
    \centering
    \begin{subfigure}{0.32\textwidth}
        \centering
        \includegraphics[width=1.0\textwidth]{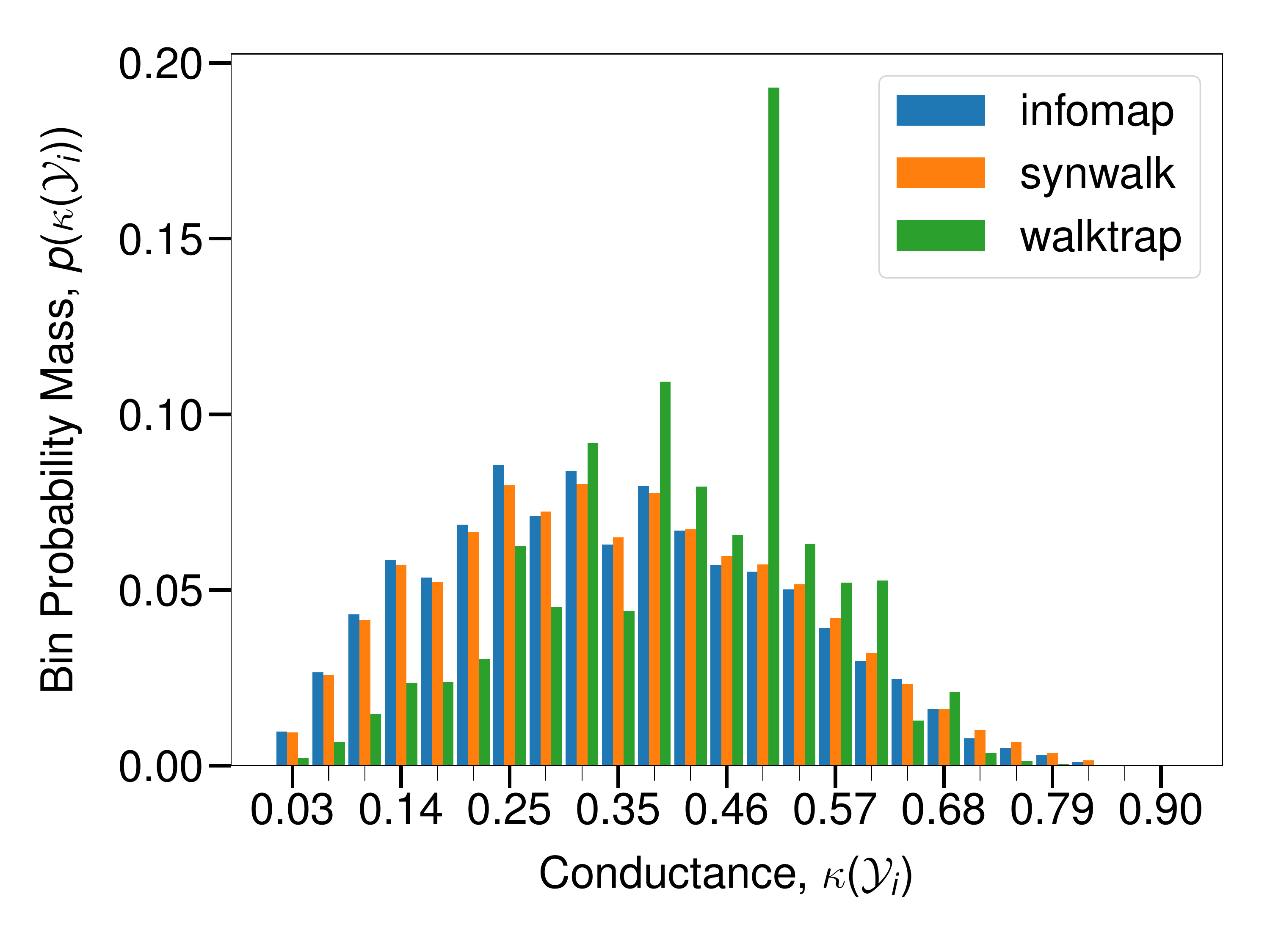}
        \caption{dblp.}
    \end{subfigure}
   \begin{subfigure}{0.32\textwidth}
        \centering
        \includegraphics[width=1.0\textwidth]{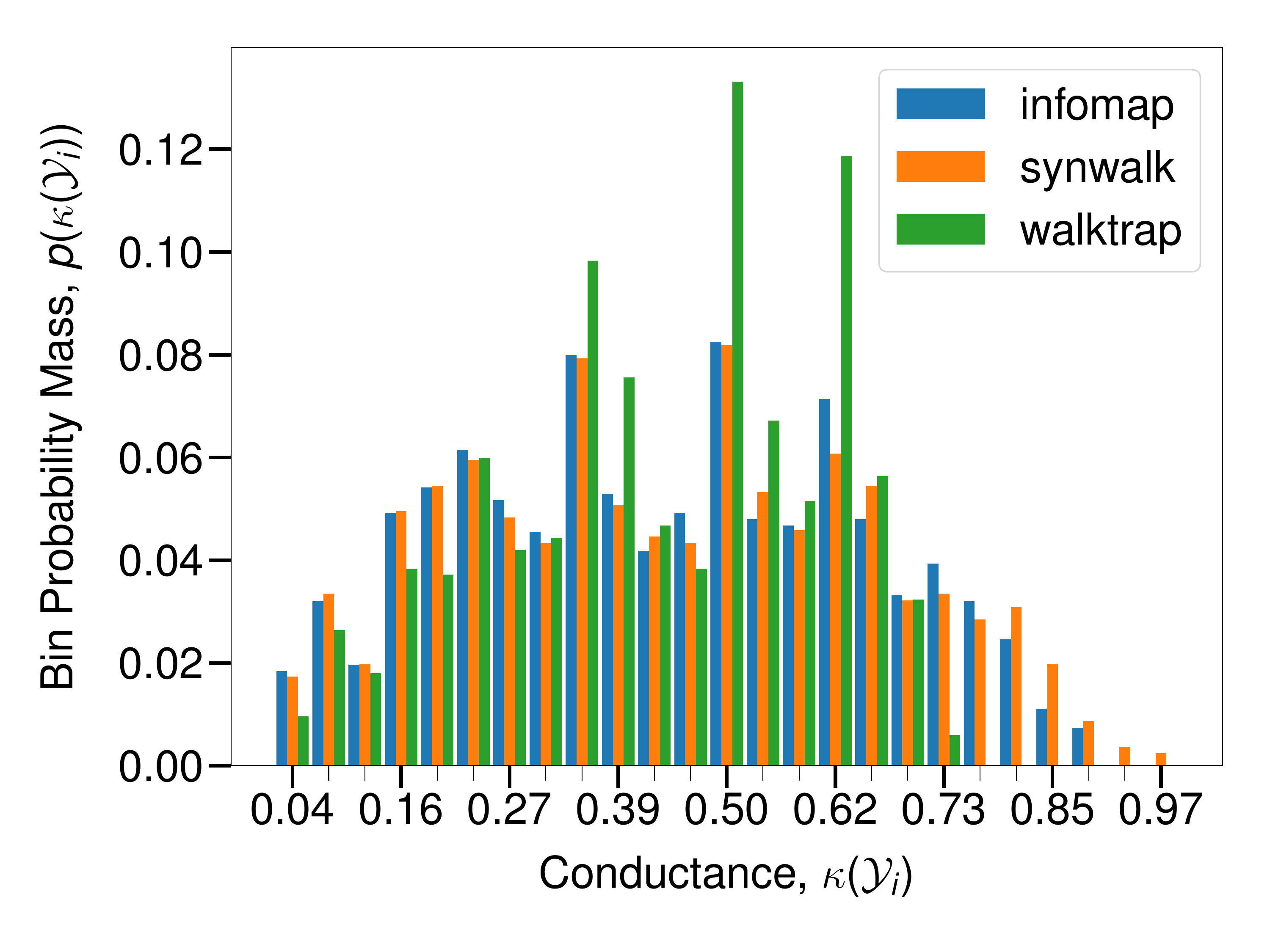}
        \caption{facebook.}
    \end{subfigure}
    \begin{subfigure}{0.32\textwidth}
        \centering
        \includegraphics[width=1.0\textwidth]{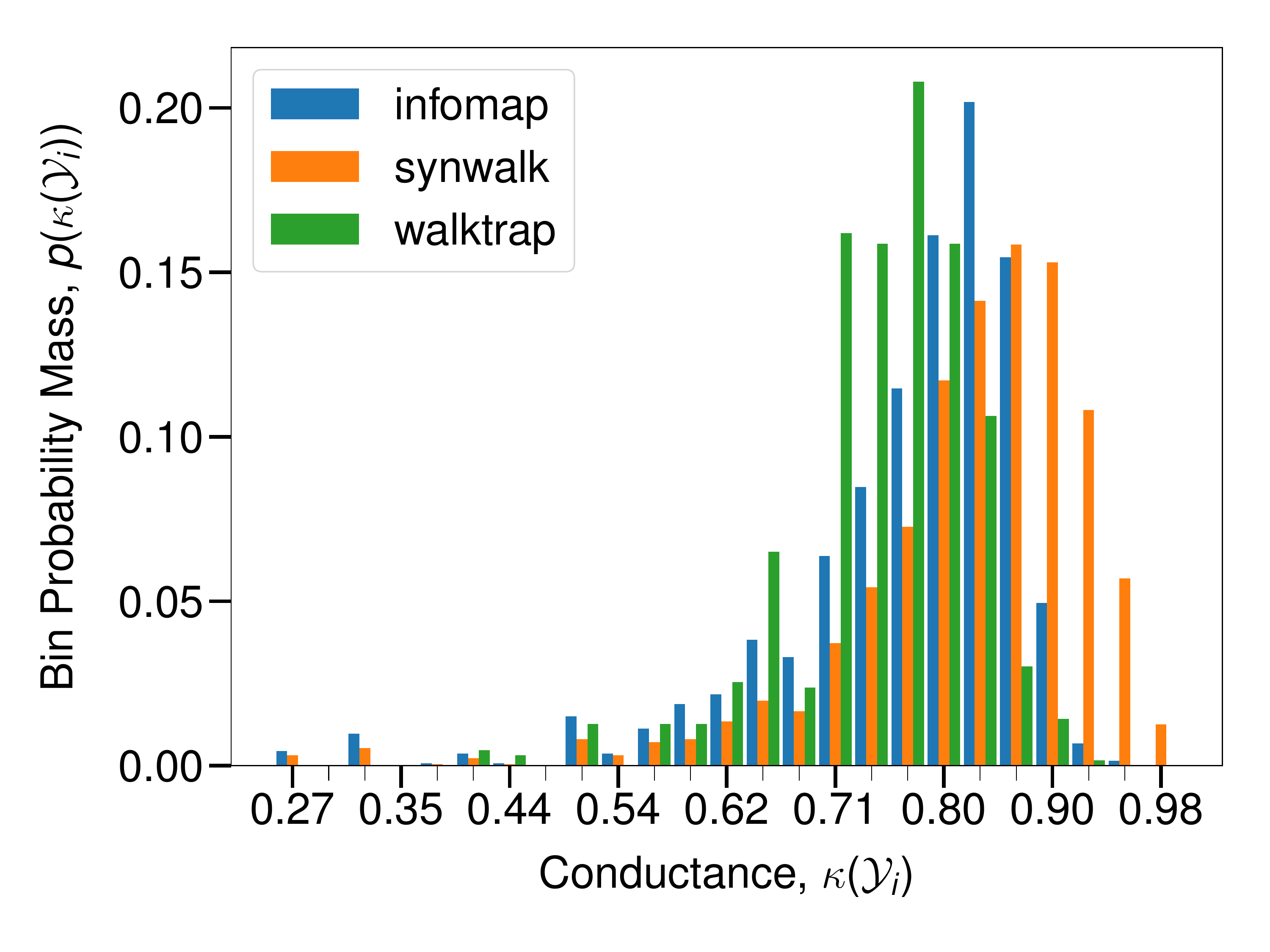}
        \caption{github.}
    \end{subfigure}
    \\
    \begin{subfigure}{0.32\textwidth}
        \centering
        \includegraphics[width=1.0\textwidth]{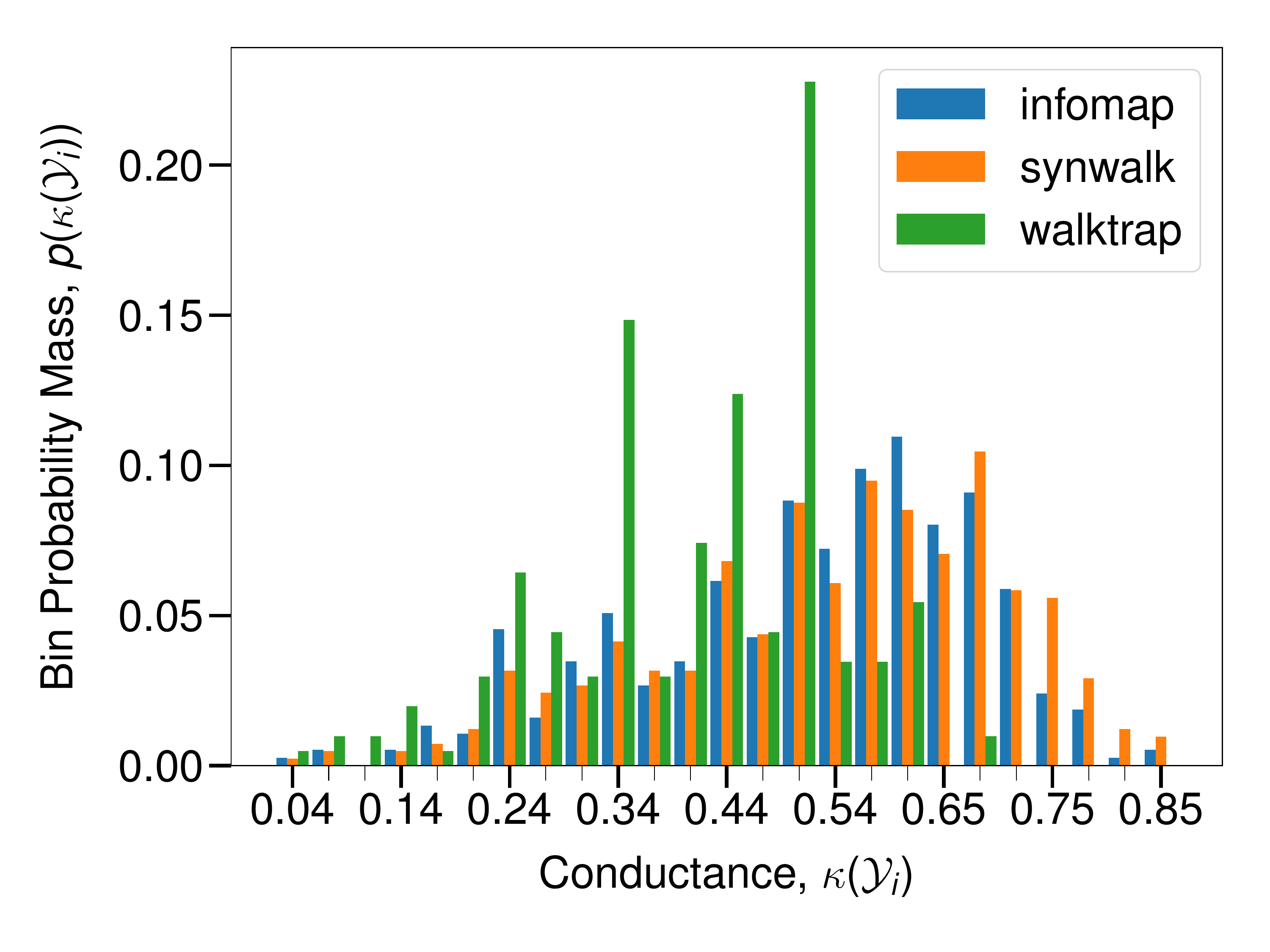}
        \caption{lastfm-asia.}
    \end{subfigure}
    \begin{subfigure}{0.32\textwidth}
        \centering
        \includegraphics[width=1.0\textwidth]{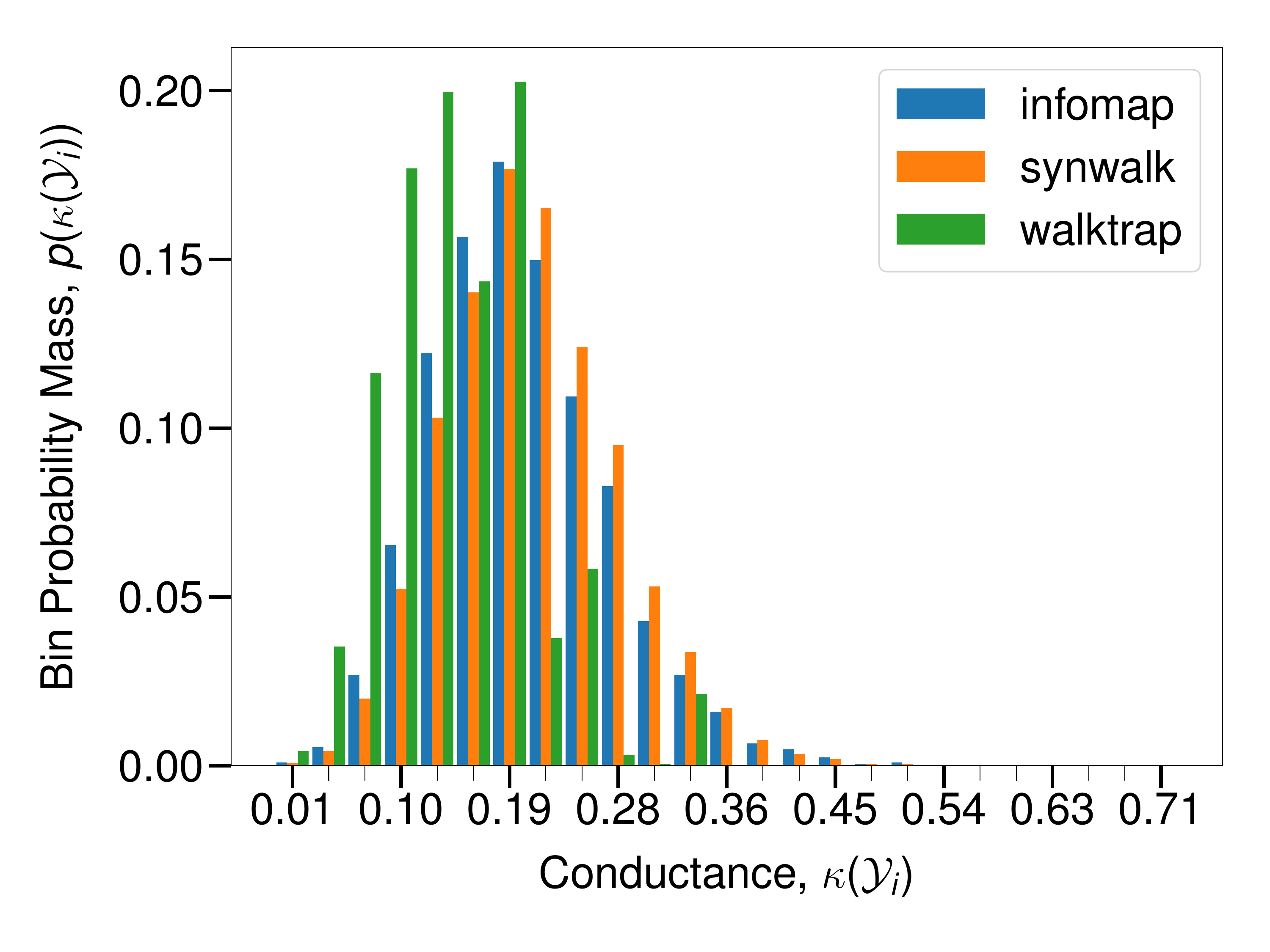}
        \caption{pennsylvania-roads.}
    \end{subfigure}
    \begin{subfigure}{0.32\textwidth}
        \centering
        \includegraphics[width=1.0\textwidth]{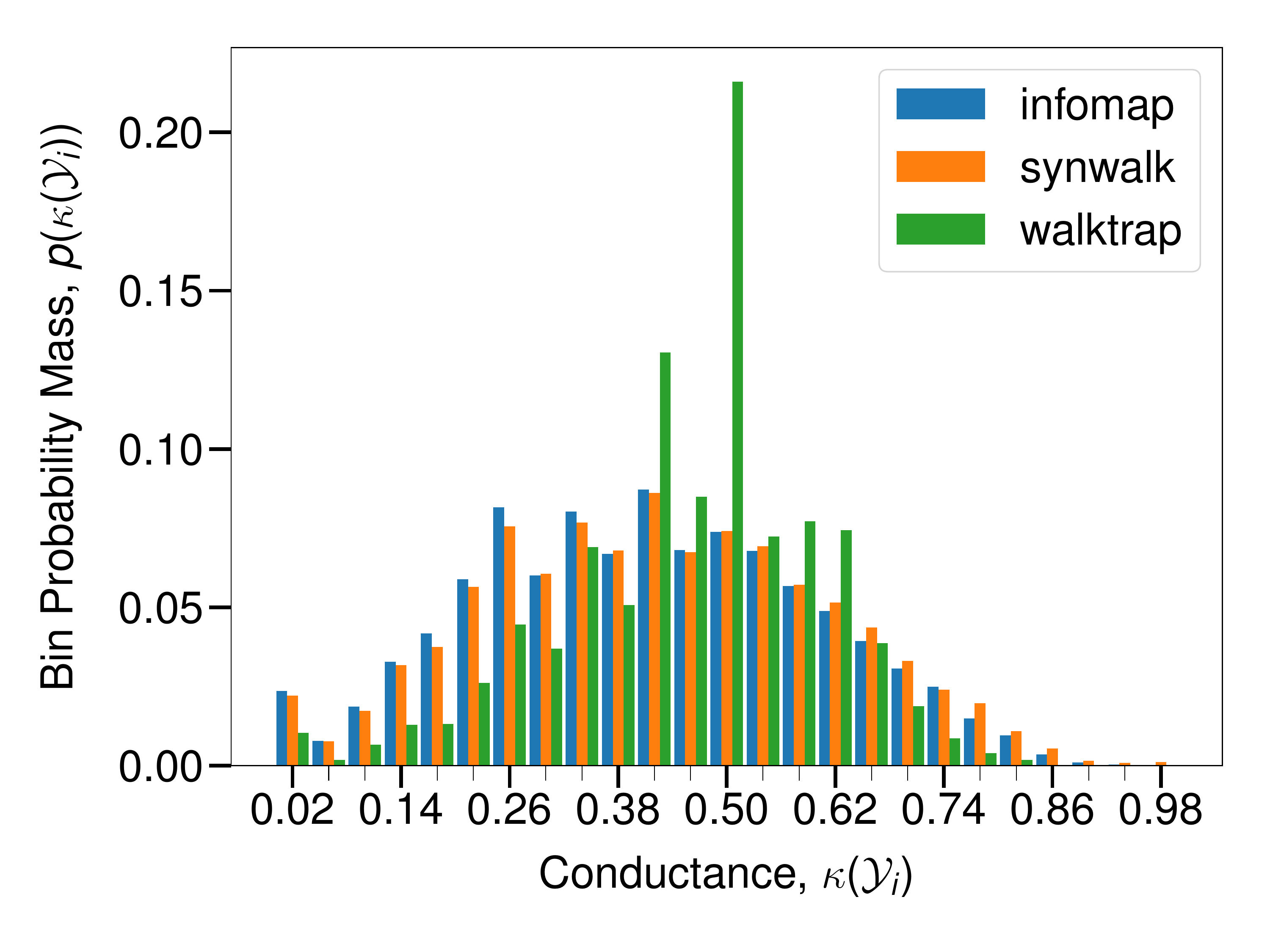}
        \caption{wordnet.}
    \end{subfigure}
    \caption{Distributions of cluster conductances for the detection results on empirical networks.}\label{fig:empirical_conductances}
\end{figure*}

\begin{figure*}[t]
    \centering
    \begin{subfigure}{0.32\textwidth}
        \centering
        \includegraphics[width=1.0\textwidth]{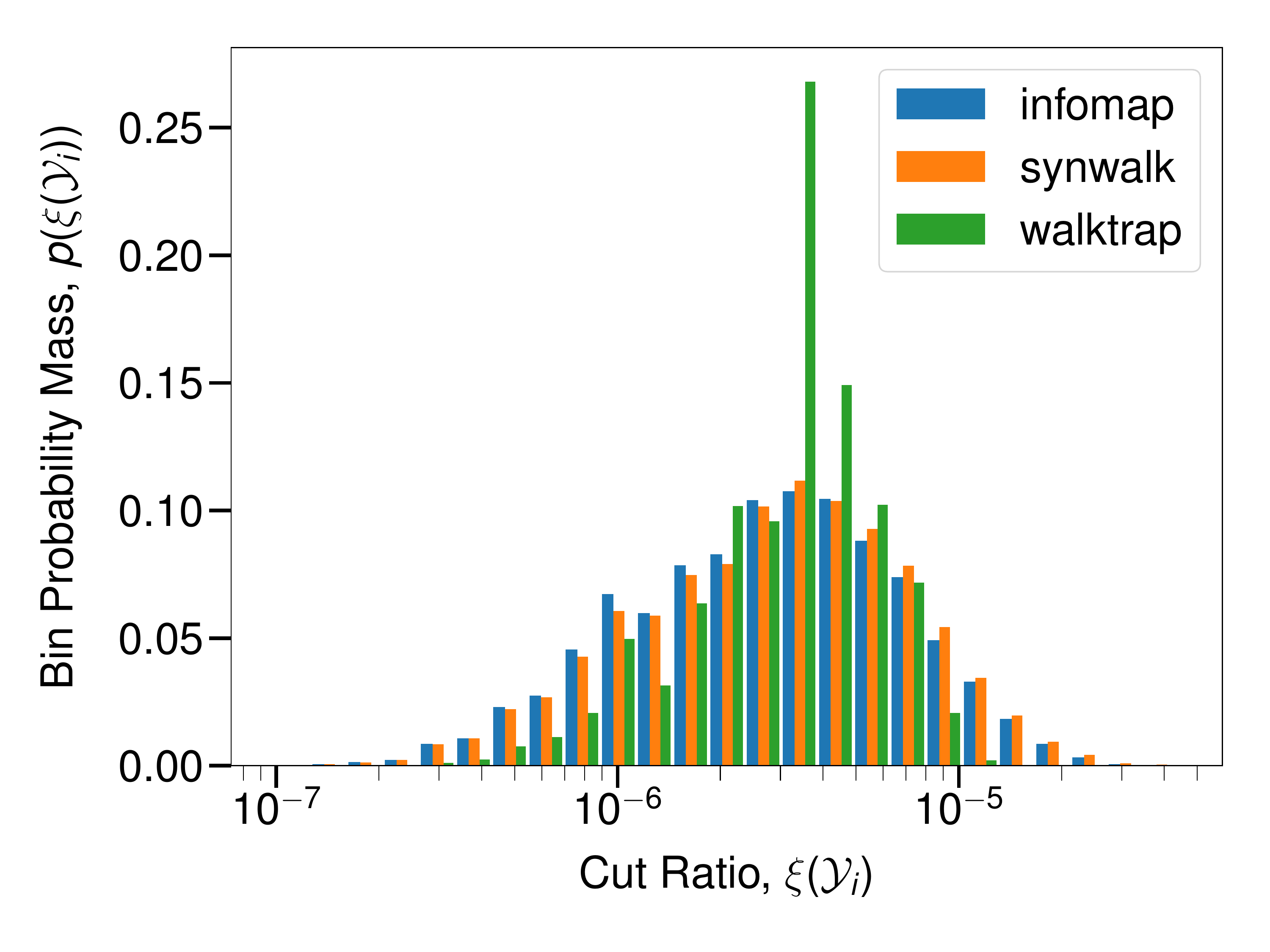}
        \caption{dblp.}
    \end{subfigure}
   \begin{subfigure}{0.32\textwidth}
        \centering
        \includegraphics[width=1.0\textwidth]{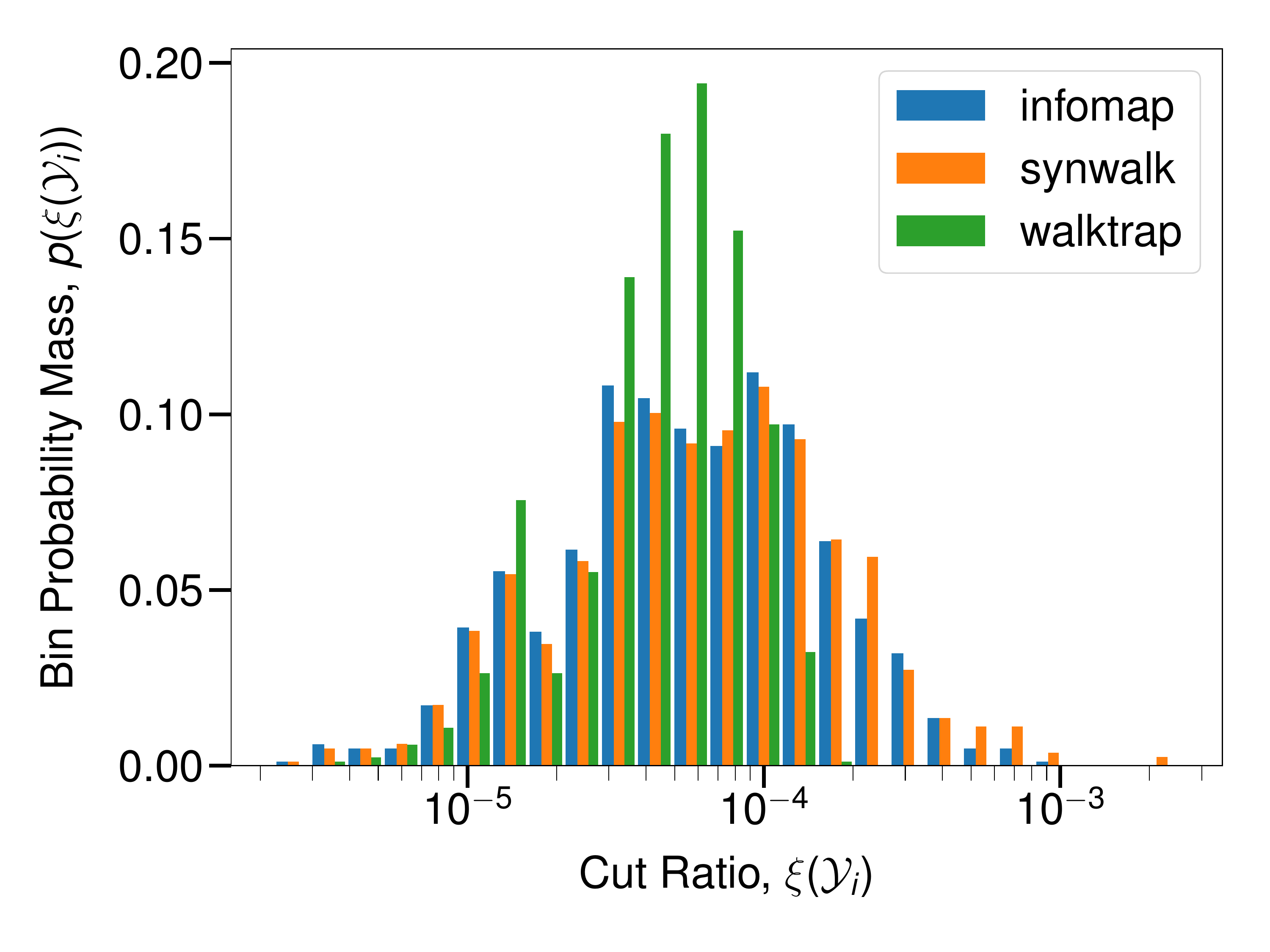}
        \caption{facebook.}
    \end{subfigure}
    \begin{subfigure}{0.32\textwidth}
        \centering
        \includegraphics[width=1.0\textwidth]{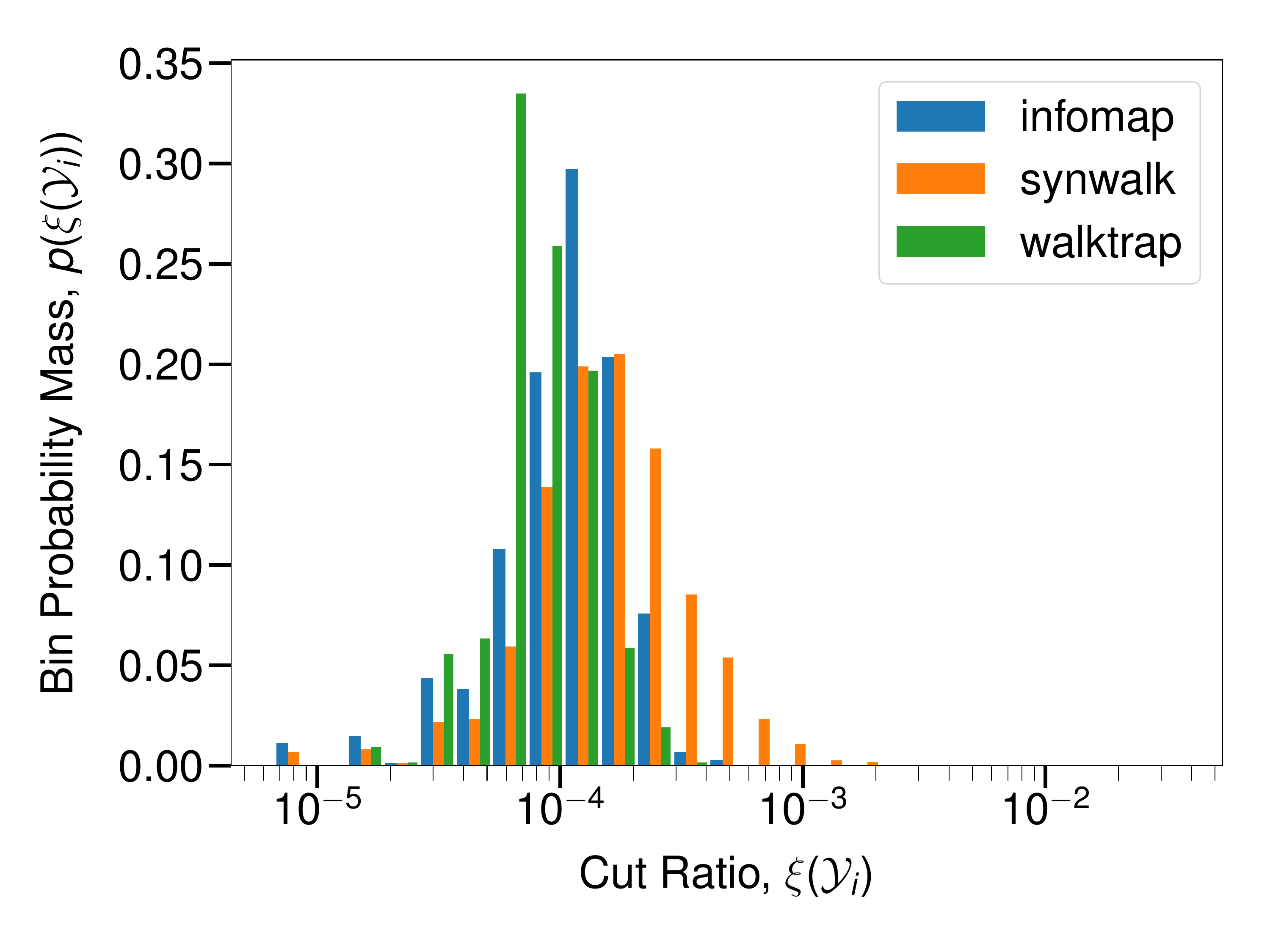}
        \caption{github.}
    \end{subfigure}
    \\
    \begin{subfigure}{0.32\textwidth}
        \centering
        \includegraphics[width=1.0\textwidth]{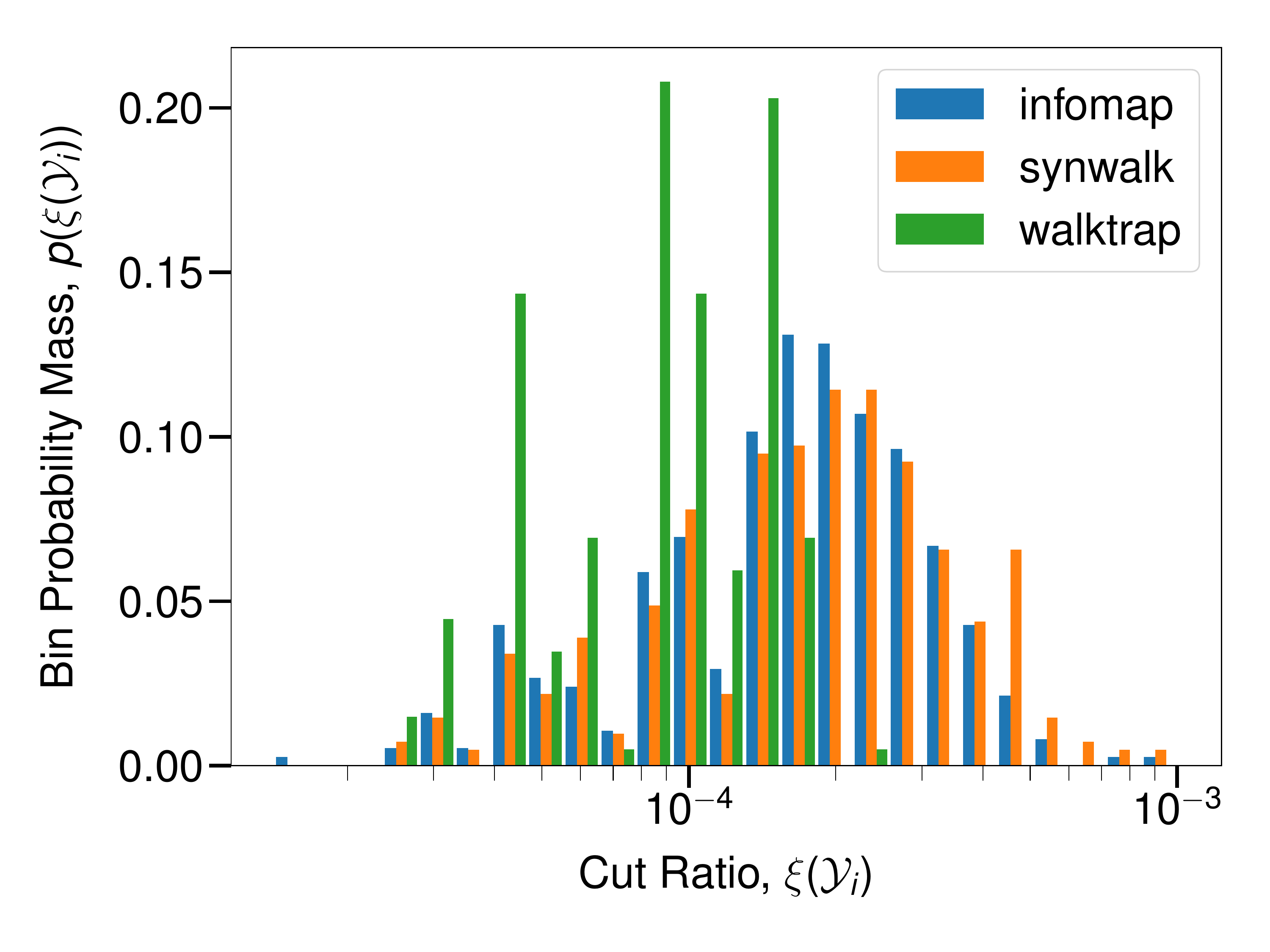}
        \caption{lastfm-asia.}
    \end{subfigure}
    \begin{subfigure}{0.32\textwidth}
        \centering
        \includegraphics[width=1.0\textwidth]{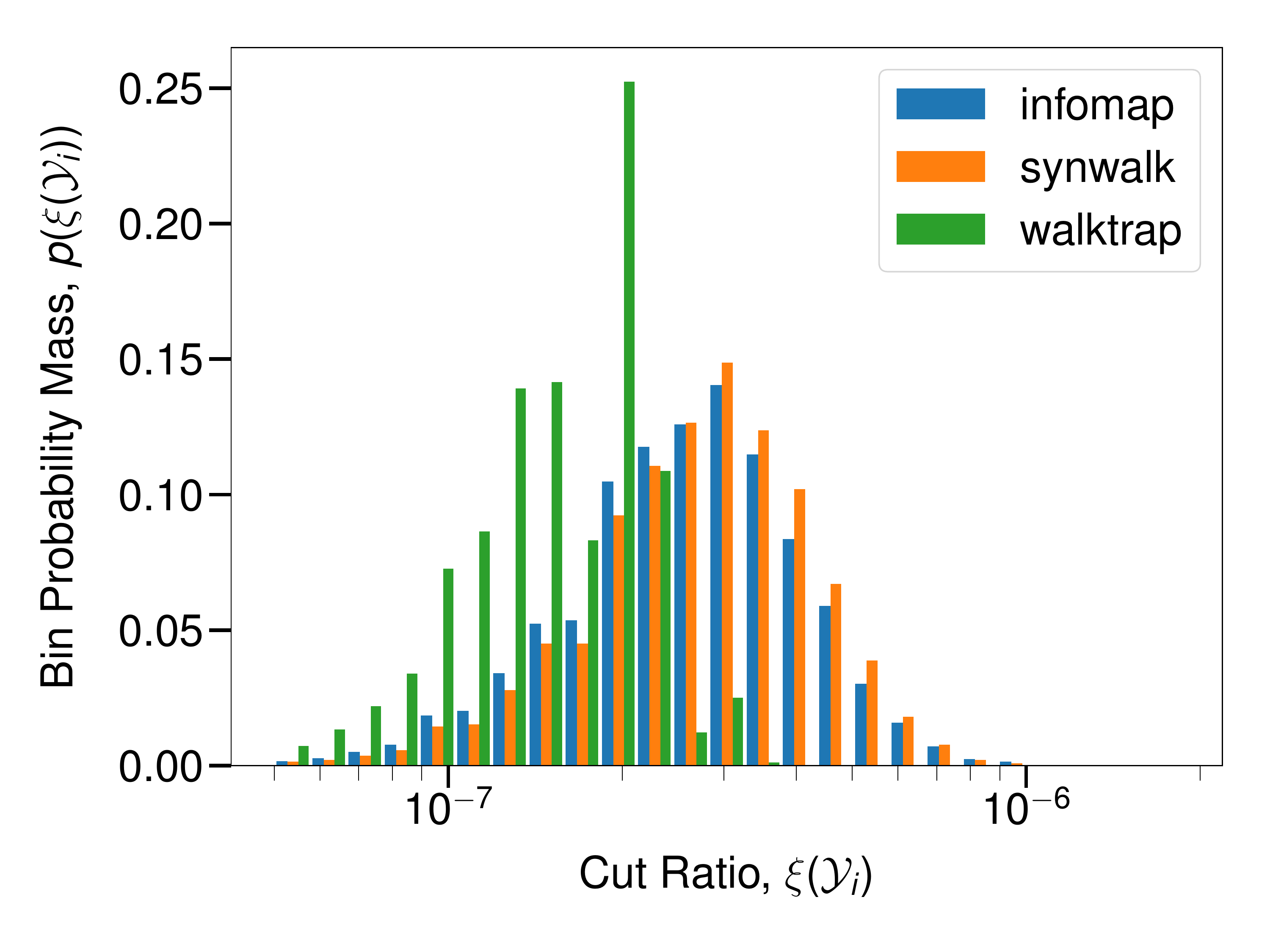}
        \caption{pennsylvania-roads.}
    \end{subfigure}
    \begin{subfigure}{0.32\textwidth}
        \centering
        \includegraphics[width=1.0\textwidth]{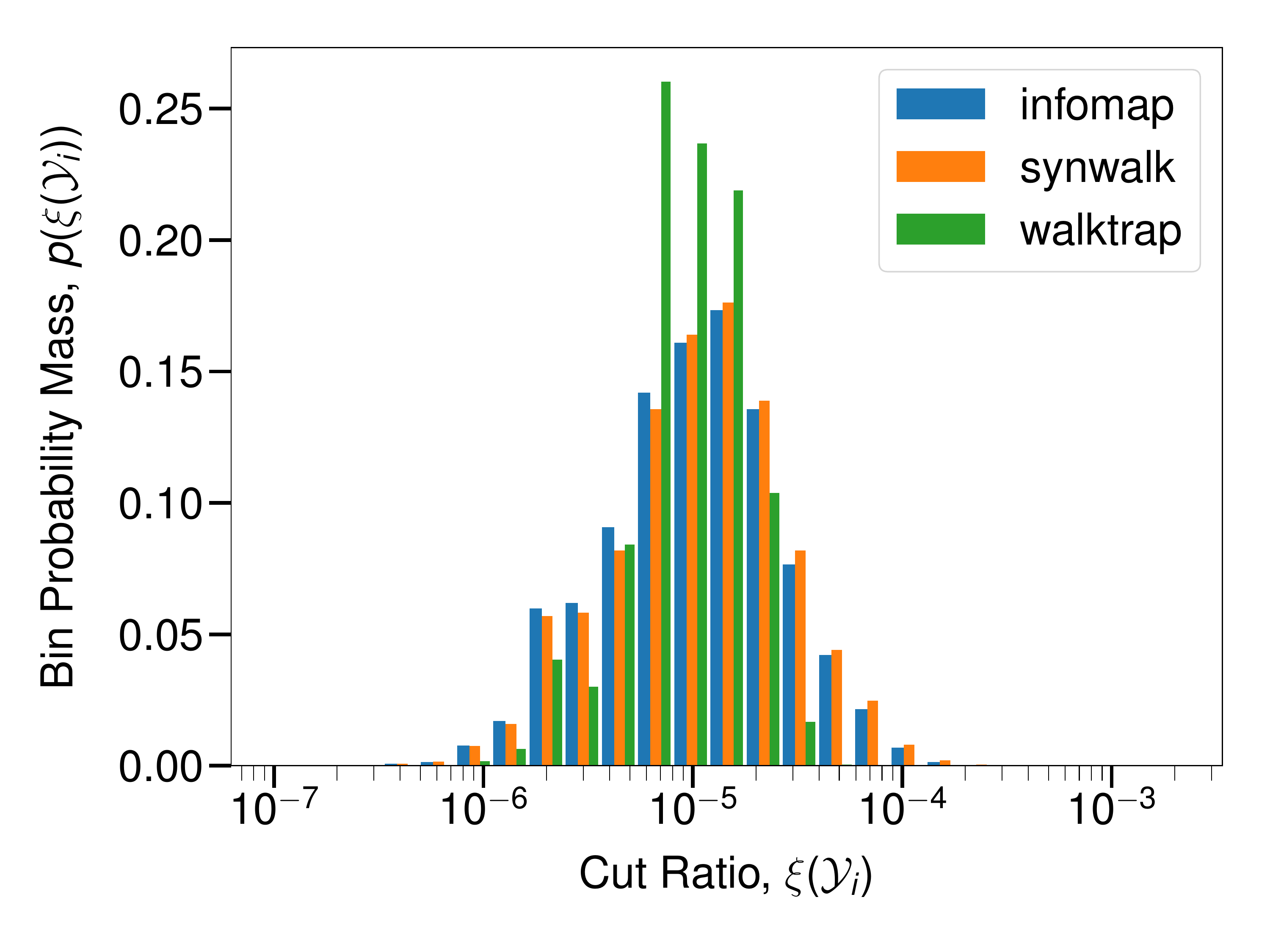}
        \caption{wordnet.}
    \end{subfigure}
    \caption{Distributions of cluster cut ratios for the detection results on empirical networks.}\label{fig:empirical_cut_ratios}
\end{figure*}

\begin{figure*}[t]
    \centering
    \begin{subfigure}{0.32\textwidth}
        \centering
        \includegraphics[width=1.0\textwidth]{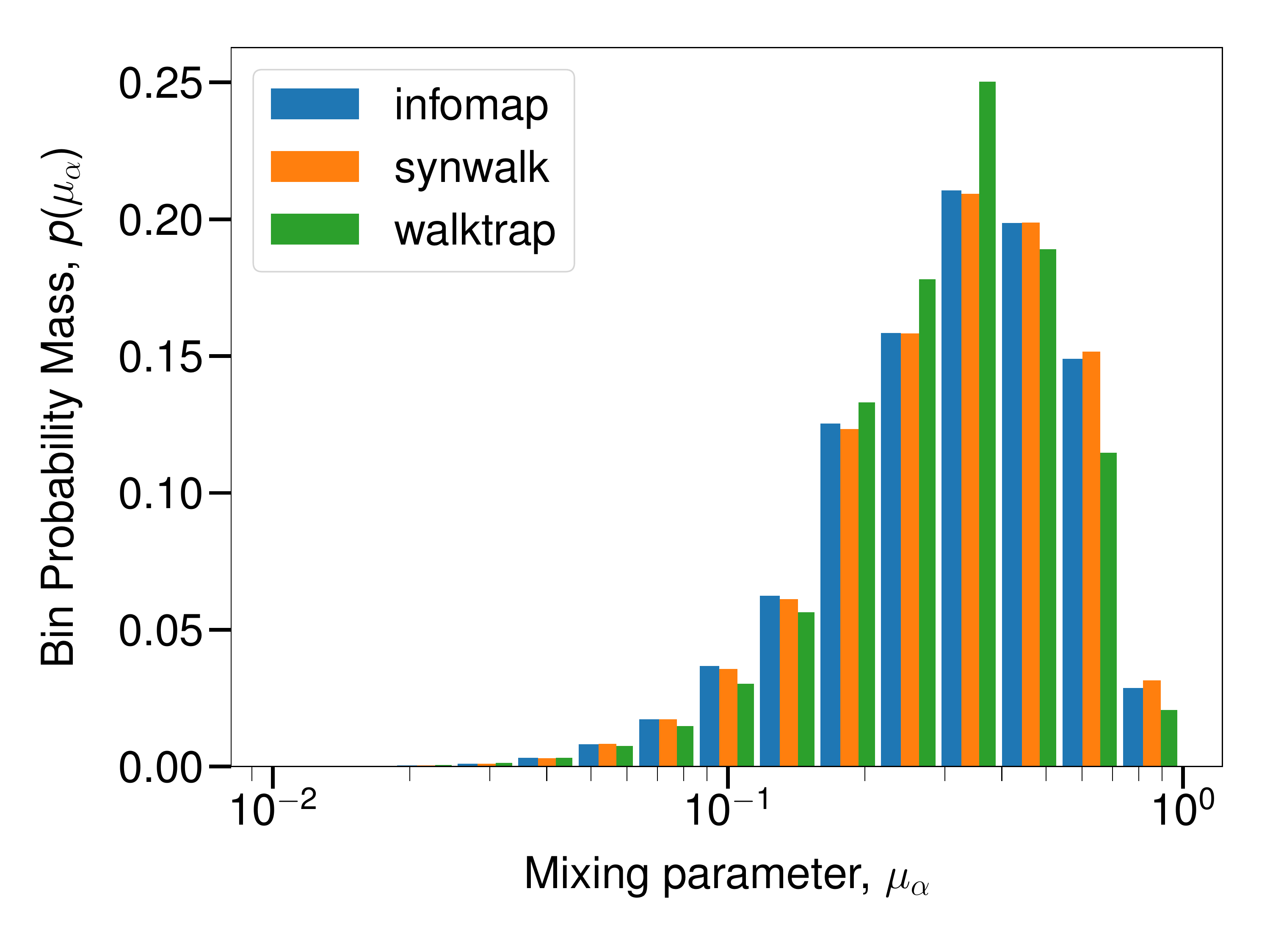}
        \caption{dblp.}
    \end{subfigure}
   \begin{subfigure}{0.32\textwidth}
        \centering
        \includegraphics[width=1.0\textwidth]{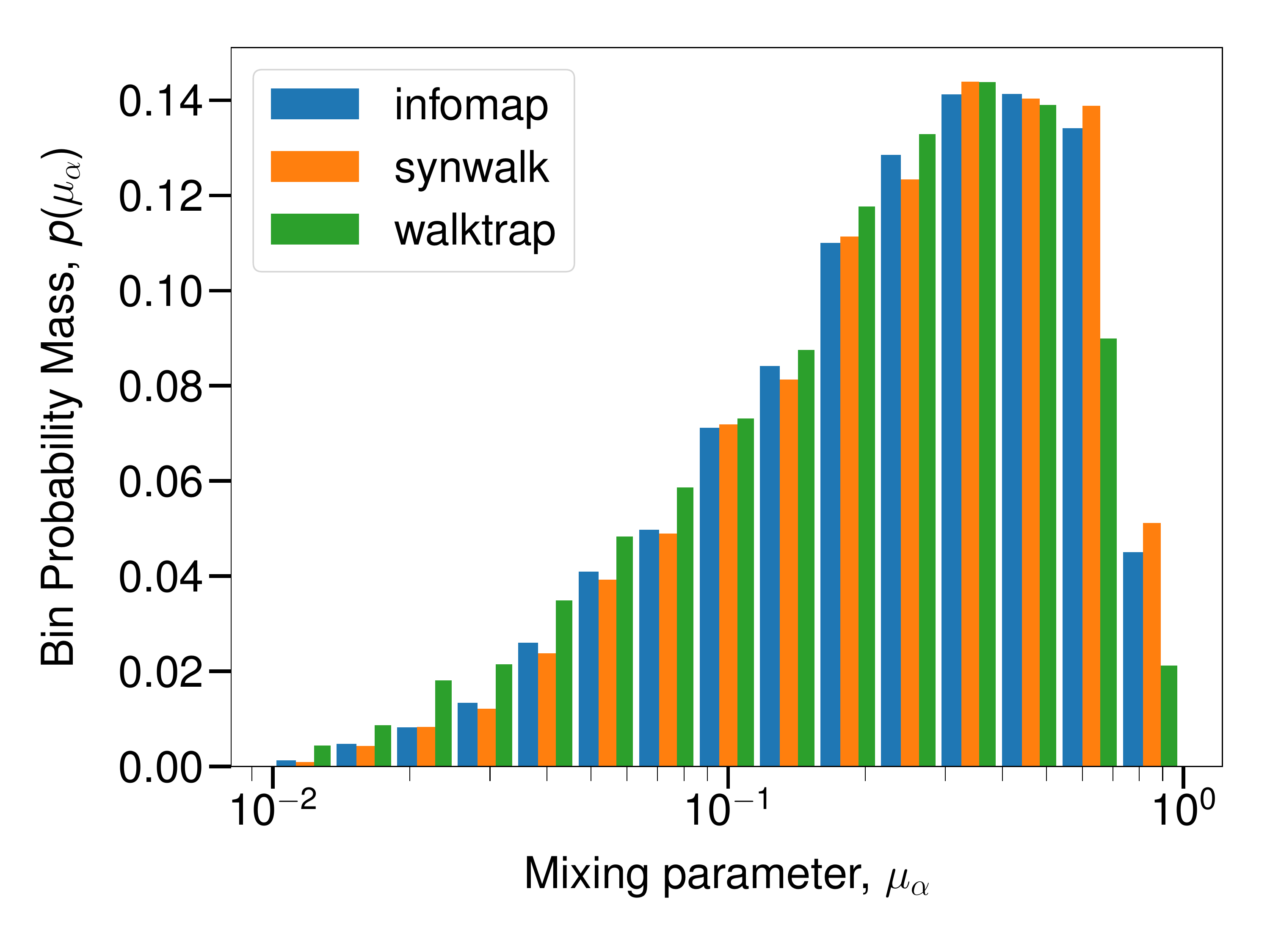}
        \caption{facebook.}
    \end{subfigure}
    \begin{subfigure}{0.32\textwidth}
        \centering
        \includegraphics[width=1.0\textwidth]{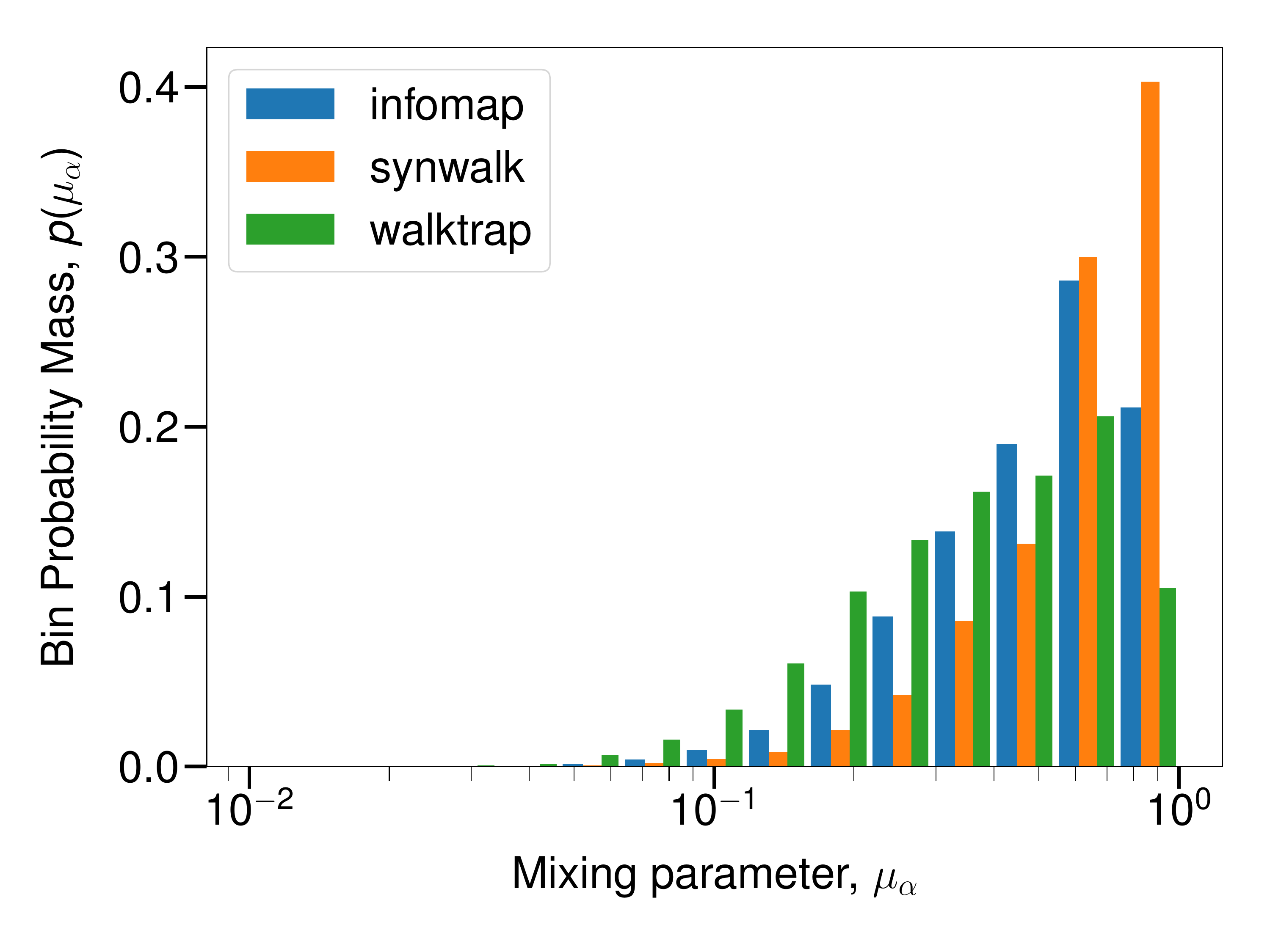}
        \caption{github.}
    \end{subfigure}
    \\
    \begin{subfigure}{0.32\textwidth}
        \centering
        \includegraphics[width=1.0\textwidth]{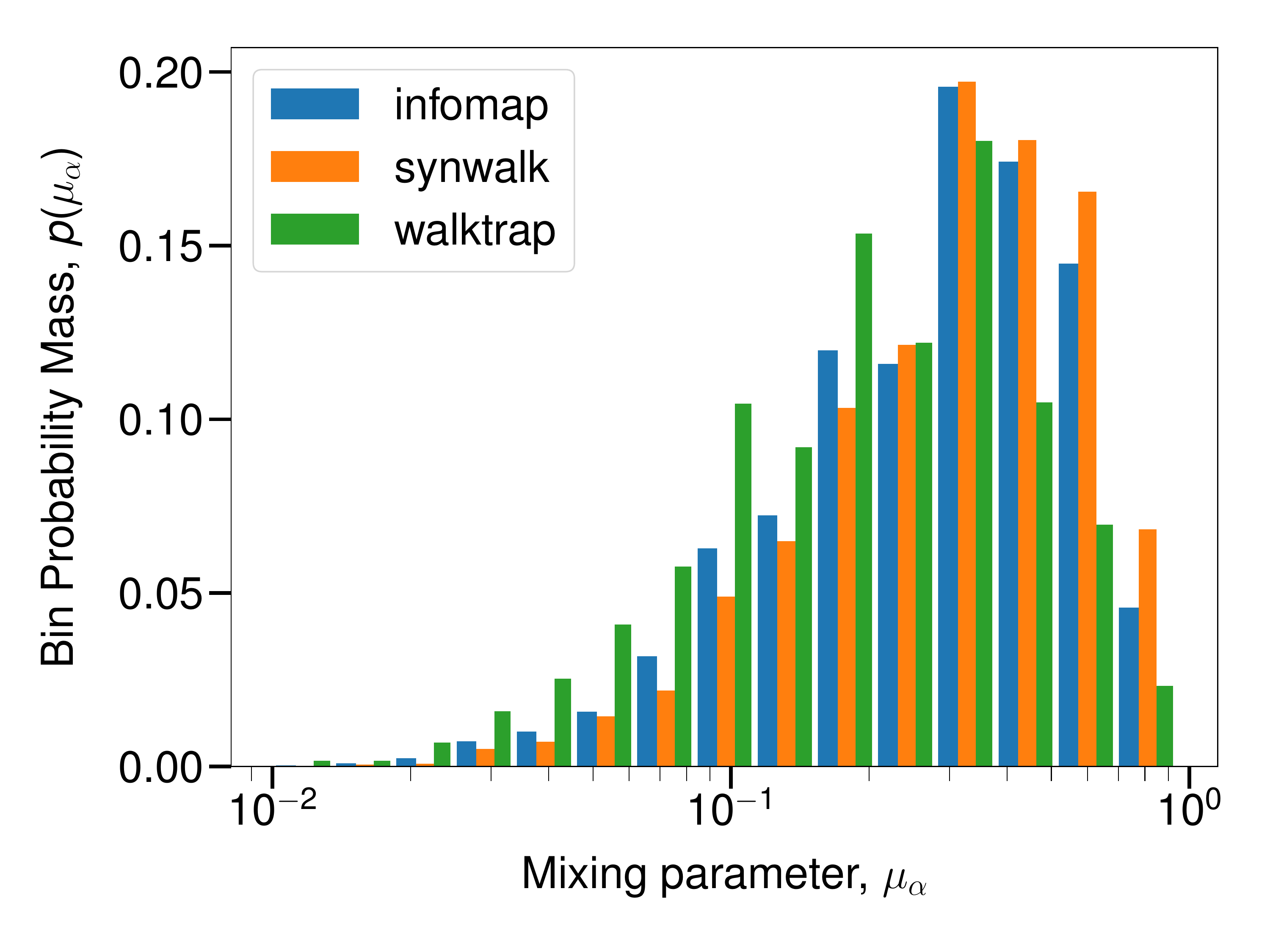}
        \caption{lastfm-asia.}
    \end{subfigure}
    \begin{subfigure}{0.32\textwidth}
        \centering
        \includegraphics[width=1.0\textwidth]{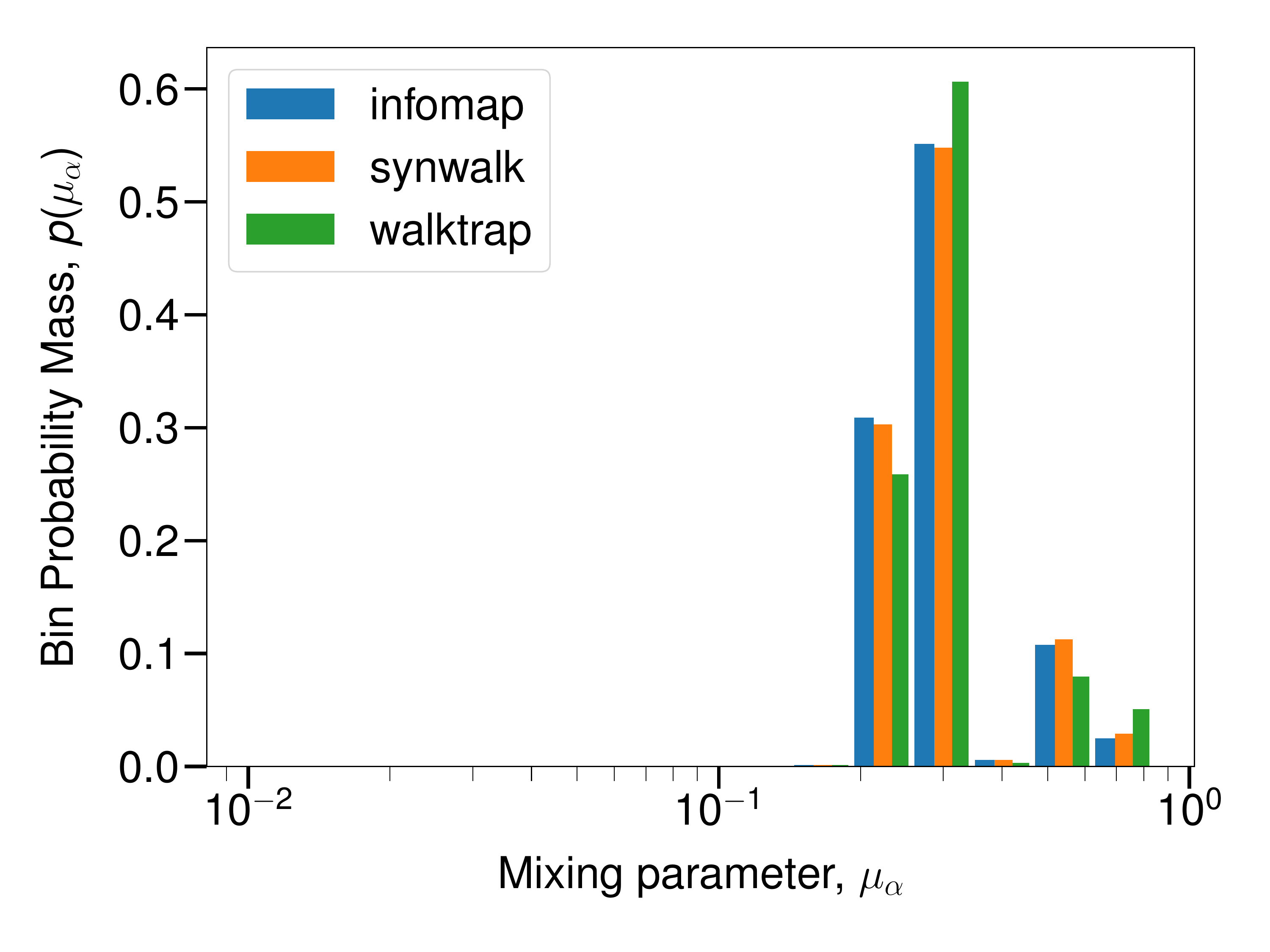}
        \caption{pennsylvania-roads.}
    \end{subfigure}
    \begin{subfigure}{0.32\textwidth}
        \centering
        \includegraphics[width=1.0\textwidth]{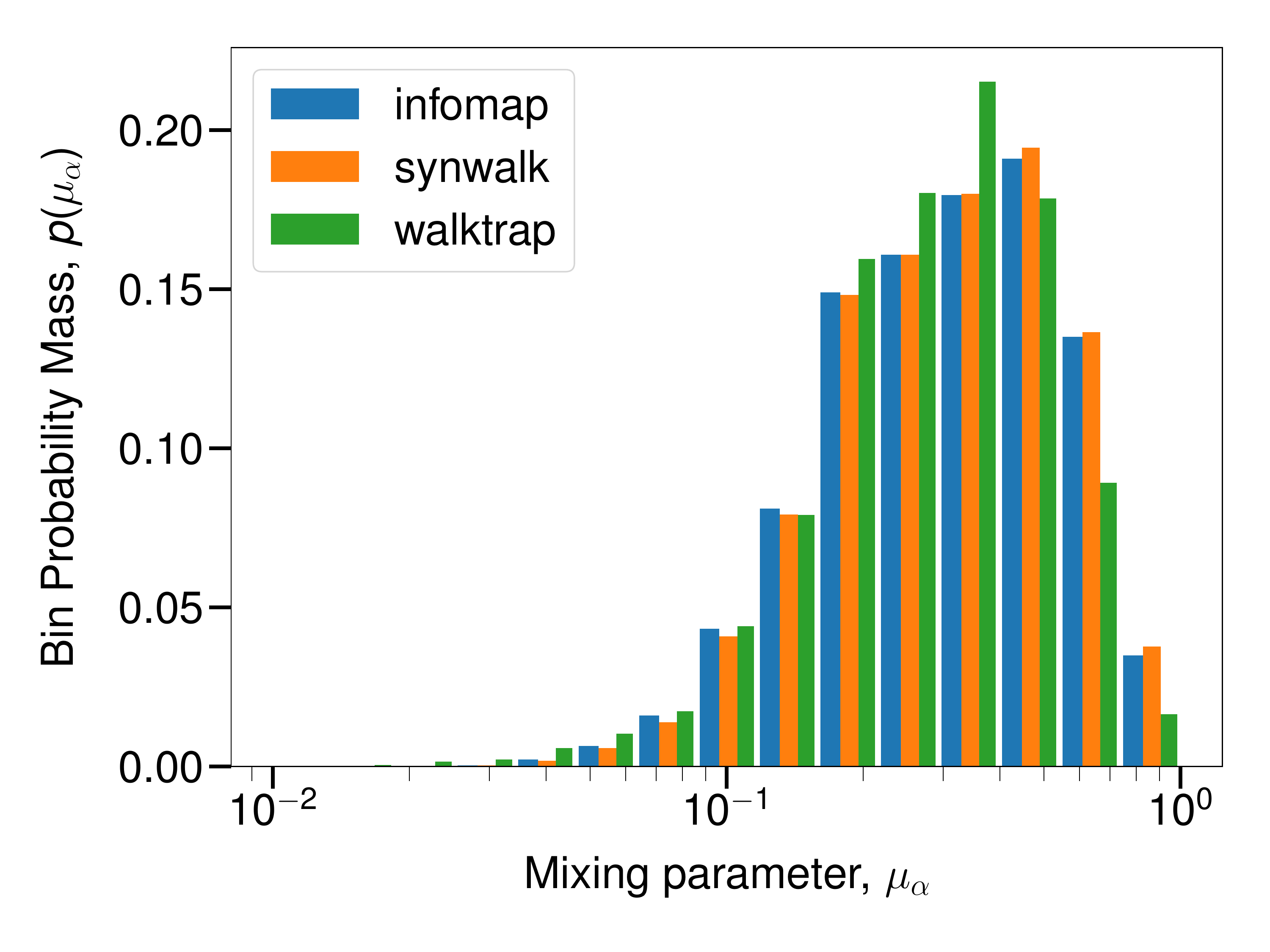}
        \caption{wordnet.}
    \end{subfigure}
    \caption{Distributions of node mixing parameters for the detection results on empirical networks.}\label{fig:empirical_mixing_parameters}
\end{figure*}


\end{document}